\documentclass[11pt]{article}
\usepackage{amsmath,amsthm,amsfonts,amssymb,graphicx,float,calc,microtype,thmtools,underscore,mathtools,paralist}
\usepackage[usenames,dvipsnames,svgnames,table]{xcolor}
\usepackage[unicode=true,pagebackref]{hyperref}
\usepackage[T1]{fontenc}
\hypersetup{
colorlinks,
linkcolor={black},
citecolor={black},
urlcolor={blue!60!black},
pdftitle={Planar Graphs have Bounded Queue-Number},
pdfauthor={Vida Dujmovi{\'c}, Gwena\"{e}l Joret, Piotr Micek, Pat Morin, Torsten Ueckerdt, David~R.~Wood}}
\usepackage[noabbrev,capitalise]{cleveref}
\usepackage[longnamesfirst,numbers,sort&compress]{natbib}
\usepackage{hypernat} % so that pagebackref and sort&compress work together
\makeatletter
\def\NAT@spacechar{~}
\makeatother
\renewcommand*{\backref}[1]{}
\renewcommand*{\backrefalt}[4]{\footnotesize\hspace*{0pt}\hfill \ifcase #1 \mbox{[not cited]} \or  \mbox{[p.\,#2]}  \else \mbox{[pp.\,#2]} \fi}
\usepackage[margin=27mm]{geometry}
\renewcommand{\baselinestretch}{1.025}
\allowdisplaybreaks

\DeclarePairedDelimiter{\floor}{\lfloor}{\rfloor}
\DeclarePairedDelimiter{\ceil}{\lceil}{\rceil}
\DeclareMathOperator{\qn}{qn}

\DeclareMathOperator{\sn}{sn}
\DeclareMathOperator{\tn}{tn}

\DeclareMathOperator{\join}{+}
\DeclareMathOperator{\dist}{dist}
\renewcommand{\ge}{\geqslant}
\renewcommand{\le}{\leqslant}
\renewcommand{\geq}{\geqslant}
\renewcommand{\leq}{\leqslant}
\renewcommand{\preceq}{\preccurlyeq}
\renewcommand{\thefootnote}{\fnsymbol{footnote}}
\allowdisplaybreaks
\newcommand{\arXiv}[1]{arXiv:\,\href{https://arxiv.org/abs/#1}{#1}}
\newcommand{\doi}[1]{\href{https://doi.org/#1}{\tt https://doi.org/#1}}
\newcommand{\msn}[1]{MR:\,\href{http://www.ams.org/mathscinet-getitem?mr=MR#1}{#1}}

\newcommand{\PP}{\mathcal{P}}

\theoremstyle{plain}
\newtheorem{theorem}{Theorem}
\newtheorem{lemma}[theorem]{Lemma}
\newtheorem{corollary}[theorem]{Corollary}
\newtheorem{prop}[theorem]{Proposition}
\newtheorem{obs}[theorem]{Observation}
\theoremstyle{definition}

\newtheorem{claim}{Claim}

\begin{document}

\author{Vida Dujmovi{\'c}\,\footnotemark[2]
\qquad Gwena\"{e}l Joret\,\footnotemark[3]
\qquad Piotr  Micek\,\footnotemark[4] \\
 Pat Morin\,\footnotemark[5]
\qquad Torsten Ueckerdt\,\footnotemark[6]
\qquad David~R.~Wood\,\footnotemark[9]}

\footnotetext[2]{School of Computer Science and Electrical Engineering, University of Ottawa, Ottawa, Canada (\texttt{vida.dujmovic@uottawa.ca}). Research  supported by NSERC and the Ontario Ministry of Research and Innovation.}

\footnotetext[3]{D\'epartement d'Informatique, Universit\'e Libre de Bruxelles, Brussels, Belgium (\texttt{gjoret@ulb.ac.be}). Research supported by an ARC grant from the Wallonia-Brussels Federation of Belgium.}

\footnotetext[4]{Theoretical Computer Science Department, Faculty of Mathematics and Computer Science, Jagiellonian University, Krak\'ow, Poland (\texttt{piotr.micek@tcs.uj.edu.pl}).
Research partially supported by the Polish National Science Center grant (SONATA BIS 5; UMO-2015/18/E/ST6/00299).}

\footnotetext[5]{School of Computer Science, Carleton University, Ottawa, Canada (\texttt{morin@scs.carleton.ca}). Research  supported by NSERC.}

\footnotetext[6]{Institute of Theoretical Informatics, Karlsruhe Institute of Technology, Germany (\texttt{torsten.ueckerdt@kit.edu}).}

\footnotetext[9]{School of Mathematics, Monash   University, Melbourne, Australia  (\texttt{david.wood@monash.edu}). Research supported by the Australian Research Council.}

\sloppy

\title{ \textbf{Planar Graphs have Bounded Queue-Number}\footnote{An extended abstract of this paper appeared in \emph{Proceedings 60th Annual Symposium on Foundations of Computer Science} (FOCS '19), pp.~862--875, IEEE. \doi{10.1109/FOCS.2019.00056}.}}

\date{April 9, 2019\\ revised: \today}

\maketitle

% \thanks{\textbf{MSC Classification}: ???}

\begin{abstract}
We show that planar graphs have bounded queue-number, thus proving a conjecture of Heath, Leighton and Rosenberg from 1992. The key to the proof is a new structural tool called layered partitions, and the result that every planar graph has a vertex-partition and a layering, such that each part has a bounded number of vertices in each layer, and the quotient graph has bounded treewidth. This result generalises for graphs of bounded Euler genus. Moreover, we prove that every graph in a minor-closed class has such a layered partition if and only if the class excludes some apex graph. Building on this work and using the graph minor structure theorem, we prove that every proper minor-closed class of graphs has bounded queue-number. 

Layered partitions have strong connections to other topics, including the following two examples. First, they can be interpreted in terms of strong products. We show that every planar graph is a subgraph of the strong product of a path with some graph of bounded treewidth. Similar statements hold for all proper minor-closed classes. Second, we give a simple proof of the result by DeVos et al.\ (2004) that graphs in a proper minor-closed class have low treewidth colourings. 
\end{abstract}

\renewcommand{\thefootnote}{\arabic{footnote}}

\newpage
\tableofcontents
\newpage

%%%%%%%%%%%%
\section{Introduction}
\label{Introduction}

Stacks and queues are fundamental data structures in computer science. But what is more powerful, a stack or a queue? In 1992, \citet{HLR92} developed a graph-theoretic formulation of this question, where they defined the graph parameters stack-number and queue-number which respectively measure the power of stacks and queues to represent a given graph. Intuitively speaking, if some class of graphs has bounded stack-number and unbounded queue-number, then we would consider stacks to be more powerful than queues for that class (and vice versa). It is known that the stack-number of a graph may be much larger than the queue-number. For example, \citet{HLR92} proved that the $n$-vertex ternary Hamming graph has queue-number at most $O(\log n)$ and stack-number at least $\Omega(n^{1/9-\epsilon})$. Nevertheless, it is open whether every graph has stack-number bounded by a function of its queue-number, or whether every graph has queue-number bounded by a function of its stack-number~\citep{HLR92,DujWoo05}.

Planar graphs are the simplest class of graphs where it is unknown whether both stack and queue-number are bounded. In particular, \citet{BS84} first proved that planar graphs have bounded stack-number; the best known upper bound is 4 due to \citet{Yannakakis89}. However,  for the last 27 years of research on this topic, the most important open question in this field has been whether planar graphs have bounded queue-number. This question was first proposed by \citet{HLR92} who conjectured that planar graphs have bounded queue-number.\footnote{Curiously, in a later paper, \citet{HR11} conjectured that planar graphs have unbounded queue-number.} This paper proves this conjecture. Moreover, we generalise this result for graphs of bounded Euler genus, and for every proper minor-closed class of graphs.\footnote{The \textit{Euler genus} of the orientable surface with $h$ handles is $2h$. The \textit{Euler genus} of the non-orientable surface with $c$ cross-caps is $c$. The \textit{Euler genus} of a graph $G$ is the minimum integer $k$ such that $G$ embeds in a surface of Euler genus $k$. Of course, a graph is planar if and only if it has Euler genus 0; see~\citep{MoharThom} for more about graph embeddings in surfaces. A graph $H$ is a \textit{minor} of a graph $G$ if a graph isomorphic to $H$ can be obtained from a subgraph of $G$ by contracting edges. A class $\mathcal{G}$ of graphs is \emph{minor-closed} if for every graph $G\in\mathcal{G}$, every minor of $G$ is in $\mathcal{G}$. 
A minor-closed class is \emph{proper} if it is not the class of all graphs. For example, for fixed $g\geq 0$, the class of graphs with Euler genus at most $g$ is a proper minor-closed class.}  

First we define the stack-number and queue-number of a graph $G$.  Let $V(G)$ and $E(G)$ respectively denote the vertex and edge set of $G$. Consider disjoint edges $vw,xy\in E(G)$ and a linear ordering $\preceq$ of $V(G)$. Without loss of generality, $v\prec w$ and $x\prec y$ and $v\prec x$. Then $vw$ and $xy$ are said to \emph{cross} if  $v\prec x\prec w \prec y$ and are said to \emph{nest} if $v\prec x\prec y \prec w$. A \emph{stack} (with respect to $\preceq$) is a set of pairwise non-crossing edges, and a \emph{queue} (with respect to $\preceq$) is a set of pairwise non-nested edges. Stacks resemble the stack data structure in the following sense. In a stack, traverse the vertex ordering left-to-right. When visiting vertex $v$, because of the non-crossing property, if $x_1,\dots,x_d$ are the neighbours of $v$ to the left of $v$ in left-to-right order, then the edges $x_dv,x_{d-1}v,\dots,x_1v$ will be on top of the stack in this order. Pop these edges off the stack. Then if $y_1,\dots,y_{d'}$ are the neighbours of $v$ to the right of $v$ in left-to-right order, then push  $vy_{d'},vy_{d'-1},\dots,vy_1$ onto the stack in this order. In this way, a stack of edges with respect to a linear ordering resembles a stack data structure. Analogously, the non-nesting condition in the definition of a queue implies that a queue of edges with respect to a linear ordering resembles a queue data structure.

For an integer $k\geq 0$, a \textit{$k$-stack layout} of a graph $G$ consists of a linear ordering $\preceq$ of $V(G)$ and a partition $E_1,E_2,\dots,E_k$ of $E(G)$ into stacks with respect to $\preceq$. Similarly, a \textit{$k$-queue layout} of $G$ consists of a linear ordering $\preceq$ of $V(G)$ and a partition $E_1,E_2,\dots,E_k$ of $E(G)$ into queues with respect to $\preceq$. The \textit{stack-number} of $G$, denoted by $\sn(G)$, is the minimum integer $k$ such that $G$ has a $k$-stack layout. The \textit{queue-number} of a graph $G$, denoted by $\qn(G)$, is the minimum integer $k$ such that $G$ has a $k$-queue layout. Note that $k$-stack layouts are equivalent to $k$-page book embeddings, first introduced by \citet{Ollmann73}, and stack-number is also called page-number, book thickness, or fixed outer-thickness.

Stack and queue layouts are inherently related to depth-first search and breadth-first search respectively. For example, a DFS ordering of the vertices of a tree has no two crossing edges, and thus defines a 1-stack layout. Similarly, a BFS ordering of the vertices of a tree has no two nested edges, and thus defines a 1-queue layout. So every tree has stack-number 1 and queue-number 1. 

For another example, consider the $n\times n$ grid graph with vertex set $\{(x,y):x,y\in[n]\}$ and edges of the form $(x,y)(x+1,y)$ and $(x,y)(x,y+1)$. Order the vertices first by $x$-coordinate and then by $y$-coordinate. Edges of the first type do not nest and edges of the second type do not nest. Thus the $n\times n$ grid graph has a 2-queue layout. In fact, as illustrated in \cref{QueueLayoutGrid}, if we order the vertices by $x+y$ and then by $x$-coordinate, then no two edges nest. So the $n\times n$ grid graph has queue-number 1.

\begin{figure}[!h]
\includegraphics[width=\textwidth]{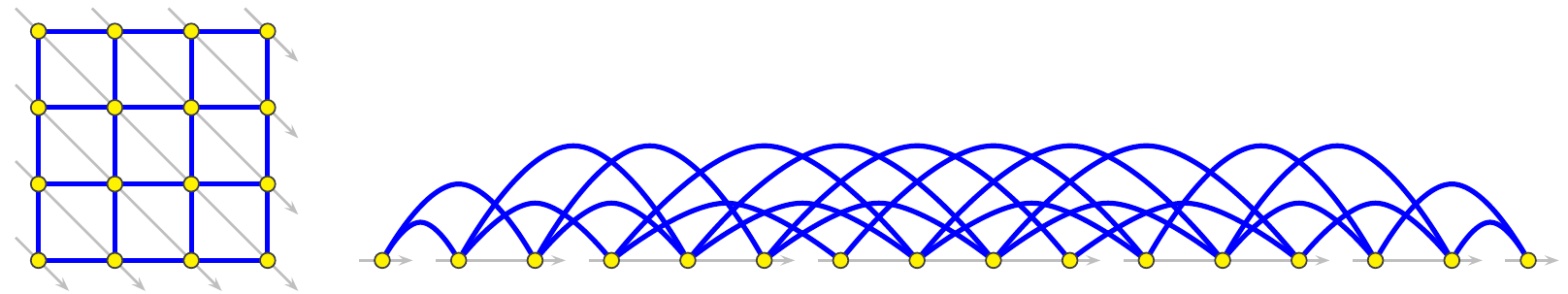}
\caption{1-Queue layout of grid graph.\label{QueueLayoutGrid}}
\end{figure}

As mentioned above, \citet{HLR92} conjectured that  planar graphs have bounded queue-number. This conjecture has remained open despite much research on queue layouts~\citep{Wiechert17,DM-GD03,DujWoo05,HLR92,HR92,Hasunuma-DAM,Pemmaraju-PhD,RM-COCOON95,DMW05,DujWoo04,DPW04,DMW17,ABGKP18,ABGKP20,BFGMMRU19,DFP13}. We now review progress on this conjecture. 

\citet{Pemmaraju-PhD} studied queue layouts and wrote that he ``suspects'' that a particular planar graph with $n$ vertices has queue-number $\Theta(\log n)$. The example he proposed had treewidth~3; see \cref{Treewidth} for the definition of treewidth. \citet{DMW05} proved that graphs of bounded treewidth have bounded queue-number. So Pemmaraju's example in fact has bounded queue-number.

The first $o(n)$ bound on the queue-number of planar graphs with $n$ vertices was proved by \citet{HLR92}, who observed that every graph with $m$ edges has a $O(\sqrt{m})$-queue layout using a random vertex ordering. Thus every planar graph with $n$ vertices has queue-number $O(\sqrt{n})$, which can also be proved using the Lipton-Tarjan separator theorem. \citet{DFP13} proved the first breakthrough on this topic, by showing that every planar graph with $n$ vertices has queue-number $O(\log^2 n)$. \citet{Duj15} improved this bound to $O(\log n)$ with a simpler proof. Building on this work, \citet{DMW17} established (poly-)logarithmic bounds for more general classes of graphs. For example, they proved that every graph with $n$ vertices and Euler genus $g$ has queue-number $O(g+\log n)$, and that every graph with $n$ vertices excluding a fixed minor has queue-number $\log^{O(1)} n$.

Recently, \citet{BFGMMRU19} proved a second breakthrough result, by showing that planar graphs with bounded maximum degree have bounded queue-number. In particular, every planar graph with maximum degree $\Delta$ has queue-number at most $O(\Delta^6)$. Subsequently, \citet{DMW19} proved that the algorithm of \citet{BFGMMRU19} in fact produces a $O(\Delta^2)$-queue layout. This was the state of the art prior to the current work.\footnote{\citet{Wang17} claimed to prove that planar graphs have bounded queue-number, but despite several attempts, we have not been able to understand the claimed proof.}

\subsection{Main Results}

The fundamental contribution of this paper is to prove the conjecture of \citet{HLR92} that planar graphs have bounded queue-number.

\begin{theorem}
\label{PlanarQueue}
The queue-number of planar graphs is bounded.
\end{theorem}

The best upper bound that we obtain for the queue-number of  planar graphs is $49$. 

We extend \cref{PlanarQueue} by showing that graphs with bounded Euler genus have bounded queue-number.

\begin{theorem}
\label{GenusQueue}
Every graph with Euler genus $g$ has queue-number at most  $O(g)$.
\end{theorem}

The best upper bound that we obtain for the queue-number of  graphs with Euler genus $g$ is $4g+49$. 

We generalise further to show the following:

\begin{theorem}
\label{MinorQueue}
Every proper minor-closed class of graphs has bounded queue-number.
\end{theorem}

These results are obtained through the introduction of a new tool, \emph{layered partitions}, that have applications well beyond queue layouts. Loosely speaking, a layered partition of a graph $G$ consists of a partition $\mathcal{P}$ of $V(G)$ along with a layering of $G$, such that each part in $\mathcal{P}$ has a bounded number of vertices in each layer (called the \emph{layered width}), and the quotient graph $G / \PP$ has certain desirable properties, typically bounded treewidth. Layered partitions are the key tool for proving the above theorems. Subsequent to the initial release of this paper, layered partitions and the results in this paper have been used to solve the following well-known problems:
\begin{compactitem}
\item  \citet{DEJWW20} prove that planar graphs have bounded non-repetitive chromatic number (resolving a conjecture of \citet{AGHR02} from 2002). This result generalises for graphs excluding any fixed graph as a subdivision. 
\item \citet{DFMS20} make dramatic improvements to the best known bounds for $p$-centered colourings of planar graphs and graphs excluding any fixed graph as a subdivision.
\item \citet{BGP20} find shorter adjacency labellings of planar graphs (improving on a sequence of results going back to 1988~\citep{KNR88,KNR92}). 
\item \citet{DEJGMM} find asymptotically optimal adjacency labellings of planar graphs. This result implies that, for every integer $n>0$, there is a graph with $n^{1+o(1)}$ vertices that contains every $n$-vertex planar graph as an induced subgraph.
\end{compactitem}

%layered partitions were used by \citet{DEJWW20} to prove that planar graphs have bounded nonrepetitive chromatic number, thus solving a well-known open problem of \citet{AGHR02}. As above, this result generalises for any proper minor-closed class.

% \cite{BGP20,DFMS20,DEJWW20,DEJGMM}. For example, our results for 
 
\subsection{Outline}

The remainder of the paper is organized as follows. In \cref{tools} we review relevant background including treewidth, layerings, and partitions, and we introduce layered partitions.

\cref{QLLP} proves a fundamental lemma which shows that every graph that has a partition of bounded layered width has queue-number bounded by a function of the queue-number of the quotient graph. 

In \cref{planar}, we prove that every planar graph has a partition of layered width 1 such that the quotient graph has treewidth at most $8$. Since graphs of bounded treewidth are known to have bounded queue-number \cite{DMW05}, this implies \cref{PlanarQueue} with an  upper bound of  $766$. We then prove a variant of this result with layered width 3, where the quotient graph is planar with treewidth 3. This variant coupled with a better bound on the queue-number of treewidth-$3$ planar graphs~\cite{ABGKP18,ABGKP20} implies  \cref{PlanarQueue} with an upper bound of $49$. 

In \cref{genus}, we prove that graphs of Euler genus $g$ have partitions of layered width $O(g)$ such that  the quotient graph has treewidth $O(1)$. This immediately implies that such graphs have queue-number $O(g)$. These partitions are also required for the proof of \cref{MinorQueue} in \cref{minor}.  A more direct argument that appeals to \cref{PlanarQueue} proves the bound $4g+49$ in \cref{GenusQueue}.

In \cref{minor}, we extend  our results for layered partitions to the setting of almost-embeddable graphs with no apex vertices. Coupled with other techniques, this allows us to prove \cref{MinorQueue}. We also characterise those minor-closed graph classes with the property that every graph in the class has a partition of bounded layered width such that the quotient has bounded treewidth. 

In \cref{Products}, we provide an alternative and helpful perspective on layered partitions in terms of strong products of graphs. With this viewpoint, we derive results about universal graphs that contain all planar graphs. Similar results are obtained for more general classes. 

In \cref{non-minor-closed}, we prove that some well-known non-minor-closed classes of graphs, such as $k$-planar graphs, also have bounded queue-number.

\cref{other} explores further applications and connections. We start off by giving an example where layered partitions lead to a simple proof of a known and difficult result about low treewidth colourings in proper minor-closed classes. Then we point out some of the many connections that layered partitions have with other graph parameters. We also present other implications of our results such as resolving open problems on 3-dimensional graph drawing.

Finally \cref{conclusion} summarizes and concludes with open problems and directions for future work.

%%%%%%%%%%%%%%%%%
\section{Tools}
\label{tools}

Undefined terms and notation can be found in Diestel's text~\citep{Diestel5}. Throughout the paper, we use the notation $\overrightarrow{X}$ to refer to a particular linear ordering of a set $X$. 

%%%%%%%%%%%%%
\subsection{Layerings}
\label{Layerings}

The following well-known definitions are key concepts in our proofs, and that of several other papers on queue layouts~\citep{DMW17,DMW05,DujWoo04,BFGMMRU19,DMW19}. A \emph{layering} of a graph $G$ is an ordered partition $(V_0,V_1,\dots)$ of $V(G)$ such that for every edge $vw\in E(G)$, if $v\in V_i$ and $w\in V_j$, then $|i-j| \leq 1$. If $i=j$ then $vw$ is an \emph{intra-level} edge. If $|i-j|=1$ then $vw$ is an \emph{inter-level} edge. 

If $r$ is a vertex in a connected graph $G$ and $V_i:=\{v\in V(G):\dist_G(r,v)=i\}$ for all $i\geq 0$, then $(V_0,V_1,\dots)$ is called a \textit{BFS layering} of $G$ \emph{rooted} at $r$. Associated with a BFS layering is a \emph{BFS spanning tree} $T$ obtained by choosing, for each non-root vertex $v\in V_i$ with $i\geq 1$, a neighbour $w$ in $V_{i-1}$, and adding the edge $vw$ to $T$. Thus $\dist_T(r,v)=\dist_G(r,v)$ for each vertex $v$ of $G$.

These notions extend to disconnected graphs. If $G_1,\dots,G_c$ are the components of $G$, and $r_j$ is a vertex in $G_j$ for each $j\in\{1,\dots,c\}$, and $V_i:=\bigcup_{j=1}^c \{v\in V(G_j):\dist_{G_j}(r_j,v)=i\}$ for all $i\geq 0$, then $(V_0,V_1,\dots)$ is called a \textit{BFS layering} of $G$. 

%%%%%%%%%%%%%
\subsection{Treewidth and Layered Treewidth}
\label{Treewidth}

First we introduce the notion of $H$-decomposition and tree-decomposition.
For graphs $H$ and $G$, an \emph{$H$-decomposition} of $G$ consists of a collection $(B_x\subseteq V(G) : x\in V(H))$ of subsets of $V(G)$, called \emph{bags}, indexed by the vertices of $H$, and with the following properties:
\begin{compactitem}
\item for every vertex $v$ of $G$, the set $\{x\in V(H) : v\in B_x\}$ induces a non-empty connected subgraph of $H$, and
\item for every edge $vw$ of $G$, there is a vertex $x\in V(H)$ for which $v,w\in B_x$.
\end{compactitem}
The \emph{width} of such an $H$-decomposition is $\max\{|B_x|:x\in V(H)\}-1$. 
The elements of $V(H)$ are called \emph{nodes}, while the elements of $V(G)$ are called \emph{vertices}. 

A \emph{tree-decomposition} is a $T$-decomposition for some tree $T$.
The \emph{treewidth} of a graph $G$ is the minimum width of a tree-decomposition of $G$.
Treewidth measures how similar a given graph is to a tree. 
It is particularly important in structural and algorithmic graph theory; see~\citep{HW17,Reed97,Bodlaender-TCS98} for surveys.
Tree decompositions were introduced by \citet{RS-II}; the more general notion of $H$-decomposition was introduced by \citet{DK05}. 

As mentioned in \cref{Introduction}, \citet{DMW05} first proved that graphs of bounded treewidth have bounded queue-number. Their bound on the queue-number was doubly exponential in the treewidth. \citet{Wiechert17} improved this bound to singly exponential.

\begin{lemma}[\citep{Wiechert17}]
\label{TreewidthQueue}
Every graph with treewidth $k$ has queue-number at most $2^k-1$.
\end{lemma}

\citet*{ABGKP18} also improved the bound in the case of planar 3-trees. (A \emph{$k$-tree} is an edge-maximal graph of tree-width $k$.)\ The following lemma that will be useful later is implied by this result and the fact that every planar graph of treewidth at most $3$ is a subgraph of a planar $3$-tree~\citep{KV12}. 

\begin{lemma}[\citep{ABGKP18,KV12}]
\label{PlanarTreewidth3Queue5}
Every planar graph with treewidth at most $3$ has queue-number at most $5$.
\end{lemma}

Graphs with bounded treewidth provide important examples of minor-closed classes. However, planar graphs have unbounded treewidth. For example, the $n\times n$ planar grid graph has treewidth $n$. So the above results do not resolve the question of whether planar graphs have bounded queue-number.

%\comment{Add a reference for  the $n\times n$ planar grid graph has treewidth $n$}

\citet{DMW17} and \citet{Shahrokhi13} independently introduced the following concept. The \emph{layered treewidth} of a graph $G$ is the minimum integer $k$ such that $G$ has a tree-decomposition $(B_x:x\in V(T))$ and a layering $(V_0,V_1,\dots)$ such that $|B_x\cap V_i|\leq k$ for every bag $B_x$ and layer $V_i$. Applications of layered treewidth include graph colouring~\citep{DMW17,LW1,vdHW18}, graph drawing~\citep{DMW17,BDDEW18}, book embeddings~\citep{DF18}, and intersection graph theory~\citep{Shahrokhi13}. The related notion of layered pathwidth has also been studied~\citep{DEJMW,BDDEW18}. Most relevant to this paper,  \citet{DMW17} proved that every graph with $n$ vertices and layered treewidth $k$ has queue-number at most $O(k\log n)$. They then proved that planar graphs have layered treewidth at most 3, that graphs of Euler genus $g$ have layered treewidth at most $2g+3$, and more generally that a minor-closed class has bounded layered treewidth if and only if it excludes some apex graph.\footnote{A graph $G$ is \emph{apex} if $G-v$ is planar for some vertex $v$.} This implies $O(\log n)$ bounds on the queue-number for all these graphs, and was the basis for the $\log^{O(1)}n$ bound for proper minor-closed classes mentioned in \cref{Introduction}.

%%%%%%%%%%%%%
\subsection{Partitions and Layered Partitions}
\label{Partitions}

The following definitions are central notions in this paper. A \emph{vertex-partition}, or simply \emph{partition}, of a graph $G$ is a set $\PP$ of non-empty sets of vertices in $G$ such that each vertex of $G$ is in exactly one element of $\PP$. Each element of $\PP$ is called a \emph{part}. The \emph{quotient} (sometimes called the \emph{touching pattern}) of $\PP$ is the graph, denoted by $G/\PP$, with vertex set $\PP$ where distinct parts $A,B\in \PP$ are adjacent in $G/\PP$ if and only if some vertex in $A$ is adjacent in $G$ to some vertex in $B$. 

A partition of $G$ is \emph{connected} if the subgraph induced by each part is connected. In this case, the quotient is the minor of $G$ obtained by contracting each part into a single vertex. Most of our results for queue layouts do not depend on the connectivity of partitions. But we consider it to be of independent interest that many of the partitions constructed in this paper are connected. Then the quotient is a minor of the original graph. 

A partition $\PP$ of a graph $G$ is called an \emph{$H$-partition} if $H$ is a graph that contains a spanning subgraph isomorphic to the quotient $G/\PP$. Alternatively, an \emph{$H$-partition} of  a graph $G$ is a partition $(A_x:x\in V(H))$ of $V(G)$ indexed by the vertices of $H$, such that for every edge $vw\in E(G)$, if $v\in A_x$ and $w\in A_y$ then $x=y$ (and $vw$ is called an \emph{intra-bag} edge) or $xy\in E(H)$ (and $vw$ is called an \emph{inter-bag} edge).  The \emph{width} of such an $H$-partition is $\max\{|A_x|: x\in V(H)\}$. Note that a layering is equivalent to a path-partition. 

A \emph{tree-partition} is a $T$-partition for some tree $T$. Tree-partitions are well studied with several applications~\citep{DO95,DO96,Wood09,Seese85,BodEng-JAlg97}. For example, every graph with treewidth $k$ and maximum degree $\Delta$ has a tree-partition of width $O(k\Delta)$; see~\citep{Wood09,DO95}. This easily leads to a $O(k\Delta)$ upper bound on the queue-number~\citep{DMW05}. However, dependence on $\Delta$ seems unavoidable when studying tree-partitions~\citep{Wood09}, so we instead consider $H$-partitions where $H$ has bounded treewidth greater than 1. This idea has been used by many authors in a variety of applications, including cops and robbers~\citep{Andreae86}, fractional colouring~\citep{ReedSeymour-JCTB98,SSW19}, generalised colouring numbers~\citep{HOQRS17}, and defective and clustered colouring~\citep{vdHW18}. See~\citep{DOSV98,DOSV-JCTB00} for more on partitions of graphs in a proper minor-closed class. 

A key innovation of this paper is to consider a layered variant of partitions  (analogous to layered treewidth being a layered variant of treewidth). The \emph{layered width} of  a partition $\PP$ of a graph $G$  is the minimum integer $\ell$ such that for some layering $(V_0,V_1,\dots)$ of $G$, each part in $\PP$ has at most $\ell$ vertices in each layer $V_i$. 

The $n\times n$ grid graph $G$ provides an instructive example. The columns determine a partition $\PP$ of layered width 1 with respect to the layering determined by the rows. The quotient $G/\PP$ is an $n$-vertex path. 

Throughout this paper we consider partitions with bounded layered width such that the quotient has bounded treewidth. We therefore introduce the following definition. A class $\mathcal{G}$ of graphs is said to \emph{admit bounded layered partitions} if there exist $k,\ell\in\mathbb{N}$ such that every graph $G\in \mathcal{G}$ has a partition $\PP$ with layered width at most $\ell$ such that $G / \PP$ has treewidth at most $k$. We first show that this property immediately implies bounded layered treewidth.

\begin{lemma}
\label{PartitionLayeredTreewidth}
If a graph $G$ has an $H$-partition with layered width at most $\ell$ such that $H$ has treewidth at most $k$, 
then $G$ has layered treewidth at most $(k+1)\ell$.
\end{lemma}

\begin{proof}
Let $(B_x:x\in V(T))$ be a tree-decomposition of $H$ with bags of size at most $k+1$. Replace each instance of a vertex $v$ of $H$ in a bag $B_x$ by the part corresponding to $v$ in the $H$-partition. Keep the same layering of $G$. Since $|B_x|\leq k+1$, we obtain a tree-decomposition of $G$ with layered width at most $(k+1)\ell$.
\end{proof}

\cref{PartitionLayeredTreewidth} means that any property that holds for graph classes with bounded layered treewidth also holds for graph classes that admit bounded layered partitions. For example, Norin proved that every $n$-vertex graph with layered treewidth at most $k$ has treewidth less than $2\sqrt{kn}$ (see~\citep{DMW17}). With \cref{PartitionLayeredTreewidth}, this implies that if an $n$-vertex graph $G$ has a partition with layered width $\ell$ such that the quotient graph has treewidth at most $k$, then $G$ has treewidth at most $2\sqrt{(k+1)\ell n}$. This in turn leads to $O(\sqrt{n})$ balanced separator theorems for such graphs. 

\cref{PartitionLayeredTreewidth} suggests that having a partition of bounded layered width, whose quotient has bounded treewidth, seems to be a more stringent requirement than having bounded layered treewidth. Indeed the former structure leads to $O(1)$ bounds on the queue-number, instead of $O(\log n)$ bounds obtained via layered treewidth. That said, it is open whether graphs of bounded layered treewidth have bounded queue-number. 
%It is even possible that graphs of bounded layered treewidth admit bounded layered partitions.

Before continuing, we show that if one does not care about the exact treewidth bound, then it suffices to consider partitions with layered width 1.

\begin{lemma}
\label{MakeWidth1}
If a graph $G$ has an $H$-partition of layered width $\ell$ with respect to a layering $(V_0,V_1,\dots)$, for some graph $H$ of treewidth at most $k$, then $G$ has an $H'$-partition of layered width 1 with respect to the same layering, for some graph $H'$ of treewidth at most $(k+1)\ell-1$. 
\end{lemma}

\begin{proof}
Let $(A_v:v\in V(H))$ be an $H$-partition of $G$ of layered width $\ell$ with respect to $(V_0,V_1,\dots)$, for some graph $H$ of treewidth at most $k$. 
Let $(B_x:x\in V(T))$ be a tree-decomposition of $H$ with width at most $k$. 
Let $H'$ be the graph obtained from $H$ by replacing each vertex $v$ of $H$ by an $\ell$-clique $X_v$ and replacing each edge $vw$ of $H$ by a complete bipartite graph $K_{\ell,\ell}$ between $X_v$ and $X_w$. 
For each $x\in V(T)$, let $B'_x := \cup\{X_v:v\in B_x\}$.
Observe that $(B'_x:x\in V(T))$ is a tree-decomposition of $H'$ of width at most $(k+1)\ell-1$.
For each vertex $v$ of $H$, and layer $V_i$,
there are at most $\ell$ vertices in $A_v \cap V_i$.
Assign each vertex in $A_v\cap V_i$ to a distinct element of $X_v$.
We obtain an $H'$-partition of $G$ with layered width 1, and the treewidth of $H$ is at most $(k+1)\ell-1$.
\end{proof}

%%%%%%%%%%%%%%%%%%%%%%%%%%%%%%%
\section{Queue Layouts via Layered Partitions}
\label{QLLP}

The next lemma is at the heart of all our results about queue layouts.

\begin{lemma}
\label{PartitionQueue}
For all graphs $H$ and $G$, if $H$ has a $k$-queue layout and $G$ has an $H$-partition of layered width $\ell$ with respect to some layering $(V_0,V_1,\dots)$ of $G$, then $G$ has a $(3\ell k+\floor*{\frac{3}{2}\ell})$-queue layout using vertex ordering $\overrightarrow{V_0},\overrightarrow{V_1},\dots$, where $\overrightarrow{V_i}$ is some ordering of $V_i$. In particular,
$$\qn(G) \leq 3 \ell \,\qn(H) + \floor*{\tfrac{3}{2}\ell}.$$
\end{lemma}

The next lemma is useful in the proof of \cref{PartitionQueue}.

\begin{lemma}
\label{Blowup} 
Let $v_1,\dots,v_n$ be the vertex ordering in a 1-queue layout of a graph $H$. Define a graph $G$ with vertex-set $B_1\cup\dots\cup B_n$, where $B_1,\dots,B_n$ are pairwise disjoint sets of vertices (called `blocks'), each with at most $\ell$ vertices. For each edge $v_iv_j\in E(H)$, add an edge to $G$ between each vertex in $B_i$ and each vertex in $B_j$. Then the vertex-ordering of $G$ obtained from $v_1,\dots,v_n$ by replacing each $v_i$ by $B_i$ admits an $\ell$-queue layout of $G$.
\end{lemma}

\begin{proof}
 A \emph{rainbow} in a vertex ordering of a graph $G$ is a set of pairwise nested edges (and thus a matching).  Say $R$ is a rainbow in the ordering of $V(G)$.  \citet{HR92} proved that a vertex ordering of any graph admits a $k$-queue layout if and only if every rainbow has size at most $k$. Thus it suffices to prove that $|R|\leq\ell$.   If the right endpoints of $R$ belong to at least two different blocks, and the left endpoints of $R$ belong to at least two different blocks, then no endpoint of the innermost edge in $R$ and no endpoint of the outermost edge in $R$ are in a common block, implying that the corresponding edges in $H$ have no endpoint in common, and therefore are nested. Since no two edges in $H$ are nested, without loss of generality, the left endpoints of $R$ belong to one block. Hence there are at most $\ell$ left endpoints of $R$, implying $|R|\leq\ell$, as desired. 
\end{proof}

In what follows, the graph $G$ in \cref{Blowup} is called an \emph{$\ell$-blowup} of $H$.

\begin{proof}[Proof of \cref{PartitionQueue}]
Let $(A_x:x\in V(H))$ be an $H$-partition of $G$ of layered width $\ell$ with respect to some layering $(V_0,V_1,\dots)$ of $G$; that is, $|A_x\cap V_i|\leq \ell$  for all $x\in V(H)$ and $i \geq 0$. Let $(x_1,\dots,x_h)$ be the vertex ordering and $E_1,\dots,E_k$ be the queue assignment in a $k$-queue layout of $H$.

We now construct a $(3\ell k+\floor*{\frac{3}{2}\ell})$-queue layout of $G$. Order each layer $V_i$ by
$$\overrightarrow{V_i} := A_{x_1}\cap V_i, A_{x_2}\cap V_i, \dots,A_{x_h}\cap V_i,$$
where each set $A_{x_j}\cap V_i$ is ordered arbitrarily.
We use the ordering
$\overrightarrow{V_0},\overrightarrow{V_1},\dots$
of $V(G)$ in our queue layout of $G$.
It remains to assign the edges of $G$ to queues.
We consider four types of edges, and use distinct queues for edges of each type. 

%%%%%%%%
\textbf{Intra-level intra-bag edges:} 
Let $G^{(1)}$ be the subgraph formed by the edges $vw\in E(G)$, where $v,w\in A_x\cap V_i$ for some $x\in V(H)$ and $i \geq 0$. \citet{HR92} noted that the complete graph on $\ell$ vertices has queue-number $\floor{\frac{\ell}{2}}$. Since $|A_{x}\cap V_i|\leq \ell$,  at most $\floor{\frac{\ell}{2}}$ queues suffice for edges in the subgraph of $G$ induced by $A_{x}\cap V_i$. These subgraphs are separated in $\overrightarrow{V_0},\overrightarrow{V_1},\dots$. Thus $\floor{\frac{\ell}{2}}$ queues suffice for all intra-level intra-bag edges.

%%%%%%%%
\textbf{Intra-level inter-bag edges:}
For $\alpha\in\{1,\dots,k\}$ and $i \geq 0$, let $G^{(2)}_{\alpha,i}$ be the subgraph of $G$ formed by those edges $vw\in E(G)$ such that $v\in A_x\cap V_i$ and $w\in A_y \cap V_i$ for some edge $xy\in E_\alpha$.
Let $Z^{(2)}_\alpha$ be the $1$-queue layout of the subgraph $(V(H),E_\alpha)$ of $H$ on all edges in queue $\alpha$.
Observe that $G^{(2)}_{\alpha,i}$ is a subgraph of the graph isomorphic to the $\ell$-blowup of $Z^{(2)}_\alpha$.
By \cref{Blowup},
$\overrightarrow{V_0},\overrightarrow{V_1},\dots$
admits an $\ell$-queue layout of $G^{(2)}_{\alpha,i}$.
As the subgraphs $G^{(2)}_{\alpha,i}$ for fixed $\alpha$ but different $i$ are separated in $\overrightarrow{V_0},\overrightarrow{V_1},\dots$, $\ell$ queues suffice for edges in $\bigcup_{i \geq 0} G^{(2)}_{\alpha,i}$ for each $\alpha \in \{1,\ldots,k\}$.
Hence $\overrightarrow{V_0},\overrightarrow{V_1},\dots$
admits an $\ell k$-queue layout of the intra-level inter-bag edges.

%%%%%%%%
\textbf{Inter-level intra-bag edges:}
Let $G^{(3)}$ be the subgraph of $G$ formed by those edges $vw \in E(G)$ such that $v\in A_x \cap V_i$ and $w \in A_x \cap V_{i+1}$ for some $x \in V(H)$ and $i \geq 0$.
Consider the graph $Z^{(3)}$ with ordered vertex set
\[
 z_{0,x_1},\dots,z_{0,x_h};\, z_{1,x_1},\dots,z_{1,x_h};\, \dots
\]
and edge set $\{z_{i,x}z_{i+1,x} : i \geq 0, x \in V(H)\}$.
Then no two edges in $Z^{(3)}$ are nested.
Observe that $G^{(3)}$ is isomorphic to a subgraph of the $\ell$-blowup of $Z^{(3)}$.
By \cref{Blowup},
$\overrightarrow{V_0},\overrightarrow{V_1},\dots$ admits
an $\ell$-queue layout of the intra-level inter-bag edges.

%%%%%%%%
\textbf{Inter-level inter-bag edges:}
We partition these edges into $2k$ sets.
For $\alpha\in\{1,\dots,k\}$,
let  $G^{(4a)}_\alpha$ be the spanning subgraph of $G$
formed by those edges $vw\in E(G)$ where $v\in A_x\cap V_i$ and $w\in A_y\cap V_{i+1}$ for some $i\geq 0$
and for some edge $xy$ of $H$ in $E_\alpha$, with $x\prec y$  in the ordering of $H$.
Similarly,
for $\alpha\in\{1,\dots,k\}$,
let  $G^{(4b)}_\alpha$ be the spanning subgraph of $G$
formed by those edges $vw\in E(G)$ where $v\in A_x\cap V_i$ and $w\in A_y\cap V_{i+1}$ for some $i \geq 0$
and for some edge $xy$ of $H$ in $E_\alpha$, with $y \prec x$ in the ordering of $H$.

For $\alpha\in\{1,\dots,k\}$, let $Z^{(4a)}_\alpha$ be the graph with ordered vertex set
\[
 z_{0,x_1},\dots,z_{0,x_h};\, z_{1,x_1},\dots,z_{1,x_h};\, \dots
\]
and edge set
$\{z_{i,x}z_{i+1,y}: i \geq 0, x,y \in V(H), xy \in E_\alpha, x \prec y \}$.
Suppose that two edges in $Z^{(4a)}$ nest.
This is only possible for edges $z_{i,x}z_{i+1,y}$ and $z_{i,p}z_{i+1,q}$,
where $z_{i,x} \prec z_{i,p} \prec z_{i+1,q} \prec z_{i+1,y}$.
Thus, in $H$, we have $x \prec p$ and $q\prec y$.
By the definition of $Z^{(4a)}$, we have $x \prec y$ and $p\prec q$.
Hence $x \prec p \prec q\prec y$, which contradicts that $xy,pq\in E_\alpha$.
Therefore no two edges are nested in $Z^{(4a)}$.

Observe that $G^{(4a)}_\alpha$ is isomorphic to a subgraph of the $\ell$-blowup of $Z^{(4)}_\alpha$.
By \cref{Blowup},
$\overrightarrow{V_0},\overrightarrow{V_1},\dots$
admits an $\ell$-queue layout of $G^{(4a)}_\alpha$.
An analogous argument shows that
$\overrightarrow{V_0},\overrightarrow{V_1},\dots$
admits an $\ell$-queue layout of $G^{(4b)}_\alpha$.
Hence
$\overrightarrow{V_0},\overrightarrow{V_1},\dots$
admits a $2k\ell$-queue layout of all the inter-level inter-bag edges.

In total, we use $\floor*{\frac{\ell}{2}}+ k\ell+\ell+ 2k\ell$ queues.
\end{proof}

The upper bound of $3 \ell \,\qn(H) + \floor*{\frac{3}{2}\ell}$ in \cref{PartitionQueue} is tight, in the sense that it is possible that the vertex ordering produced by \cref{PartitionQueue} has 
$3 \ell \,\qn(H) + \floor*{\frac{3}{2}\ell}$ pairwise nested edges, and thus at least this many queues are needed.

\cref{PartitionQueue,TreewidthQueue} imply that a graph class that admits bounded layered partitions has bounded queue-number. In particular:

\begin{corollary}
\label{klPartitionQueue}
If a graph $G$ has a partition $\PP$ of layered width $\ell$ such that $G / \PP$ has treewidth at most $k$, then $G$ has queue-number at most 
$3 \ell (2^k-1) + \floor*{\tfrac{3}{2}\ell}$.
\end{corollary}

%%%%%%%%%%%%%%%%%%%%%
\section{Proof of \cref{PlanarQueue}: Planar Graphs}\label{planar}

Our proof that planar graphs have bounded queue-number employs \cref{klPartitionQueue}. Thus our goal is to show that planar graphs admit bounded layered partitions, which is achieved  in the following key contribution of the paper.

\begin{theorem}
\label{Width1LayeredPartition}
Every planar graph $G$ has a connected partition $\PP$ with layered width $1$ such that $G / \PP$ has treewidth at most $8$. 
Moreover, there is such a partition for every BFS layering of $G$.
\end{theorem}

This theorem and \cref{klPartitionQueue} imply that planar graphs have bounded queue-number (\cref{PlanarQueue}) with an upper bound of $3 (2^8-1) + \floor*{\tfrac{3}{2} 3} = 766$.

We now set out to prove \cref{Width1LayeredPartition}. The proof is inspired by the following elegant result  of \citet{PS18}: Every planar graph $G$ has a partition $\PP$ into geodesics such that $G/\PP$ has treewidth at most $8$. Here, a \emph{geodesic} is a path of minimum length between its endpoints. We consider the following particular type of geodesic. If $T$ is a tree rooted at a vertex $r$, then a non-empty path $(x_1,\dots,x_p)$ in $T$ is \emph{vertical} if for some $d\geq 0$ for all $i\in\{0,\dots,p\}$ we have $\dist_T(x_i,r)=d+i$.  The vertex $x_1$ is called the \emph{upper endpoint} of the path and $x_p$ is its \emph{lower endpoint}. Note that every vertical path in a BFS spanning tree is a geodesic. Thus the next theorem strengthens the result of \citet{PS18}.

\begin{theorem}
\label{FindVerticalPaths}
Let $T$ be a rooted spanning tree in a connected planar graph $G$. Then $G$ has a partition $\PP$ into vertical paths in $T$ such that $G / \PP$ has treewidth at most $8$.
\end{theorem}

\begin{proof}[Proof of \cref{Width1LayeredPartition} assuming \cref{FindVerticalPaths}] 
We may assume that $G$ is connected (since if each component of $G$ has the desired partition, then so does $G$). 
Let $T$ be a BFS spanning tree of $G$. 
By \cref{FindVerticalPaths}, $G$ has a partition $\PP$ into vertical paths in $T$ such that $G/\PP$ has treewidth at most $8$.  
Each path in $\PP$ is connected and has at most one vertex in each BFS layer corresponding  to $T$. Hence 
$\PP$ is connected and has layered width 1.
\end{proof}

The proof of \cref{FindVerticalPaths} is an inductive proof of a stronger statement given in \cref{NearTriang} below. A \emph{plane graph} is a graph embedded in the plane with no crossings. A \emph{near-triangulation} is a plane graph, where the outer-face is a simple cycle, and every internal face is a triangle. For a cycle $C$, we write $C=[P_1,\dots,P_k]$ if $P_1,\dots,P_k$ are pairwise disjoint non-empty paths in $C$, and the endpoints of each path $P_i$ can be labelled $x_i$ and $y_i$ so that $y_ix_{i+1} \in E(C)$ for $i\in\{1,\dots,k\}$, where $x_{k+1}$ means $x_1$. This implies that $V(C)=\bigcup_{i=1}^k V(P_i)$.

\begin{lemma}
\label{NearTriang}
Let $G^+$ be a plane triangulation, let $T$ be a spanning tree of $G^+$ rooted at some vertex $r$ on the outer-face of $G^+$, and let $P_1,\ldots,P_k$ for some $k\in\{1,2,\dots, 6\}$, be pairwise disjoint vertical paths in $T$ such that $F=[P_1,\ldots,P_k]$ is a cycle in $G^+$. Let $G$ be the near-triangulation consisting of all the edges and vertices of $G^+$ contained in $F$ and the interior of $F$.

Then $G$ has a partition $\PP$ into paths in $G$ that are vertical in $T$, such that $P_1,\ldots,P_k\in\PP$ and the quotient $H:=G/\PP$ has a tree-decomposition in which every bag has size at most 9 and some bag contains all the vertices of $H$ corresponding to $P_1,\ldots,P_k$. 
\end{lemma}

%\comment{Piotr: I have changed the statement. Before it was "Then $G$ has a connected partition $\PP$ into vertical paths in $T$ where ...". Since we have the supergraph $G^+$, I think it is better to be explicit here. DW: revised}

\begin{proof}[Proof of \cref{FindVerticalPaths} assuming \cref{NearTriang}] 
The result is trivial if $|V(G)|<3$. Now assume $|V(G)|\geq 3$. 
Let $r$ be the root of $T$. 
Let $G^+$ be a plane triangulation containing $G$ as a spanning subgraph with $r$ on the outer-face of $G$.  
The three vertices on the outer-face of $G$ are vertical (singleton) paths in $T$. 
Thus $G^+$ satisfies the assumptions of \cref{NearTriang} with $k=3$, which implies that $G^+$ has a partition $\PP$ into vertical paths in $T$ such that $G^+/\PP$ has treewidth at most $8$. 
Note that $G/\PP$ is a subgraph of $G^+/\PP$.
Hence $G/\PP$ has treewidth at most $8$.
\end{proof}

%\comment{Piotr: Within the "Proof of Thm 10 assuming Lemma 11": I don't get why do we introduce this new vertex. It should be enou to keep v as the root on the outerface of the triangulation. Let's simplify it or help me to understand the reason why it is there.\\ Piotr: So, I have removed this extra vertex $r$. Please double-check. If it is correct, I think we should change $G_0$ to $G$ and $T_0$ to $T$ in the statement of Thm 10.\\ DW: revised}

Our proof of \cref{NearTriang} employs the following well-known variation of Sperner's Lemma (see~\citep{Proofs4}):

  \begin{lemma}[Sperner's Lemma]
  \label{Sperner}
Let $G$ be a near-triangulation whose vertices are coloured $1,2,3$, with the outer-face $F=[P_1,P_2,P_3]$ where each vertex in $P_i$ is coloured $i$. Then $G$ contains an internal face whose vertices are coloured $1,2,3$. 
  \end{lemma}

\begin{proof}[Proof of \cref{NearTriang}] 
The proof is by induction on $n=|V(G)|$. If $n=3$, then $G$ is a 3-cycle and $k\le 3$. The partition into vertical paths is $\PP=\{P_1,\ldots,P_k\}$.  The tree-decomposition of $H$ consists of a single bag that contains the $k\le 3$ vertices corresponding to $P_1,\ldots,P_k$.

  \begin{figure}[!b]
    \begin{center}
      \begin{tabular}{cc}
        \includegraphics[scale=1.4]{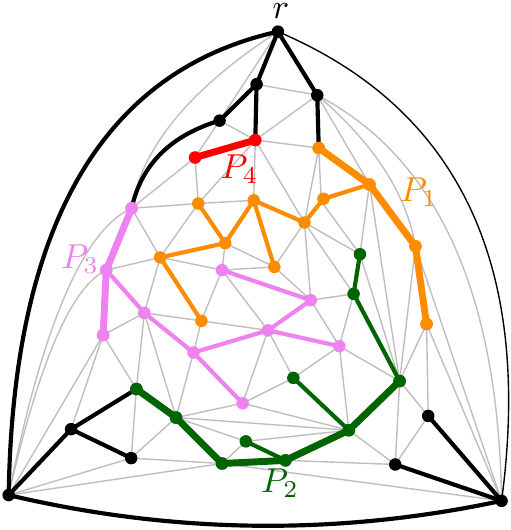} &
        \includegraphics[scale=1.4]{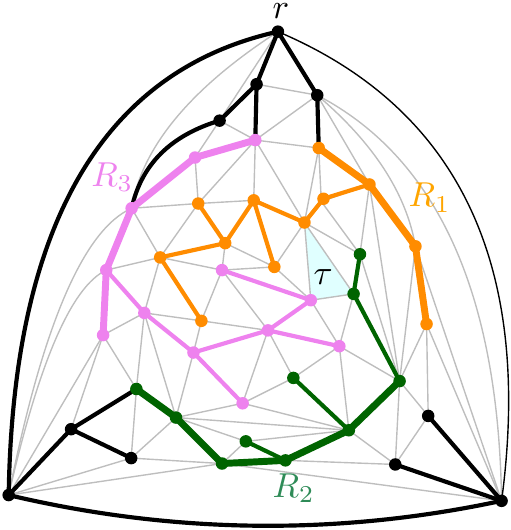} \\
        (a) & (b) \\[1em]
        \includegraphics[scale=1.4]{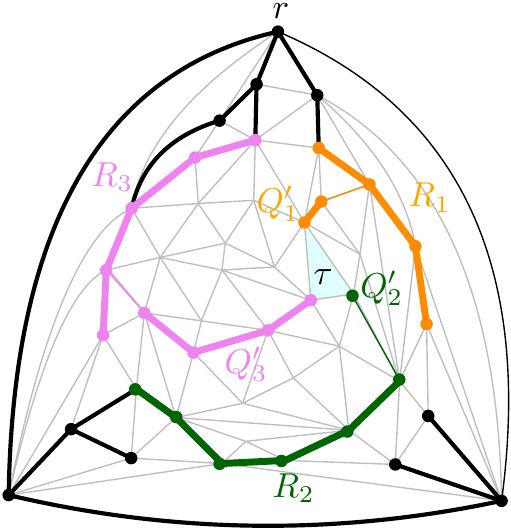} &
        \includegraphics[scale=1.4]{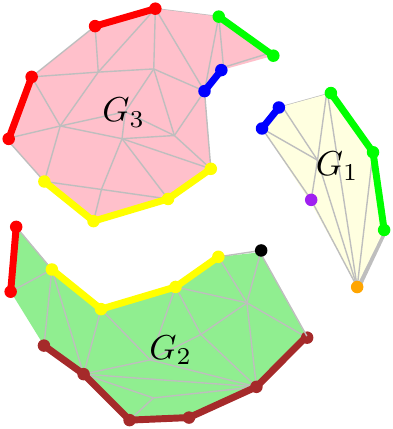} \\
        (c) & (d) 
      \end{tabular}
    \end{center}
    \caption{The inductive proof of \cref{NearTriang}: (a)~the spanning tree $T$ and the paths $P_1,\ldots,P_4$; 
    (b) the paths $R_1$, $R_2$, $R_3$, and the Sperner triangle $\tau$; 
    (c)~the paths $Q_1'$, $Q_2'$ and $Q_3'$; 
    (d)~the near-triangulations $G_1$, $G_2$, and $G_3$, with the vertical paths of $T$ on $F_1$, $F_2$, and $F_3$.}
    \label{NearTriangProof}
  \end{figure}
  
%\comment{Piotr:  There is a mistake in the figure. On the boundary of $G_3$, the depicted
%subpath of $R_3$ is not vertical. It is a union of two vertical paths.\\
%DW: Pat fixed it.}

  For $n > 3$ we wish to make use of Sperner's Lemma on some (not necessarily proper) 3-colouring of the vertices of $G$.  We begin by colouring the vertices of $F$, as illustrated in \cref{NearTriangProof}.  There are three cases to consider:
  \begin{compactenum}
  \item If $k=1$ then, since $F$ is a cycle, $P_1$ has at least three vertices, so $P_1=[v, P_1', w]$ for two distinct vertices $v$ and $w$.  We set $R_1:=v$, $R_2:=P_1'$ and $R_3:=w$.

  \item If $k=2$ then we may assume without loss of generality that $P_1$ has at least two vertices so $P_1=[v,P_1']$.  We set $R_1:=v$, $R_2:=P_1'$ and $R_3:=P_2$.

  \item If $k\in\{3,4,5,6\}$ then we group consecutive paths by taking $R_1:=[P_1,\ldots,P_{\floor{k/3}}]$, $R_2:=[P_{\floor{k/3}+1},\ldots,P_{\floor{2k/3}}]$ and $R_3:=[P_{\floor{2k/3}+1},\ldots,P_k]$. Note that in this case each $R_i$ consists of one or two of $P_1,\dots,P_k$. 
  \end{compactenum}

  For $i\in\{1,2,3\}$, colour each vertex in $R_i$ by $i$. Now, for each remaining vertex $v$ in $G$, consider the path $P_v$ from $v$ to the root of $T$. Since $r$ is on the outer-face of $G^+$, $P_v$ contains at least one vertex of $F$. If the first vertex of $P_v$ that belongs to $F$ is in $R_i$ then assign the colour $i$ to $v$.  In this way we obtain a 3-colouring of the vertices of $G$ that satisfies the conditions of Sperner's Lemma. Therefore, by Sperner's Lemma there exists a triangular face $\tau=v_1v_2v_3$ of $G$ whose vertices are coloured $1,2,3$ respectively.

  For each $i\in\{1,2,3\}$, let $Q_i$ be the path in $T$ from $v_i$ to the first ancestor $v_i'$ of $v_i$ in $T$ that is contained in $F$. Observe that $Q_1$, $Q_2$, and $Q_3$ are disjoint since $Q_i$ consists only of vertices coloured $i$. Note that $Q_i$ may consist of the single vertex $v_i=v_i'$.  Let $Q_i'$ be $Q_i$ minus its final vertex $v_i'$.  Imagine for a moment that the cycle $F$ is oriented clockwise, which defines an orientation of $R_1$, $R_2$ and $R_3$. Let $R_i^-$ be the subpath of $R_i$ that contains $v'_i$ and all vertices that precede it, and let $R_i^+$ be the subpath of $R_i$ that contains $v'_i$ and all vertices that succeed it.  
  %Again, $R_i^-$ and $R_i^+$ may be empty if $v'_i$ is the first and/or last vertex of $R_i$.

%\comment{Piotr: Regarding ``Again, $R_i^-$ and $R_i^+$ may be empty if $v'_i$ is the first and/or last vertex of $R_i$'',  I don't think it could be empty, since $v'_i$ is always in.\\ Piotr: I commented out that line.\\ DW: checked}

Consider the subgraph of $G$ that consists of the edges and vertices of $F$, the edges and vertices of $\tau$, and the edges and vertices of $Q_1\cup Q_2\cup Q_3$. 
This graph has an outer-face, an inner face $\tau$, and up to three more inner faces $F_1,F_2,F_3$ where $F_i=[Q_i',R_i^+,R_{i+1}^-,Q_{i+1}']$, where we use the convention that $Q_4=Q_1$ and $R_4=R_1$. Note that $F_i$ may be \emph{degenerate} in the sense that $[Q_i',R_i^+,R_{i+1}^-,Q_{i+1}']$ may consist only of a single edge $v_iv_{i+1}$.
  
%\comment{Piotr:  I don't like that we overload the word "empty". $F_i$ could be also  considered to be empty when its a triangle and there are no vertices inside. How about replacing it with "degenerated"?\\Piotr: Implemented.\\ DW: replaced with ``degenerate''}

Consider any non-degenerate $F_i=[Q_i',R_i^+,R_{i+1}^-,Q_{i+1}']$. Note that these four paths  are pairwise disjoint, and thus $F_i$ is a cycle. If  $Q_i'$ and $Q_{i+1}'$ are non-empty, then each is a vertical path in $T$.  Furthermore, each of $R_i^-$ and $R_{i+1}^+$ consists of at most two vertical paths in $T$.  Thus, $F_i$ is the concatenation of at most six vertical paths in $T$.  Let $G_i$ be the near-triangulation consisting of all the edges and vertices of $G^+$ contained in $F_i$ and the interior of $F_i$.  Observe that $G_i$ contains $v_i$ and $v_{i+1}$ but not the third vertex of $\tau$. Therefore $F_i$ satisfies the conditions of the lemma and has fewer than $n$ vertices. So we may apply induction on $F_i$ to obtain a partition $\PP_i$ of $G_i$ into vertical paths in $T$, such that $H_i := G_i / \PP_i$ has a tree-decomposition $(B^i_x:x\in V(J_i))$ in which every bag has size at most 9, and some bag $B^i_{u_i}$ contains  the vertices of $H_i$ corresponding to the at most six vertical paths that form $F_i$.  We do this for each non-degenerate $F_i$. 

%\comment{Piotr: I have rewritten all the proof from here. Please verify and improve.\\ DW: Checked and improved}

We now construct the desired partition $\PP$ of $G$. 
Initialise $\PP:= \{P_1,\ldots,P_k\}$. Then add each non-empty $Q_i'$ to $\PP$. 
Now for each non-degenerate $F_i$, each path in  $\PP_i$ is either an \emph{external} path (that is, fully contained in $F_i$) or is an \emph{internal path} with none of its vertices in $F_i$. 
Add all the internal paths of  $\PP_i$ to $\PP$. 
By construction, $\PP$ partitions $V(G)$ into vertical paths in $T$ and $\PP$ contains $P_1,\ldots,P_k$.

Let $H:=G/\PP$. Next we exhibit the desired tree-decomposition $(B_x:x\in V(J))$ of $H$. 
Let $J$ be the tree obtained from the disjoint union of $J_i$, taken over the $i\in\{1,2,3\}$ such that $F_i$ is non-degenarate, by adding one new node $u$ adjacent to each $u_i$.
(Recall that $u_i$ is the node of $J_i$ for which the bag $B^i_{u_i}$ contains the vertices of $H_i$ corresponding to the paths that form $F_i$.)\ 
Let the bag $B_u$ contain all the vertices of $H$ corresponding to $P_1,\ldots,P_k,Q'_1,Q'_2,Q'_3$. 
%It is helpful to think of $J$ as being rooted at $u$.  Since $k\le 6$, $|B_u|\leq 9$.
For each non-degenerate $F_i$, and for each node $x\in V(J_i)$, initialise $B_x:= B^i_x$. 
Recall that vertices of $H_i$ correspond to contracted paths in  $\PP_i$.
Each internal path in $\PP_i$ also lies  in $\PP$.
Each external path $P$ in $\PP_i$ is a subpath of $P_j$ for some $j\in\{1,\dots,k \}$ or is one of the paths among $Q'_1, Q'_2, Q'_3$. 
For each such path $P$, for every $x\in V(J)$, in bag $B_x$, 
replace each instance of the vertex of $H_i$ corresponding to $P$ by the vertex of $H$ corresponding to the path among $P_1,\ldots, P_k, Q'_1,\ldots,Q'_3$ 
that contains $P$.
This completes the description of $(B_x : x\in V(J))$. 
By construction, $|B_x|\leq 9$ for every $x\in V(J)$. 

First we show that for each vertex $a$ in $H$, the set $X:=\{x\in V(J) : a\in B_x\}$ forms a subtree of $J$.
If $a$ corresponds to a path distinct from $P_1,\ldots,P_k,Q'_1,Q'_2,Q'_3$ then $X$ is fully contained in $J_i$ for some $i\in\{1,2,3\}$.
Thus, by induction $X$ is non-empty and connected in $J_i$, so it is in $J$.
If $a$ corresponds to $P$ which is one of the paths among $P_1,\ldots,P_k,Q'_1,Q'_2,Q'_3$ then 
$u\in X$ and whenever $X$ contains a vertex of $J_i$ it is because some external path of $\PP_i$ was replaced by $P$.
In particular, we would have $u_i \in X$ in that case. 
Again by induction each $X\cap J_i$ is connected and since $uu_i\in E(T)$, we conclude that $X$ induces a (connected) subtree of $J$.

Finally we show that, for every edge $ab$ of $H$, there is a bag $B_x$ that contains $a$ and $b$. 
If $a$ and $b$ are both obtained by contracting any of $P_1,\ldots,P_k,Q'_1,Q'_2,Q'_3$, 
then $a$ and $b$ both appear in $B_u$.  
If $a$ and $b$ are both in $H_i$ for some $i\in\{1,2,3\}$, 
then some bag $B^i_x$ contains both $a$ and $b$. 
Finally, when $a$ is obtained by contracting a path $P_a$ in $G_i-V(F_i)$ and $b$ is obtained by contracting a path $P_b$ not in $G_i$, 
then  the cycle $F_i$ separates $P_a$ from $P_b$ so the edge $ab$ is not present in $H$. 
This concludes the proof that $(B_x:x\in V(J))$ is the desired tree-decomposition  of $H$. 
\end{proof}

%%%%%%%%%%%%%%%
\subsection{Reducing the Bound}

We now set out to reduce the constant in  \cref{PlanarQueue} from $766$ to $49$. This is achieved by proving the following variant of \cref{Width1LayeredPartition}.

\begin{theorem}
\label{Width3LayeredPartition}
Every planar graph $G$ has a partition $\PP$ with layered width $3$ such that $G / \PP$ is planar and has treewidth at most $3$. Moreover, there is such a partition for  every BFS layering of $G$. 
\end{theorem}

This theorem with \cref{PartitionQueue,PlanarTreewidth3Queue5} imply that planar graphs have bounded queue-number (\cref{PlanarQueue}) with an upper bound of $3 \cdot 3 \cdot 5 + \floor*{\tfrac{3}{2}\cdot 3} = 49$. 

Note that \cref{Width3LayeredPartition} is stronger than \cref{Width1LayeredPartition} in that the treewidth bound is smaller, whereas \cref{Width1LayeredPartition} is stronger than \cref{Width3LayeredPartition} in that the partition is connected and the layered width is smaller. Also note that \cref{Width3LayeredPartition} is tight in terms of the treewidth of $H$: For every $\ell$, there exists a planar graph $G$ such that, if $G$ has a partition $\PP$ of layered width $\ell$, then $G / \PP$ has treewidth at least~$3$. We give this construction at the end of this section, and prove \cref{Width3LayeredPartition} first.  \cref{Width1LayeredPartition} was proved via an inductive proof of a stronger statement given in  \cref{NearTriang}. Similarly, the proof of \cref{Width3LayeredPartition} is via an inductive proof of a stronger statement given in \cref{NearTriangTripods}, below.

%%%%%%%%%%%%%%%%%%%%%%%%%%%%%%%%%%%%%%%%%%%%%%%%%

While \cref{FindVerticalPaths} partitions the vertices of a planar graph into vertical paths, 
to prove \cref{Width3LayeredPartition} we instead partition the vertices of a triangulation $G^+$ into parts each of which is a union of up to three vertical paths. 
Formally, in a rooted spanning tree $T$ of a graph $G$, 
a \emph{tripod} consists of up to three pairwise disjoint vertical paths in $T$ whose lower endpoints form a clique in $G$.  
%When a tripod consists of at most two vertical paths in $T$, it is called a \emph{bipod}. Note that a bipod is in fact a path in $G$.
\cref{Width3LayeredPartition} quickly follows from the next result. 

%Let $T_0$ be a rooted spanning tree in a connected planar graph $G_0$. Then $G_0$ has a partition $\PP$ into tripods in $T_0$ such that $G_0 / \PP$ is  planar with treewidth at most $3$.
%Let $T$ be a rooted spanning tree in a connected planar graph $G$. Then $G$ has a partition $\PP$ into tripods in $T$ such that $G / \PP$ is  planar with treewidth at most $3$.

\begin{theorem}
\label{FindTripods}
Let $T$ be a rooted spanning tree in a triangulation $G$. 
Then $G$ has a partition $\PP$ into tripods in $T$ such that $G / \PP$ has treewidth at most $3$.
\end{theorem}

\begin{proof}[Proof of \cref{Width3LayeredPartition} assuming \cref{FindTripods}] 
We may assume that $G$ is connected (since if each component of $G$ has the desired partition, then so does $G$). Let $T$ be a BFS spanning tree of $G$. Let $(V_0,V_1,\dots)$ be the BFS layering corresponding  to $T$.  Let $G'$ be a plane triangulation containing $G$ as a spanning subgraph. By \cref{FindTripods}, $G'$ has a partition $\PP$ into tripods in $T$ such that  $G' / \PP$ is planar with treewidth at most $3$.  Then $\PP$ is a partition of $G$ such that  $G / \PP$ is planar with treewidth at most $3$.  Each part in $\PP$ corresponds to a tripod, which 
has at most three vertices in each layer $V_i$. Hence $\PP$ has layered width at most 3. 
\end{proof}

\cref{FindTripods} is proved via the following lemma. 

\begin{lemma}\label{NearTriangTripods}
Let $G^+$ be a plane triangulation, let $T$ be a spanning tree of $G^+$ rooted at some vertex $r$ on the boundary of the outer-face of $G^+$, and let $P_1,\ldots,P_k$, for some $k\in\{1,2,3\}$, be pairwise disjoint bipods such that $F=[P_1,\ldots,P_k]$ is a cycle in $G^+$ with $r$ in its exterior.  Let $G$ be the near triangulation consisting of all the edges and vertices of $G^+$ contained in $F$ and the interior of $F$.

Then $G$ has a partition $\PP$ into tripods such that $P_1,\ldots,P_k\in\PP$, and the graph $H:= G/ \PP$ is  planar and has a tree-decomposition in which every bag has size at most $4$ and some bag contains all the vertices of $H$ corresponding to $P_1,\ldots,P_k$.
\end{lemma}

%\comment{DW: The following proof has been revised. Please check.}

\begin{proof}[Proof of \cref{FindTripods} assuming \cref{NearTriangTripods}.] 
	Let $T$ be a spanning tree in a triangulation $G$ rooted at vertex $v$. We may assume that $v$ is on the boundary of the outer-face of $G$. Let $G^+$ be the plane triangulation obtained from $G$ by adding one new vertex $r$ into the outer-face of $G$ and adjacent to each vertex on the boundary of the outer-face of $G$. Let $T^+$ be the spanning tree of $G^+$ obtained from $T$ by adding $r$ and the edge $rv$. Consider $T^+$ to be rooted at $r$. Let $P_1,P_2,P_3$ be the singleton paths consisting of the three vertices on the boundary of the outer-face of $G$. Then $P_1,P_2,P_3$ are disjoint bipods such that $F=[P_1,P_2,P_3]$ is a cycle in $G^+$ with $r$ in its exterior.  Moreover, the near triangulation consisting of all the edges and vertices of $G^+$ contained in $F$ and the interior of $F$ is $G$ itself. 	Thus $G$ and $G^+$ satisfy the assumptions of \cref{NearTriangTripods}, which implies that $G$ has a partition $\PP$ into tripods in $T$ such $G / \PP$ has treewidth at most $3$. 
\end{proof}

The remainder of this section is devoted to proving \cref{NearTriangTripods}. 

\begin{proof}[Proof of \cref{NearTriangTripods}]
This proof follows the same approach as the proof of \cref{NearTriang}, by induction on $n=|V(G)|$.  We focus mainly on the differences here. The base case $n=3$ is trivial.

  As before we partition the vertices of $F$ into paths $R_1$, $R_2$, and $R_3$.  If $k=3$, then $R_i:=P_i$ for $i\in\{1,2,3\}$.  Otherwise, as before, we split $P_1$ into two (when $k=2$) or three (when $k=1$) paths.
  
  We apply the same colouring as in the proof of \cref{NearTriang}. Then Sperner's Lemma gives a face $\tau=v_1v_2v_3$ of $G$ whose vertices are coloured 1, 2, 3 respectively.  As in  the proof of \cref{NearTriang}, we obtain vertical paths $Q_1$, $Q_2$, and $Q_3$ where each $Q_i$ is a path in $T$ from $v_i$ to $R_i$.  Remove the last vertex from each $Q_i$ to obtain (possibly empty) paths $Q_1'$, $Q_2'$, and $Q_3'$.  Let $Y$ be the tripod consisting of $Q_1'\cup Q_2' \cup Q_3'$ plus the edges of $\tau$ between non-empty $Q_1',Q_2',Q_3'$. 

As before we consider the graph consisting of the edges and vertices of $\tau$, the edges and vertices of $F$ and the edges and vertices of $Q_1,Q_2,Q_3$. This graph has up to three internal faces $F_1,F_2,F_3$ where each $F_i=[Q_i',R_i^+,R_{i+1}^-,Q_{i+1}']$ and $R_{i}^+$ and $R_{i}^-$ are the same portions of $R_i$ as defined in \cref{NearTriang}.   Observe that  $F_i=[ R_i^+ , R_{i+1}^-, I_i ]$, where $R_i^+$ and $R_{i+1}^-$ are bipods, and $I_i$ is the bipod formed by  $Q_i' \cup Q_{i+1}'$. As before, let $G_i$ be the subgraph of $G$ whose vertices and edges are in $F_i$ or its interior. 

For $i\in\{1,2,3\}$, if $F_i$ is non-empty, then $G_i$ and $F_i=[ R_i^+ , R_{i+1}^-, I_i ]$ satisfy the conditions of the lemma, and $G_i$ has fewer vertices than $G$. Thus we may apply induction to $G_i$.  (Note that one or two of 
$R_i^+$, $R_{i+1}^-$ and $I_i$ may be empty, in which case we apply the inductive hypothesis with $k=2$ or $k=1$, respectively.)\  This gives a partition $\PP_i$ of $G_i$ such that $H_i:= G_i / \PP_i$ satisfies the conclusions of the lemma. Let $(B^i_x:x\in V(J_i))$ be a tree-decomposition of $H_i$, in which every bag has size at most 4, and some bag $B^i_{u_i}$ contains the vertices of $H_i$ corresponding to  $R_i^+$, $R_{i+1}^-$ and $I_i$ (if they are non-empty).  
  
We construct $\PP$ as before. Initialise $\PP:=\{P_1,\ldots,P_k,Y\}$. Then, for $i\in \{1,2,3\}$, each tripod in  $\PP_i$ is either fully contained in $F_i$ or it is \emph{internal} with none of its vertices in $F_i$. 
Add all these internal tripods in  $\PP_i$ to $\PP$. By construction, $\PP$ partitions $V(G)$ into tripods. The graph $H:= G/ \PP$ is planar since $G$ is planar and each tripod in $\PP$ induces a connected subgraph of $G$.

Next we produce the tree-decomposition $(B_x:x\in V(J))$ of $H$ that satisfies the requirements of the lemma.  Let $J$ be the tree obtained from the disjoint union of $J_1$, $J_2$ and $J_3$ by adding one new node $u$ adjacent to $u_1$, $u_2$ and $u_3$. Let $B_u$ be the set of at most four vertices of $H$ corresponding to $Y,P_1,\dots,P_k$. For $i\in\{1,2,3\}$ and for each node $x\in V(J_i)$, initialise $B_x:= B^i_x$. 

As in  the proof of \cref{NearTriang},  the resulting structure,  $(B_x:x\in V(J))$, is not yet a tree-decomposition of $H$ since some bags may  contain vertices of $H_i$ that are not necessarily vertices of $H$. Note that unlike in  \cref{NearTriang} this does not only include elements of $\PP_i$  that are contained in $F$. 
In particular, $I_i$ is also not an element of $\PP$ and thus does not correspond to a vertex of $H$.   
We remedy this  as follows. For $x\in V(J)$, in bag $B_x$, replace each instance of the vertex of $H_i$ corresponding to $I_i$ by the vertex of $H$ corresponding to $Y$. 
Similarly, by construction, $R_i^+$ is a subgraph of $P_{\alpha_i}$ for some $\alpha_i\in \{1,\dots,k\}$. 
For $x\in V(J)$, in bag $B_x$, replace each instance of the vertex of $H_i$ corresponding to $R_i^+$ 
by the vertex of $H$ corresponding to $P_{\alpha_i}$. 
Finally, $R_{i+1}^-$ is a subgraph of $P_{\beta_i}$ for some $\beta_i\in \{1,\dots,k\}$. 
For $x\in V(J)$, in bag $B_x$, replace each instance of the vertex of $H_i$ corresponding to $R_{i+1}^-$  by the vertex of $H$ corresponding to $P_{\beta_i}$. 

This completes the description of $(B_x:x\in V(J))$. Clearly, every bag $B_x$ has size at most 4. The proof that $(B_x:x\in V(J))$  is indeed a tree-decomposition of $H$ is completely analogous to the proof in \cref{NearTriang}. 
\end{proof}

The following lemma, which is implied by \cref{Width3LayeredPartition,PartitionQueue,TreewidthQueue}, will be helpful for generalising  our results to bounded genus graphs. 

\begin{lemma}
\label{BFSPlanarQueue}
For every BFS layering $(V_0,V_1,\dots)$ of a planar graph $G$, there is a 49-queue layout of $G$ using vertex ordering $\overrightarrow{V_0},\overrightarrow{V_1}, \dots,$, where $\overrightarrow{V_i}$ is some ordering of $V_i$, $i \geq 0$.
\end{lemma}

As promised above, we now show that \cref{Width3LayeredPartition} is tight in terms of the treewidth of $H$.

\begin{theorem}
For all integers  $k \geq  2$ and $\ell \geq 1$ there is a graph $G$ with treewidth $k$ such that 
if $G$ has a partition $\PP$ with layered width at most $\ell$, then $G / \PP$ contains $K_{k+1}$ 
and thus has treewidth at least $k$. 
Moreover, if $k=2$ then $G$ is outer-planar, and if $k=3$ then $G$ is planar.
\end{theorem}

\begin{proof}
We proceed by induction on $k$. 
Consider the base case with $k=2$.
Let $G$ be the graph obtained from the path on $9\ell^2 + 3\ell$ vertices by adding one dominant vertex $v$ (the so-called \emph{fan graph}). 
Consider an $H$-partition $(A_x:x\in V(H))$ of $G$ with layered width at most $\ell$. 
Since $v$ is dominant in $G$, each vertex is on the layer containing $v$, the previous layer, or the subsequent layer. Thus we may assume there are at most three layers, and each part $A_x$ has at most $3\ell$ vertices.
Say $v$ is in part $A_x$.
Consider deleting $A_x$ from $G$. 
This deletes at most $3\ell-1$ vertices from the path $G-v$.
Thus $G-A_x$ is the union of at most $3\ell$ paths, with at least $9\ell^2 + 1$
vertices in total.
Thus, one such path $P$ in $G-A_x$ has at least $3\ell+1$ vertices.
Thus there is an edge $yz$ in $H-x$, such that $P\cap A_y\neq\emptyset $ and $P\cap A_z\neq\emptyset$. 
Since $v$ is dominant, $x$ is dominant in $H$. 
Hence $\{x,y,z\}$ induces $K_3$ in $H$. 

Now assume the result for $k-1$. 
Thus there is a graph $Q$ with treewidth $k-1$ such that
if $Q$ has an $H$-partition with width at most $\ell$, then $H$ contains $K_k$. 
Let $G$ be obtained by taking $3\ell$ copies of $Q$ and adding one dominant vertex $v$.
Thus $G$ has treewidth $k$.
Consider an $H$-partition $(A_x:x\in V(H))$ of $G$ with layered width at most $\ell$. 
Since $v$ is dominant there are at most three layers, and  each part has at most $3\ell$ vertices.
Say $v$ is in part $A_x$. Since $|A_x| \leq 3\ell$, some copy of $Q$ avoids $A_x$.
Thus this copy of $Q$ has an $(H-x)$-partition of layered width at most $\ell$. 
By assumption, $H-x$ contains $K_k$. 
Since $v$ is dominant, $x$ is dominant in $H$.
Thus $H$ contains $K_{k+1}$, as desired. 

 In the $k=2$ case, $G$ is outer-planar. Thus, in the $k=3$ case, $G$ is planar.
\end{proof}

%%%%%%%%%%%%%%%%%%%%%%%%
\section{Proof of \cref{GenusQueue}: Bounded-Genus Graphs}\label{genus}

As was the case for planar graphs, our proof that bounded genus graphs have bounded queue-number employs \cref{klPartitionQueue}. Thus the goal of this section is to show that our construction of bounded layered partitions for planar graphs can be generalised for graphs of bounded Euler genus. In particular, we show the following theorem of independent interest. 

\begin{theorem}
\label{GenusPartitionTreewidth9}
Every graph $G$ of Euler genus $g$ has a connected partition $\PP$ with layered width at most $\max\{2g,1\}$ such that $G / \PP$ is apex and has treewidth at most $9$. Moreover, there is such a partition for every BFS layering of $G$. 
\end{theorem}

This theorem and \cref{klPartitionQueue} imply that graphs of Euler genus $g$ have bounded queue-number (\cref{GenusQueue}) with an upper bound of $3 \cdot 2g \cdot (2^9-1) + \floor*{\tfrac{3}{2} \,2g}=O(g)$.

Note that \cref{GenusPartitionTreewidth9} is best possible in the following sense. Suppose that every graph $G$ of Euler genus $g$ has a partition $\PP$ with layered width at most $\ell$ such that $G / \PP$ has treewidth at most $k$. By \cref{PartitionLayeredTreewidth}, $G$ has layered treewidth $O(k\ell)$. \citet{DMW17} showed that the maximum layered treewidth of graphs with Euler genus $g$ is $\Theta(g)$. Thus $k\ell\geq \Omega(g)$. 

The rest of this section is devoted to proving \cref{GenusPartitionTreewidth9}. The next lemma is the key to the proof. Many similar results are known in the literature (for example, 
\citep{EW05} or \citep[Lemma~8]{CCL12} or \cite[Section~4.2.4]{MoharThom}), but none prove exactly what we need.

\begin{lemma}
\label{MakePlanar}
Let $G$ be a connected graph with Euler genus $g$.
For every BFS spanning tree $T$ of $G$ rooted at some vertex $r$ with corresponding BFS layering $(V_0,V_1,\dots)$, 
there is a subgraph $Z\subseteq G$ with at most $2g$ vertices in each layer $V_i$, such that $Z$ is connected and $G-V(Z)$ is planar. 
Moreover, there is a connected planar graph $G^+$ containing $G-V(Z)$ as a subgraph, and 
there is a BFS spanning tree $T^+$ of $G^+$ rooted at some vertex $r^+$ 
with corresponding BFS layering $(W_0,W_1,\dots)$ of $G^+$, 
such that  $W_{i} \cap (V(G) \setminus V(Z)) = V_i  \setminus V(Z)$ for all $i\geq 0$, and
$P \cap (V(G) \setminus V(Z))$ is a vertical path in $T$ for every vertical path $P$ in $T^+$. 
\end{lemma}

\begin{proof}
The result is trivial if $g=0$ (just take $Z=\emptyset$ and $G^+=G$ and $r^+=r$ and $W_i=V_i$). Now assume that $g\geq 1$. Fix an embedding of $G$ in a surface of Euler genus $g$. Say $G$ has $n$ vertices, $m$ edges, and $f$ faces. By Euler's formula, $n-m+f=2-g$. Let $D$ be the multigraph with vertex-set the set of faces in $G$, where for each edge $e$ of $G-E(T)$, if $f_1$ and $f_2$ are the faces of $G$ with $e$ on their boundary, then there is an edge joining $f_1$ and $f_2$ in $D$. (Think of $D$ as the spanning subgraph of the dual graph consisting of those edges that do not cross edges in $T$.)\  Note that $|V(D)|=f=2-g - n + m$ and $|E(D)|=m-(n-1)= |V(D)|-1+g$. Since $T$ is a tree, $D$ is connected; see~\citep[Lemma~11]{DMW17} for a proof. Let $T^*$ be a spanning tree of $D$. Thus $|E(D)\setminus E(T^*)|=g$. Let $Q=\{a_1b_1,a_2b_2,\dots,a_gb_g\}$ be the set of edges in $G$ dual to the edges in $E(D)\setminus E(T^*)$. For $i\in\{1,2,\dots,g\}$, let $Z_i$ be the union of the $a_ir$-path and the $b_ir$-path in $T$, plus the edge $a_ib_i$. Let $Z := Z_1\cup Z_2\cup\dots\cup Z_g$. By construction, $Z$ is a connected subgraph of $G$. Say $Z$ has $p$ vertices and $q$ edges. Since $Z$ consists of a subtree of $T$ plus the $g$ edges in $Q$, we have $q = p-1+g$.

We now describe how to `cut' along the edges of $Z$ to obtain a new graph $G'$; see \cref{Cutting}. First, each edge $e$ of $Z$ is replaced by two edges $e'$ and $e''$ in $G'$. Each vertex of $G$ that is incident with no edges in $Z$ is untouched. Consider a vertex $v$ of $G$ incident with edges $e_1,e_2,\dots,e_d$ in $Z$ in clockwise order. In $G'$ replace $v$ by new vertices $v_1,v_2,\dots,v_d$, where $v_i$ is incident with $e'_i$, $e''_{i+1}$ and all the edges incident with $v$ clockwise from $e_i$ to $e_{i+1}$ (exclusive). Here $e_{d+1}$ means $e_1$ and $e''_{d+1}$ means $e''_1$. This operation defines a cyclic ordering of the edges in $G'$ incident with each vertex (where $e''_{i+1}$ is followed by $e'_i$ in the cyclic order at $v_i$). This in turn defines an embedding of $G'$ in some orientable surface. (Note that if $G$ is embedded in a non-orientable surface, then the edge signatures for $G$ are ignored in the embedding of $G'$.)\ Let $Z'$ be the set of vertices introduced in $G'$ by cutting through vertices in $Z$. 

\begin{figure}[!b]
\centering
\includegraphics{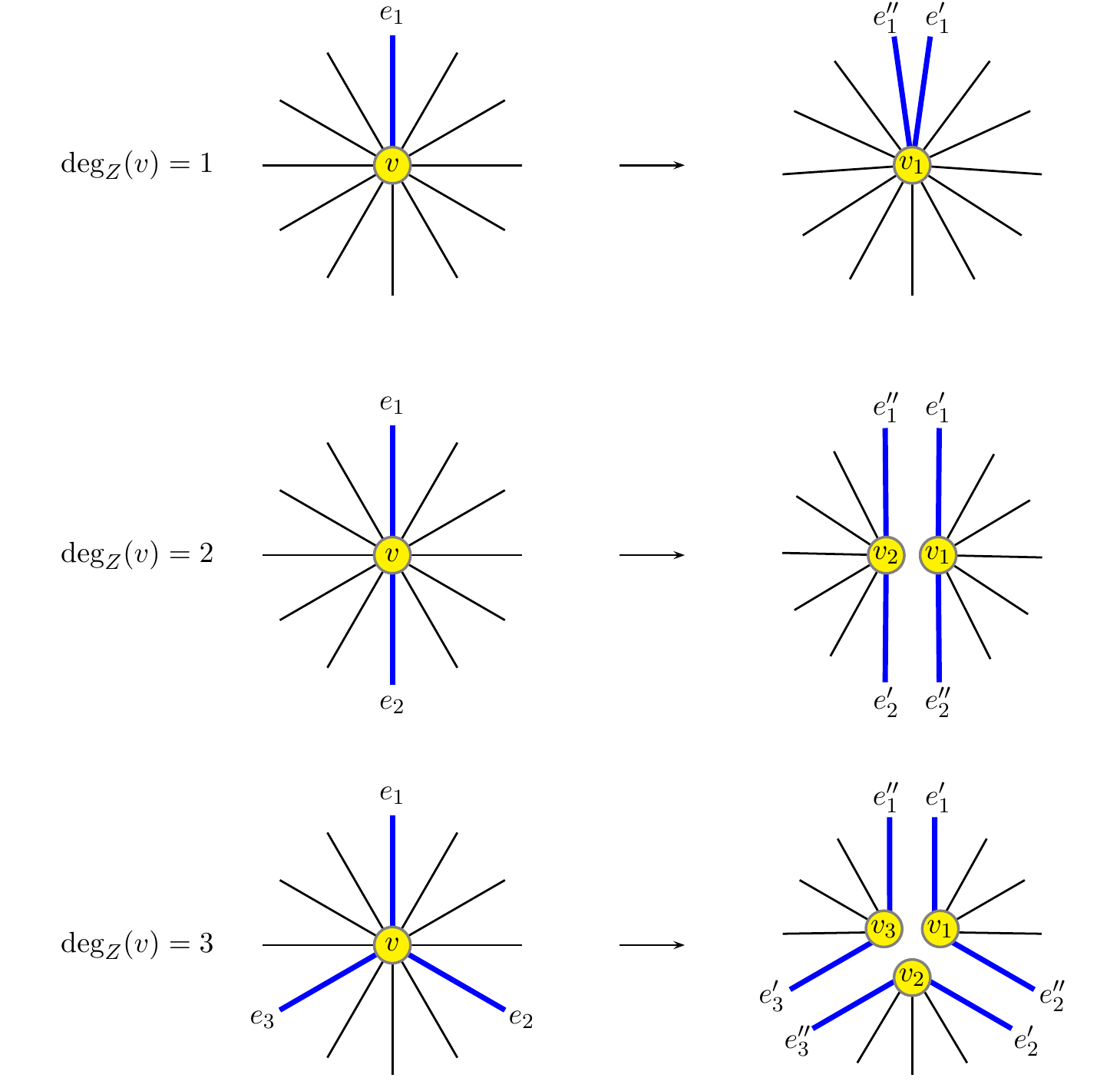}
\caption{Cutting the blue edges in $Z$ at each vertex.\label{Cutting}}
\end{figure}

We now show that $G'$ is connected. Consider vertices $x_1$ and $x_2$ of $G'$. Select faces $f_1$ and $f_2$ of $G'$ respectively incident to $x_1$ and $x_2$ that are also faces of $G$. Let $P$ be a path joining $f_1$ and $f_2$ in the dual tree $T^*$. Then the edges of $G$ dual to the edges in $P$ were not split in the construction of $G'$. Therefore an $x_1x_2$-walk in $G'$ can be obtained by following the boundaries of the faces corresponding to vertices in $P$. Hence $G'$ is connected. 

Say $G'$ has $n'$ vertices and $m'$ edges, and the embedding of $G'$ has $f'$ faces and Euler genus $g'$. Each vertex $v$ in $G$ with degree $d$ in $Z$ is replaced by $d$ vertices in $G'$. Each edge in $Z$ is replaced by two edges in $G'$, while each edge of $G-E(Z)$ is maintained in $G'$. Thus
$$n' = n -p + \sum_{v\in V(G)}\deg_Z(v) = n + 2q -p  = n + 2(p-1+g) - p = n +p - 2 + 2g$$ and $m' = m + q = m + p-1 + g$. Each face of $G$ is preserved in $G'$. Say $s$ new faces are created by the cutting. Thus $f'=f+s$. Since $G'$ is connected, $n'-m'+f'=2-g'$ by Euler's formula. Thus $(n +p - 2 + 2g) - (m + p-1 + g) + (f+s) = 2-g'$, implying $(n-m+f)   - 1 + g + s = 2-g'$. Hence $(2-g)   - 1 + g + s = 2-g'$, implying $g' = 1-s$. Since $g'\geq 0$, we have $s\leq 1$.
Since $g\geq 1$, by construction, $s\geq 1$. Thus $s=1$ and $g'=0$. Thus $G'$ is planar and all the vertices in $Z'$ are on the boundary of a single face, $f$, of $G'$. 

%If $g=0$ then $G'=G$ is planar. If $g\geq 1$ then $s\geq 1$ by construction, implying $g'=0$ and $s=1$. In both cases, 

Note that $G-V(Z)$ is a subgraph of $G'$, and thus $G-V(Z)$ is planar. By construction, each path $Z_i$ has at most two vertices in each layer $V_j$. Thus $Z$ has at most $2g$ vertices in each $V_j$.

Now construct a supergraph $G''$ of $G'$ by adding a vertex $r_0$ in $f$ and some paths from $r_0$ to vertices in $Z'$. Specifically, for each vertex $v_i\in Z'$ corresponding to some vertex $v\in V(Z)$,  add to $G''$ a path $Q_{v_i}$ from $r_0$ to $v_i$ of length $1+\dist_G(r,v)$. Note that $G''$ is planar.

\begin{claim}
\label{Distancev'}
$\dist_{G''}(r_0,v') = 1 + \dist_G(r,v)$ for every vertex $v'$ in $G'$ corresponding to $v\in V(Z)$.
\end{claim}
\begin{proof}
By construction, $\dist_{G''}(r_0,v') \leq 1 + \dist_G(r,v)$, so it is sufficient to show that $\dist_{G''}(r_0,v') \geq 1 + \dist_G(r,v)$, which we now do.
Let $P$ be a shortest path from $r_0$ to $v'$ in $G''$.
By construction $P=P_1P_2$, where $P_1$ is a path from $r_0$ to $w'$ of length $1+\dist_G(r,w)$ for some vertex $w'$ in $G'$ corresponding to $w\in V(Z)$, and $P_2$ is a path in $G'$ from $w'$ to $v'$ of length $\dist_{G''}(r_0,v') - 1-\dist_G(r,w)$. By construction,
$\dist_G(v,w) \leq \dist_{G'}(v',w')  \leq  \dist_{G''}(r_0,v') - 1-\dist_G(r,w)$. Thus
$\dist_G(v,r) \leq \dist_G(v,w) + \dist_G(w,r)  \leq \dist_{G''}(r_0,v') - 1$, as desired.
\end{proof}

\begin{claim}
\label{Distancex}
$\dist_{G''}(r_0,x)=1+\dist_G(r,x)$ for each vertex $x\in V(G) \setminus V(Z)$.
\end{claim}
\begin{proof}
We first prove that $\dist_{G''}(r_0,x) \leq 1+\dist_G(r,x)$.
Let $P$ be a shortest path from $x$ to $r$ in $G$.
Let $v$ be the first vertex in $Z$ on $P$ (which is well defined since $r$ is in $Z$).
So $\dist_G(x,r) = \dist_G(x,v) + \dist_G(v,r)$.
Let $z$ be the vertex prior to $v$ on the $xv$-subpath of $P$.
Then $z$ is adjacent to some copy $v'$ of $v$ in $G'$.
In $G''$, there is a path from $r_0$ to $v'$ of length $1+\dist_G(r,v)$.
Thus $\dist_{G''}(r_0,x) \leq 1+\dist_G(r,v) + \dist_G(v,x) = 1  + \dist_G(r,x)$.

We now prove that $\dist_{G''}(r_0,x) \geq 1+\dist_G(r,x)$.
Let $P$ be a shortest path from $x$ to $r_0$ in $G''$.
Let $v'$ be the first vertex not in $G$ on $P$.
Then $v'$ corresponds to some vertex $v$ in $Z$.
Since $P$ is shortest,
$\dist_{G''}(r_0,x) = \dist_{G''}(r_0,v') + \dist_{G''}(v',x)$.
By \cref{Distancev'},
$\dist_{G''}(r_0,v') = 1 + \dist_G(r,v)$.
By the choice of $v$, the subpath of $P$ from $x$ to $v'$ corresponds to a shortest path in $G$ from $x$ to $v$.
Thus $\dist_{G''}(v',x) = \dist_{G}(v,x)$.
Combining these equalities,
$\dist_{G''}(r_0,x) = 1 + \dist_G(r,v) + \dist_{G}(v,x) \geq 1 + \dist_G(r,x)$, as desired.
\end{proof}

Let $T''$ be the following spanning tree of $G''$ rooted at $r_0$. 
Initialise $T''$ to be the union of the above-defined paths $Q_{v_i}$ taken over all vertices $v_i\in Z'$. 
Consider each edge $vw\in E(T)$ where $v\in Z$ and $w\in V(G)\setminus V(Z)$. 
Then $w$ is adjacent to exactly one vertex $v_i$ introduced when cutting through $v$. 
Add the edge $wv_i$ to $T''$. 
Finally, add the induced forest $T[ V(G) \setminus V(Z) ]$ to $T''$. 
Observe that $T''$ is a spanning tree of $G''$. 

Construct the desired graph $G^+$ by contracting $r_0$ and all its neighbours in $G''$ into a single vertex $r^+$.  
Let $T^+$ be the spanning tree of $G^+$ obtained from $T''$ by the same contraction.  
Then $G^+$ is planar because $G''$ is planar. 
By \cref{Distancex}, the BFS layering of $G^+$ from $r^+$ satisfies the conditions of the lemma. 

Every maximal vertical path in $T''$ consists of some path $Q_{v_i}$ (where $v_i\in Z'$), 
followed by some edge $v_iw$ (where $w\in V(G) \setminus V(Z)$, followed by a path in $T[ V(G) \setminus V(Z) ]$ from $w$ to a leaf in $T$. 
Since every vertical path $P$ in $T^+$ is contained in some maximal vertical path in $T''$, 
it follows that $P \cap V(G) \setminus V(Z)$ is a vertical path in $T$.
\end{proof}

We are now ready to complete the proof of \cref{GenusPartitionTreewidth9}.

\begin{proof}[Proof of \cref{GenusPartitionTreewidth9}] 
We may assume that $G$ is connected (since if each component of $G$ has the desired partition, then so does $G$).
Let $T$ be a BFS spanning tree of $G$ rooted at some vertex $r$ with corresponding 
BFS layering $(V_0,V_1,\dots)$.  
By \cref{MakePlanar}, 
there is a subgraph $Z\subseteq G$ with at most $2g$ vertices in each layer $V_i$, 
a connected planar graph $G^+$ containing $G-V(Z)$ as a subgraph, and 
a BFS spanning tree $T^+$ of $G^+$ rooted at some vertex $r^+$ 
with corresponding BFS layering $(W_0,W_1,\dots)$, 
such that  $W_{i} \cap V(G)\setminus V(Z) = V_i  \setminus V(Z)$ for all $i\geq 0$, and
$P \cap V(G) \setminus V(Z)$ is a vertical path in $T$ for every vertical path $P$ in $T^+$.

By \cref{FindVerticalPaths}, $G^+$ has a partition $\PP^+$ into vertical paths in $T^+$ 
such that $G^+ / \PP^+$ has treewidth at most $8$. 
Let $\PP:= \{P \cap V(G) \setminus V(Z) : P \in \PP^+\} \cup \{V(Z) \}$. 
Thus $\PP$ is a partition of $G$. 
Since $P \cap V(G) \setminus V(Z)$ is a vertical path in $T$ and $Z$ is a connected subgraph of $G$, 
$\PP$ is a connected partition. 
Note that the quotient $G / \PP$  is obtained from a subgraph of $G^+ / \PP^+$ by adding one vertex corresponding to $Z$. 
Since $G^+/\PP^+$ is planar and has treewidth at most 8, $G /\PP$ is apex and has treewidth at most 9. 
Thus $G / \PP$ has treewidth at most 9. 
Since $P \cap V(G) \setminus V(Z)$ is a vertical path in $T$, it has at most one vertex in each layer $V_i$. 
Thus each part of $\PP$ has at most $\max\{2g,1\}$ vertices in each layer $V_i$. 
Hence $\PP$  has layered width at most $\max\{2g,1\}$.
\end{proof}

The same proof in conjunction with \cref{Width3LayeredPartition} instead of \cref{FindVerticalPaths} shows the following.

\begin{theorem}
\label{GenusPartitionTreewidth4}
Every graph of Euler genus $g$ has a partition $\PP$ with layered width at most $\max\{2g,3\}$ such that $G / \PP$ is apex and has treewidth at most $4$. Moreover, there is such a partition for every BFS layering of $G$.
\end{theorem}

%\begin{proof}
%	We may assume that $G$ is connected (since if each component of $G$ has the desired partition, then so does $G$). 
%	Let $(V_0,V_1,\dots)$ be a given BFS layering of $G$. 
%	Let $T$ be a corresponding BFS spanning tree of $G$ rooted at vertex $r$.
%	By \cref{MakePlanar}, 
%	there is a subgraph $Z\subseteq G$ with at most $2g$ vertices in each layer $V_i$, a connected planar graph $G^+$ containing $G-V(Z)$ as a subgraph, and a BFS spanning tree $T^+$ of $G^+$ rooted at some vertex $r^+$ with corresponding BFS layering $(W_0,W_1,\dots)$, 
%	such that  $W_{i} \cap V(G)\setminus V(Z) = V_i  \setminus V(Z)$ for all $i\geq 0$.\medskip		
%	%%%%
%By \cref{Width3LayeredPartition}, $G^+$ has a partition $\PP^+$ with layered width $3$ with respect to $(W_0,W_1,\dots)$ such that $G^+ / \PP^+$ is planar and has treewidth at most $3$. For each $P\in\PP^+$, let $P^- := P \cap V(G) \setminus V(Z)$. Let $\PP:= \{P^- : P \in \PP^+\} \cup \{V(Z) \}$. Thus $\PP$ is a partition of $G$. Note that the quotient $G / \PP$  is obtained from a subgraph of $G^+ / \PP^+$ by adding one vertex corresponding to $Z$. Since $G^+/\PP^+$ is planar and has treewidth at most 3, $G /\PP$ is apex and has treewidth at most 4. 
%	For each $P\in \PP^+$, we have $P^- \cap V_i \subseteq P \cap W_{i}$, implying $|P^- \cap V_i|\leq 3$. By construction, $|Z\cap V_i|\leq 2g$. 
%	Thus each part of $\PP$ has at most $\max\{2g,3\}$ vertices in each layer $V_i$. 
%	Hence $\PP$  has layered width at most $\max\{2g,3\}$.
%\end{proof}

Note that \cref{GenusPartitionTreewidth4} is stronger than \cref{GenusPartitionTreewidth9} in that the treewidth bound is smaller, whereas 
\cref{GenusPartitionTreewidth9} is stronger than \cref{GenusPartitionTreewidth4} in that the partition is connected (and the layered width is smaller for $g\in\{0, 1\}$). Both \cref{GenusPartitionTreewidth9,GenusPartitionTreewidth4}  (with \cref{PartitionQueue}) imply that graphs with Euler genus $g$ have $O(g)$ queue-number, but better constants are obtained by the following more direct argument that uses \cref{MakePlanar,PlanarQueue} to circumvent the use of \cref{GenusPartitionTreewidth9} and obtain a proof of \cref{GenusQueue} with the best known bound.

\begin{proof}[Proof of \cref{GenusQueue} with a $4g+49$ upper bound]
Let $G$ be a graph $G$ with Euler genus $g$.
We may assume that $G$ is connected.  
Let $(V_0,V_1,\dots,V_t)$ be a BFS layering of $G$.
By \cref{MakePlanar},
there is a subgraph $Z\subseteq G$ with at most $2g$ vertices in each layer $V_i$, such that $G-V(Z)$ is planar, and there is a connected planar graph $G^+$ containing $G-V(Z)$ as a subgraph, such that there is a BFS layering $(W_0,\dots,W_t)$ of $G^+$ such that $W_i \cap V(G)\setminus V(Z) = V_i \setminus V(Z)$ for all $i\in\{0,1,\dots,t\}$.

By \cref{BFSPlanarQueue}, there is a $49$-queue layout of $G^+$ with vertex ordering $\overrightarrow{W_0}, \dots,\overrightarrow{W_t}$, where $\overrightarrow{W_i}$ is some ordering of $W_i$.
Delete the vertices of $G^+$ not in $G-V(Z)$ from this queue layout. We obtain a $49$-queue layout of $G-V(Z)$ with vertex ordering $\overrightarrow{V_0 \setminus V(Z)}, \dots,\overrightarrow{V_t\setminus V(Z)}$, where $\overrightarrow{V_i-V(Z)}$ is some ordering $V_i-V(Z)$.
Recall that $|V_j \cap V(Z)|\leq 2g$ for all $j\in\{0,1,\dots,t\}$.
Let $\overrightarrow{V_j\cap V(Z)}$ be an arbitrary ordering of $V_j\cap V(Z)$.
Let $\preceq$ be the ordering
$$\overrightarrow{V_0 \cap V(Z)}, \overrightarrow{V_0 \setminus V(Z)},\;\overrightarrow{V_1 \cap V(Z)}, \overrightarrow{V_1\setminus V(Z)}, \;\dots ,\overrightarrow{V_t \cap V(Z)}, \overrightarrow{V_t \setminus V(Z)}$$ of $V(G)$.
Edges of $G-V(Z)$ inherit their queue assignment.  We now assign edges incident with vertices in $V(Z)$ to queues.
For $i \in \{1,\dots,2g\}$ and odd $j\geq 1$, put each edge incident with the $i$-th vertex in $\overrightarrow{V_j \cap V(Z)}$ in a new queue $S_i$.
For $i \in \{1,\dots,2g\}$ and even $j\geq 0$, put each edge incident with the $i$-th vertex in $\overrightarrow{V_j \cap V(Z)}$  (not already assigned to a queue) in a new queue $T_i$.
Suppose that two edges $vw$ and $pq$ in $S_i$ are nested, where $v\prec p \prec q \prec w$. Say $v\in V_a$ and $p\in V_b$ and $q\in V_c$ and $w\in V_d$. By construction, $a\leq b\leq c\leq d$. Since $vw$ is an edge, $d\leq a+1$.
At least one endpoint of $vw$ is in $V_j\cap V(Z)$ for some odd $j$, and one endpoint of $pq$ is in $V_\ell\cap V(Z)$ for some odd $\ell$. Since $v,w,p,q$ are distinct, $j\neq \ell$. Thus $|i-j|\geq 2$. This is a contradiction since $a\leq b\leq c\leq d\leq a+1$. Thus $S_i$ is a queue. Similarly $T_i$ is a queue.
Hence this step introduces $4g$ new queues, and in total we have $4g+49$ queues.
\end{proof}

%%%%%%%%%%%%%%%%%%%%%%
\section{Proof of \cref{MinorQueue}: Excluded Minors}
\label{minor}
\label{Minors}

This section first introduces the graph minor structure theorem of Robertson and Seymour, which shows that every graph in a proper minor-closed class can be constructed using four ingredients: graphs on surfaces, vortices, apex vertices, and clique-sums. We then use this theorem to prove that every proper minor-closed class has bounded queue-number (\cref{MinorQueue}).

Let $G_0$ be a graph embedded in a surface $\Sigma$. Let $F$ be a facial cycle of $G_0$ (thought of as a subgraph of $G_0$). An  \emph{$F$-vortex} is an $F$-decomposition $(B_x\subseteq V(H):x\in V(F))$ of a graph $H$ such that $V(G_0\cap H)=V(F)$ and $x\in B_x$ for each $x\in V(F)$.  For $g,p,a,k\geq0$, a graph $G$ is \emph{$(g,p,k,a)$-almost-embeddable} if for some set $A\subseteq V(G)$ with $|A|\leq a$, there are graphs $G_0,G_1,\dots,G_s$ for some $s\in\{0,\dots,p\}$ such that:
\begin{compactitem}
\item $G-A = G_{0} \cup G_{1} \cup \cdots \cup G_s$,
\item $G_{1}, \dots, G_s$ are pairwise vertex-disjoint;
\item $G_{0}$ is embedded in a surface of Euler genus at most $g$,
\item there are $s$ pairwise vertex-disjoint facial cycles $F_1,\dots,F_s$ of $G_0$, and
\item for $i\in\{1,\dots,s\}$, there is an $F_i$-vortex $(B_x\subseteq V(G_i):x\in V(F_i))$ of  $G_i$ of width at most $k$.
\end{compactitem}
The vertices in $A$ are called \emph{apex} vertices. They can be adjacent to any vertex in $G$.

A graph is \emph{$k$-almost-embeddable} if it is  $(k,k,k,k)$-almost-embeddable.

Let $C_1=\{v_1,\dots,v_k\}$ be a $k$-clique in a graph $G_1$. Let $C_2=\{w_1,\dots,w_k\}$ be a $k$-clique in a graph $G_2$. Let $G$ be the graph obtained from the disjoint union of $G_1$ and $G_2$ by identifying $v_i$ and $w_i$ for $i\in\{1,\dots,k\}$, and possibly deleting some edges in $C_1$ ($=C_2$). Then $G$ is  a \emph{clique-sum} of $G_1$ and $G_2$. 

The following graph minor structure theorem by \citet{RS-XVI} is at the heart of  graph minor theory.

\begin{theorem}[\citep{RS-XVI}]
\label{GMST}
For every proper minor-closed class $\mathcal{G}$, there is a constant $k$ such that every graph in $\mathcal{G}$ is  obtained by clique-sums of $k$-almost-embeddable graphs. 
\end{theorem}

Every clique in a strongly $k$-almost-embeddable graph has size at most $8k$ (see~\citep[Lemma~21]{DMW17}). Thus the clique-sums in \cref{GMST} are of size in $\{0,1,\dots,8k\}$. 

We now set out to show that graphs that satisfy the ingredients of the graph minor structure theorem have bounded queue-number. First consider the case of no apex vertices.

\begin{lemma}
\label{AlmostEmbeddableStructure}
Every $(g,p,k,0)$-almost-embeddable graph $G$ has a connected partition $\PP$ with layered width at most $\max\{2g+4p-4,1\}$ such that $G / \PP$ has treewidth at most $11k+10$. 
\end{lemma}

\begin{proof}
By definition, $G=G_0\cup G_1\cup \dots\cup G_s$ for some $s\leq p$,
where $G_0$ has an embedding in a surface of Euler genus $g$ with
pairwise disjoint facial cycles $F_1,\dots,F_s$, and there is an
$F_i$-vortex $(B^i_x\subseteq V(G_i):x\in V(F_i))$ of  $G_i$ of width
at most $k$. If $s=0$ then \cref{GenusPartitionTreewidth9} implies the result. Now assume that $s\geq 1$. 

We may assume that $G_0$ is connected. Fix an arbitrary vertex $r$ in $F_1$. 
Let $G_0^+$ be the graph obtained from $G_0$ by adding an edge between $r$ and every other vertex
in $F_1\cup\dots\cup F_s$. Note that we may add $s-1$ handles, and embed
$G_0^+$ on the resulting surface. Thus $G_0^+$ has Euler genus at most
$g+2(s-1)\leq g+2p-2$.

Let $(V_0,V_1,\dots)$ be a BFS layering of $G_0^+$ rooted at $r$.
So $V_0=\{r\}$ and $V(F_1)\cup\dots\cup V(F_s) \subseteq V_0\cup V_1$.  By
\cref{GenusPartitionTreewidth9}, there is a graph $H_0$ with
treewidth at most $9$, and there is a connected $H_0$-partition $(A_x:x\in
V(H_0))$ of $G_0^+$ of layered width at most $\max\{2g+4p-4,1\}$ with
respect to $(V_0,V_1,\dots)$. Let $(C_y:y\in V(T))$ be a
tree-decomposition of $H_0$ with width at most $9$. 

Let $X := \bigcup_{i=1}^s V(G_i) \setminus V(G_0)$. Note that
$(V_0\cup X,V_1, V_2,\dots)$ is a layering of $G$ (since 
all the neighbours of vertices in $X$ are in $V_0\cup V_1\cup X$). We now add
the vertices in $X$ to the partition of $G_0^+$ to obtain the desired
partition of $G$. We add each such vertex as a singleton part.
Formally, let $H$ be the graph with $V(H):= V(H_0) \cup X $. For each
vertex  $v\in X$, let $A_v:= \{v\}$.  Initialise $E(H) := E(H_0)$. For each edge $vw$ in some vortex $G_i$, 
if $x$ and $y$ are the vertices of $H$ for which $v\in A_x$ and $w\in A_y$, then 
add the edge $xy$ to $H$. 
Now $(A_x:x\in V(H))$ is a connected $H$-partition of $G$ with width
$\max\{2g+4p-4,1\}$ with respect to $(V_0\cup X, V_1,V_2,V_3,\dots)$ (since
each new part is a singleton).

We now modify the tree-decomposition of $H_0$ to obtain the desired
tree-decomposition of $H$.
Let $(C'_y:y\in V(T))$ be the tree-decomposition of $H$ obtained from
$(C_y:y\in V(T))$ as follows.
Initialise $C'_y:=C_y$ for each $y\in V(T)$.
For $i\in\{1,\dots,s\}$ and for each vertex $u\in V(F_i)$ and for each
node $y\in V(T)$ with $u\in C_y$,
add $B^i_u$ to $C'_y$. Since $|C_y| \leq 10$ and $|B^i_u| \leq k+1$,
we have $|C'_y|\leq 11(k+1)$.
We now show that $(C'_y:y\in V(T))$ is a tree-decomposition of $H$.
Consider a vertex $v\in X$. So $v$ is in $G_i$ for some $i\in\{1,\dots,s\}$.
Let $u_1,\dots,u_t$ be the sequence of vertices in $F_i$ for which
$v\in B^i_{u_1} \cap \dots \cap B^i_{u_t}$.
Then $u_1,\dots,u_t$ is a path in $G_0$.
Say $x_j$ is the vertex of $H$ for which $u_j\in A_{x_j}$.
Let $T_j$ be the subtree of $T$ corresponding to bags that contain $x_j$.
Since $u_ju_{j+1}$ is an edge of $G_0$, either $x_j=x_{j+1}$ or
$x_jx_{j+1}$ is an edge of $H$.
In each case, by the definition of tree-decomposition, $T_j$ and
$T_{j+1}$ share a vertex in common.
Thus $T_1\cup\dots\cup T_t$ is a (connected) subtree of $T$.
By construction, $T_1\cup\dots\cup T_t$ is precisely the subtree of
$T$ corresponding to bags that contain $v$.
This show the `vertex-property' of $(C'_y:y\in V(T))$ holds.
Since each edge of $G_1\cup\dots\cup G_s$ has both its endpoints in
some bag $B^i_u$, and some bag $C'_y$ contains $B^i_u$, the
`edge-property' of $(C'_y:y\in V(T))$ also holds.
Hence $(C'_y:y\in V(T))$ is a tree-decomposition of $H$ with width at
most $11k+10$.
\end{proof}

\cref{AlmostEmbeddableStructure,PartitionQueue} imply the following result, where the edges incident to each apex vertex are put in their own queue:

\begin{lemma}
\label{AlmostEmbeddableQueue}
Every $(g,p,k,a)$-almost-embeddable graph has queue-number at most
$$a+3 \max\{2g+4p-4,1\}\, 2^{11k+10}  - \ceil*{\tfrac{3}{2}\max\{2g+4p-4,1\}}.$$
In particular, for $k\geq 1$, every $k$-almost-embeddable graph has queue-number less than 
%$$k+3 \max\{6k,1\}\, 2^{11k+10}  - \ceil*{\tfrac{3}{2}\max\{6k,1\}}$$
%$$k+18k\, 2^{11k+10}  - 9k$$
%$$18k\, 2^{11k+10}  - 8k$$
%$$k ( 18\, 2^{11k+10}  - 8)$$
%$$k ( 9\, 2^{11k+11}  - 8)$$
$9k \cdot 2^{11(k+1)}$. 
\end{lemma}

We now extend \cref{AlmostEmbeddableQueue} to allow for clique-sums using some general-purpose machinery of \citet{DMW17}. 
%First note that \cref{AlmostEmbeddableQueue,QueueTrack} imply:
%
%\begin{lemma}
%\label{AlmostEmbeddableTrack}
%For every integer $\ell\geq 0$, there is an integer $k$ such that every $\ell$-almost-embeddable graph has queue-number at most $k$.
%\end{lemma}
A tree-decomposition $(B_x\subseteq V(G):x\in V(T))$ of a graph $G$ is \emph{$k$-rich} if $B_x\cap B_y$ is a clique in $G$ on at most $k$ vertices, for each edge $xy\in E(T)$. Rich tree-decomposition are implicit in the graph minor structure theorem, as demonstrated by the following lemma, which is little more than a restatement of the graph minor structure theorem.

\begin{lemma}[\citep{DMW17}]
\label{ProduceRichDecomp}
For every proper minor-closed class $\mathcal{G}$, there are constants $k\geq 1$ and $\ell\geq 1$, such that every graph $G_0\in\mathcal{G}$ is a spanning subgraph of a graph $G$ that has a $k$-rich tree-decomposition such that each bag induces an $\ell$-almost-embeddable subgraph of $G$.
\end{lemma}

\citet{DMW17} used so-called shadow-complete layerings to establish the following result.\footnote{In~\citep{DMW17}, \cref{RichTrack}  is expressed in terms of the track-number of a graph. However, it is known that the track-number and the queue-number of a graph are tied; see \cref{Track}. So \cref{RichTrack}  also holds for queue-number. }

\begin{lemma}[\citep{DMW17}]
\label{RichTrack}
Let $G$ be a graph that has a $k$-rich tree-decomposition such that the subgraph induced by each bag has queue-number at most $c$. Then $G$ has an $f(k,c)$-queue layout for some function $f$. 
\end{lemma}

\cref{MinorQueue}, which says that every proper minor-closed class has bounded queue-number, is an immediate corollary of \cref{AlmostEmbeddableQueue,ProduceRichDecomp,RichTrack}. 

%Then \cref{TrackQueue} implies that every proper minor-closed class has bounded queue-number, as claimed in \cref{MinorQueue}.

%%%%%%%%%%%%%%%%%%
\subsection{Characterisation}
\label{Characterisation}

Bounded layered partitions are the key structure in this paper. So it is natural to ask which minor-closed classes admit bounded layered partitions. The following definition leads to the answer to this question. A graph $G$ is \emph{strongly $(g,p,k,a)$-almost-embeddable} if it is $(g,p,k,a)$-almost-embeddable and (using the notation in the definition of $(g,p,k,a)$-almost-embeddable) there is no edge between an apex vertex and a vertex in $G_0-(G_1\cup\dots\cup G_s)$. That is, each apex vertex is only adjacent to other apex vertices or vertices in the vortices. A graph is \emph{strongly $k$-almost-embeddable} if it is strongly $(k,k,k,k)$-almost-embeddable.

\cref{AlmostEmbeddableStructure} generalises as follows:

\begin{lemma}
\label{StronglyAlmostEmbeddableStructure}
Every strongly $(g,p,k,a)$-almost-embeddable graph $G$ has a connected partition $\PP$ with layered width at most $\max\{2g+4p-4,1\}$ such that $G / \PP$ has treewidth at most $11k+a+10$. 
\end{lemma}

\begin{proof}
By definition, 
$G-A=G_0\cup G_1\cup \dots\cup G_s$ for some $s\leq p$,
and for some set $A\subseteq V(G)$ of size at most $a$, 
where $G_0$ has an embedding in a surface of Euler genus $g$ with
pairwise disjoint facial cycles $F_1,\dots,F_s$, such that there is an
$F_i$-vortex $(B^i_x\subseteq V(G_i):x\in V(F_i))$ of  $G_i$ of width
at most $k$, and $N_G(v)\subseteq A \cup \bigcup_{i=1}^s V(G_i)$ for each $v\in A$. 

As proved in \cref{AlmostEmbeddableStructure}, 
$G-A$ has a connected partition $\PP$ with layered width at most $\max\{2g+4p-4,1\}$ 
with respect to some layering $(V_0,V_1, V_2,\dots)$ with $ \bigcup_{i=1}^s V(G_i) \subseteq V_0\cup V_1$, 
such that $G / \PP$ has treewidth at most $11k+10$. 
Thus $(A\cup V_0,V_1, V_2,\dots)$ is a layering of $G$. 
Add each vertex in $A$ to the partition as a singleton part. 
That is, let $\PP':= \PP \cup \{ \{v\} : v\in A\}$. 
The treewidth of $G / \PP'$ is at most the treewidth of $(G-A)/\PP$ plus $|A|$. 
Thus $\PP'$ is a connected partition with layered width at most $\max\{2g+4p-4,1\}$ with respect to 
$(A\cup V_0,V_1, V_2,\dots)$, such that $G / \PP$ has treewidth at most $11k+a+10$. 
\end{proof}

Let $C$ be a clique in a graph $G$, and let $\{C_0,C_1\}$  and $\{P_1,\dots,P_c\}$ be partitions of $C$. 
An $H$-partition $(A_x:x\in V(H))$ and layering $(V_0,V_1,\dots)$ of $G$ is 
\emph{$(C, \{C_0,C_1\}, \{P_1,\dots,P_c\})$-friendly} 
if $C_0\subseteq V_0$ and $C_1\subseteq V_1$ and there are vertices $x_1,\dots,x_c$ of $H$, such that  $A_{x_i} = P_i$ for all $i\in\{1,\dots,c\}$. 
A graph class $\mathcal{G}$ \emph{admits clique-friendly $(k,\ell)$-partitions} if for every graph $G\in \mathcal{G}$, for every clique $C$ in $G$, for all partitions $\{C_0,C_1\}$  and $\{P_1,\dots,P_c\}$ of $C$, there is a $(C, \{C_0,C_1\}, \{P_1,\dots,P_c\})$-friendly $H$-partition of $G$ with layered width at most $\ell$, such that $H$ has treewidth at most $k$. 

\begin{lemma}
\label{MakeFriendly}
Let $(A_x:x\in V(H))$ be an $H$-partition of $G$ with layered width at most $\ell$ with respect to some layering $(W_0,W_1,\dots)$ of $G$, for some graph $H$ with treewidth at most $k$. Let $C$ be a clique in $G$, and let $\{C_0,C_1\}$  and $\{P_1,\dots,P_c\}$ be partitions of $C$ such that $|C_j \cap P_i | \leq 2\ell$ for each $j\in \{0,1\}$ and $i\in \{1,\dots,c\}$. Then $G$ has a $(C, \{C_0,C_1\}, \{P_1,\dots,P_c\})$-friendly $(k+c,2\ell)$-partition.
\end{lemma}

\begin{proof}
Since $C$ is a clique, $C\subseteq W_i\cup W_{i+1}$ for some $i$. 
Let $V_j:= ( W_{i-j+1} \cup W_{i+j} ) \setminus C_0$ for $j\geq 1$. 
Let $V_0:= C_0$. 
Thus $(V_0,V_1,\dots)$ is a layering of $G$ and $C_1\subseteq V_1$. 
Let $H'$ be obtained from $H$ by adding $c$ dominant vertices $x_1,\dots,x_c$. 
Thus $H'$ has treewidth at most $k+c$. 
Let $A'_x:= A_x \setminus C$ for $x\in V(H)$. 
By construction, $|A'_x \cap V_j| \leq 2\ell$ for $x\in V(H)$ and $j\geq 0$. 
Let $A'_{x_i} := P_i$ for each $i\in\{1,\dots,c\}$. 
Thus $(A'_x:x\in V(H'))$ is a $(C, \{C_0,C_1\}, \{P_1,\dots,P_c\})$-friendly $H'$-partition of $G$ with layered width at most $2\ell$ with respect to $(V_0,V_1,\dots)$. 
\end{proof}

Every clique in a strongly $k$-almost-embeddable graph has size at most $8k$ (see~\citep[Lemma~21]{DMW17}). 
Thus \cref{StronglyAlmostEmbeddableStructure,MakeFriendly} imply:

\begin{corollary}
\label{StronglyAlmostEmbeddableCliqueFriendly}
For $k\in\mathbb{N}$, the class of strongly $k$-almost-embeddable graphs admits clique-friendly $(20k+10,12k)$-partitions. 
\end{corollary}

\begin{lemma}
\label{DoCliqueSums}
Let $\mathcal{G}$ be a class of graphs that admit clique-friendly $(k,\ell)$-partitions. Then the class of graphs obtained from clique-sums of graphs in $\mathcal{G}$ admits clique-friendly $(k,\ell)$-partitions.
\end{lemma}

\begin{proof}
Let $G$ be obtained from summing graphs $G_1$ and $G_2$ in $\mathcal{G}$ on a clique $K$. 
Let $C$ be a clique in $G$, and let $\{C_0,C_1\}$  and $\{P_1,\dots,P_c\}$ be partitions of $C$. 
Our goal is to produce a $(C, \{C_0,C_1\}, \{P_1,\dots,P_c\})$-friendly $(k,\ell)$-partition of $G$.
Without loss of generality, $C$ is in $G_1$. 
By assumption,  there is a $(C, \{C_0,C_1\}, \{P_1,\dots,P_c\})$-friendly $H_1$-partition $(A^1_x:x\in V(H_1))$ of $G_1$ with layered width  $\ell$ with respect to some layering $(V_0,V_1,\dots)$ of $G_1$, for some graph $H_1$ of treewidth at most $k$. Thus, for some vertices $x_1,\dots,x_c$ of $H$, we have $ A_{x_i} = P_i$ for all $i\in\{1,\dots,c\}$.

Since $K$ is a clique, $K\subseteq V_\kappa\cup V_{\kappa+1}$ for some $\kappa\geq 0$. 
Let $K_j := K\cap V_{\kappa+j}$ for $j\in\{0,1\}$. 
Thus $K_0,K_1$ is a partition of $K$. 
Let $y_1,\dots,y_b$ be the vertices of $H_1$ such that $A^1_{y_i} \cap K \neq\emptyset$. 
Let $Q_i := A^1_{y_i} \cap K$. 
Thus $Q_1,\dots,Q_b$ is a partition of $K$. 
By assumption,  there is a $(K, \{K_0,K_1\}, \{Q_1,\dots,Q_b\})$-friendly $H_2$-partition $(A^2_x:x\in V(H_2))$ of $G_2$ with layered width at most $\ell$ with respect to some layering $(W_0,W_1,\dots)$ of $G_2$, for some graph $H_2$ of treewidth at most $k$. Thus, for some vertices $z_1,\dots,z_b$ of $H_2$, we have $ A^2_{z_i}= Q_i$ for all $i\in\{1,\dots,b\}$.

Let $H$ be obtained from $H_1$ and $H_2$ by identifying $y_i$ and $z_i$ into $y_i$ for $i\in \{1,\dots,b\}$. 
Since $K$ is a clique, $y_1,\dots,y_b$ is a clique in $H_1$ and 
$z_1,\dots,z_b$ is a clique in $H_2$. Given tree-decompositions of $H_1$ and $H_2$ with width at most $k$, we obtain a  
tree-decomposition of $H$ by simply adding an edge between a bag that contains $y_1,\dots,y_b$ and a bag that contains $z_1,\dots,z_b$. Thus $H$ has treewidth at most $k$. 

Let $X_a := V_a \cup W_{a-\kappa}$ for $a\geq 0$ (where $W_{a-\kappa}=\emptyset$ if $a-\kappa<0$). 
Then $(X_0,X_1,\dots)$ is a layering of $G$, since $K_0\subseteq V_\kappa \cap W_0$ and 
$K_1\subseteq V_{\kappa+1}\cap W_1$. By construction, $C_0\subseteq V_0 \subseteq X_0$ and  $C_1\subseteq V_1 \subseteq X_1$, as desired. 

For $x\in V(H_1)$, let $A_x:= A^1_x$. 
For $x\in V(H_2)\setminus\{z_1,\dots,z_b\}$, let $A_x:= A^2_x$. 
For $i\in\{1,\dots,b\}$, we have $A^2_{z_i} = Q_i \subseteq A^1_{y_i}$. 
Thus $(A_x:x\in V(H))$ is an $H$-partition of $G$ with layered width at most $\ell$ with respect to 
$(X_0,X_1,\dots)$. Moreover, since $(A^1_x:x\in V(H_1))$ is 
 $(C, \{C_0,C_1\}, \{P_1,\dots,P_c\})$-friendly with respect to $(V_0,V_1,\dots)$, and 
 $V_i \subseteq X_i$, the partition 
 $(A_x:x\in V(H))$ is 
 $(C, \{C_0,C_1\}, \{P_1,\dots,P_c\})$-friendly with respect to $(X_0,X_1,\dots)$. 
 \end{proof}

The following is the main result of this section. See~\citep{Eppstein-Algo00,DH-SODA04,DMW17} for the definition of (linear) local treewidth.

\begin{theorem}
\label{Equivalent}
The following are equivalent for a minor-closed class of graphs $\mathcal{G}$:
\begin{compactenum}[(1)]
%%%%%%
\item there exist $k,\ell\in\mathbb{N}$ such that every graph $G \in \mathcal{G}$ has a partition $\PP$ with layered width at most $\ell$, such that $G / \PP$ has treewidth at most $k$. 
%%%%%%
\item there exists $k\in\mathbb{N}$ such that every graph $G \in \mathcal{G}$ has a partition $\PP$ with layered width at most $1$, such that $G / \PP$ has treewidth at most $k$. 
%%%%%%
\item there exists $k\in\mathbb{N}$ such that every graph in $\mathcal{G}$ has layered treewidth at most $k$,
%%%%%%
\item $\mathcal{G}$ has linear local treewidth,
%%%%%%
\item $\mathcal{G}$ has bounded local treewidth,
%%%%%%
\item there exists an apex graph not in $\mathcal{G}$, 
%%%%%%
\item there exists $k\in\mathbb{N}$ such that every graph in $\mathcal{G}$ is obtained from clique-sums of strongly $k$-almost-embeddable graphs.
\end{compactenum}
\end{theorem}

\begin{proof}
\cref{MakeWidth1} says that (1) implies (2). 
\cref{PartitionLayeredTreewidth} says that (2) implies (3).
\citet{DMW17} proved that (3) implies (4), which implies (5) by definition.
\citet{Eppstein-Algo00}  proved that (5) and (6) are equivalent; see~\citep{DH-Algo04} for an alternative proof.
\citet{DvoTho} proved  that (6) implies (7); see \cref{DvoTho} below. 
\cref{DoCliqueSums,StronglyAlmostEmbeddableCliqueFriendly} imply that every graph obtained from clique-sums of strongly $k$-almost-embeddable graphs has a partition of layered width $12k$ such that the quotient has treewidth at most $20k+10$.  This says that  (7) implies (1). 
\end{proof}

Several notes about \cref{Equivalent} are in order:
\begin{compactitem}
\item \citet{DH-SODA04} previously proved that (4) and (5) are equivalent. 
\item While the partitions $\mathcal{P}$ for strongly $k$-almost-embeddable
graphs provided by \cref{StronglyAlmostEmbeddableStructure} are connected, the partitions
$\mathcal{P}$ in \cref{Equivalent} are no longer guaranteed to be connected.
\item The assumption of a minor-closed  class in \cref{Equivalent} is essential: \citet{DEW17} proved that the $n\times n \times n$ grid $G_n$ has bounded local treewidth but has unbounded, indeed $\Omega(n)$, layered treewidth. By \cref{PartitionLayeredTreewidth}, if $G_n$ has a partition with layered width $\ell$ such that the quotient has treewidth at most $k$, then $k\ell \geq\Omega(n)$. 
%That said, it is open whether (1), (2) and (3) are equivalent in a subgraph-closed class. 
\end{compactitem}

The above proof that (6) implies (7) employed a structure theorem for apex-minor-free graphs by \citet{DvoTho}. \citet{DvoTho} actually proved the following strengthening of the graph minor structure theorem. For a graph $X$ and a surface $\Sigma$, let $a(X,\Sigma)$ be the minimum size of a set $S\subseteq V(X)$, such that $X-S$ can be embedded in $\Sigma$. Let $a(X):=a(X,\mathbb{S}_0)$ where $\mathbb{S}_0$ is the sphere. Note that $a(X)=1$ for every apex graph. 

\begin{theorem}[\citep{DvoTho}]
\label{DvoTho}
For every graph $X$, there are integers  $p,k,a$, such that every $X$-minor-free graph $G$ is a clique-sum of graphs $G_1, G_2, \dots, G_n$ such that for $i\in\{1,\dots,n\}$ there exists a surface $\Sigma_i$ and a set $A_i \subseteq V(G_i)$ satisfying the following:
\begin{compactitem}
\item $|A_i| \leq a$,
\item $X$ cannot be embedded in $\Sigma_i$, 
\item $G_i - A_i$ can be almost embedded in $\Sigma_i$ with at most $p$ vortices of width at most $k$,
%¥ every triangle in the embedding bounds a 2-cell face, and
\item all but at most $a(X, \Sigma_i) - 1$ vertices of $A_i$ are only adjacent in $G_i$ to vertices contained either in $A_i$ or in the vortices.
\end{compactitem}
\end{theorem}

\cref{DvoTho} leads to the following result of interest. 

\begin{theorem}
\label{MinorStructure}
For every graph $X$ there is an integer $k$ such that every $X$-minor-free graph $G$ can be obtained from clique-sums of graphs $G_1,G_2,\dots,G_n$ such that for $i\in\{1,2,\dots,n\}$ there is a set $A_i\subseteq V(G_i)$ of size at most $\max\{a(X)-1,0\}$ such that $G_i-A_i$ has a connected partition $\PP_i$ with layered width at most $1$, such that $(G_i-A_i) / \PP_i$ has treewidth at most $k$. 
\end{theorem}

\begin{proof}
In \cref{DvoTho}, since $X$ cannot be embedded in $\Sigma_i$, there is an integer $g$ depending only on $X$ such that $\Sigma_i$ has Euler genus at most $g$. Thus each graph $G_i$ has a set $A_i$ of at most $\max\{a(X, \Sigma_i) - 1,0\} \leq \max\{a(X)-1,0\}$ vertices, such that $G_i-A_i$ is strongly $(g,p,k,a)$-almost-embeddable.  By \cref{StronglyAlmostEmbeddableStructure},  $G_i-A_i$ has a connected partition $\PP_i$ with layered width at most $\max\{2g+4p-4,1\}$, such that $(G_i-A_i) / \PP_i$ has treewidth at most $11k+a+10$. The result follows from \cref{MakeWidth1}.
\end{proof}

%%%%%%%%%%%%%%%
\section{Strong Products}
\label{Products}

This section provides an alternative and helpful perspective on layered partitions. The \emph{strong product} of graphs $A$ and $B$, denoted by $A\boxtimes B$, is the graph with vertex set $V(A)\times V(B)$, where distinct vertices $(v,x),(w,y)\in V(A)\times V(B)$ are adjacent if: 
\begin{compactitem}
\item  $v=w$ and $xy\in E(B)$, or 
\item  $x=y$ and $vw\in E(A)$, or  
\item  $vw\in E(A)$ and $xy\in E(B)$. 
\end{compactitem}
The next observation follows immediately from the definitions. 

\begin{obs}
\label{PartitionProduct}
For every graph $H$, a graph $G$ has an $H$-partition of layered width at most $\ell$ if and only if $G$ is a subgraph of 
$H \boxtimes P \boxtimes K_\ell$ for some path $P$.
\end{obs}

Note that a general result about the queue-number of strong products by \citet{Wood-Queue-DMTCS05} implies that $\qn(H \boxtimes P ) \leq 3 \qn(H)+1$. 
%In fact, the following more general bound is given:
%$$\qn(H \boxtimes  J) \leq 2\sqn(J) \qn(H)+ \sqn(J)+ \qn(H) .$$
%Here $\sqn(J)$ is the `strict' queue-number of $J$, where in a strict queue layout, edges $uv$ and $vw$ with $u\prec w \prec v$ or $v \prec u \prec w$ are in distinct queues. Note that a path has $\sqn(P)=1$. Also note that
%$\sqn(J)\geq\Delta(J)/2$, so it only makes sense to use $J$ with bounded degree.
\cref{Blowup} and the fact that $\qn(K_\ell)=\floor{\frac{\ell}{2}}$ 
implies that $\qn( Q \boxtimes K_\ell ) \leq \ell \cdot \qn(Q) + \floor{\frac{\ell}{2}}$. 
Together these results say that 
$\qn( H \boxtimes P \boxtimes K_\ell  ) \leq \ell(3 \qn(H)+1) +\floor{\frac{\ell}{2}}$, which is equivalent to \cref{PartitionQueue}. 

Several papers in the literature study minors in graph products~\citep{QYY10,Kotlov01,CKR08,Wood-ProductMinor,KOY14,Wood12}. 
The results in this section are complementary: they show that every graph in certain minor-closed classes is a subgraph of a particular graph product, such as a subgraph of $H \boxtimes P$ for some bounded treewidth graph $H$ and path $P$. First note that \cref{PartitionProduct,Width1LayeredPartition,Width3LayeredPartition} imply the following result conjectured by \citet{WoodBanff08}.\footnote{To be precise, \citet{WoodBanff08} conjectured that for every planar graph $G$ there are graphs $X$ and $Y$, such that both $X$ and $Y$ have bounded treewidth, $Y$ has bounded maximum degree, and $G$ is a minor of $X\boxtimes   Y$, such that the preimage of each vertex of $G$ has bounded radius in $X\boxtimes  Y$. \cref{PlanarProduct}(a)  is stronger than this conjecture since it has a subgraph rather than a shallow minor, and $Y$ is a path. }

\begin{theorem}
\label{PlanarProduct}
Every planar graph is a subgraph of:
\begin{compactenum}[(a)]
\item $H \boxtimes P$ for some planar graph $H$ with treewidth at most $8$ and some path $P$.
\item $H \boxtimes P \boxtimes K_3$ for some planar graph $H$ with treewidth at most $3$ and some path $P$.
\end{compactenum}
\end{theorem}

\cref{PlanarProduct} generalises for graphs of bounded Euler genus as follows. Let $A \join B$ be the \emph{complete join} of graphs $A$ and $B$. That is, take disjoint copies of $A$ and $B$, and add an edge between each vertex in $A$ and each vertex in $B$. 

\begin{theorem}
\label{GenusProduct}
Every graph of Euler genus $g$ is a subgraph of:
\begin{compactenum}[(a)]
\item  $H \boxtimes P \boxtimes K_{\max\{2g,1\}}$ for some apex graph $H$ of treewidth at most $9$ and for some path $P$.
\item  $H \boxtimes P \boxtimes K_{\max\{2g,3\}}$ for some apex graph $H$ of treewidth at most $4$ and for some path $P$.
\item $(K_{2g} \join H )  \boxtimes P$ for some planar graph $H$ of treewidth at most $8$  and some path $P$. 
\end{compactenum}
\end{theorem}

\begin{proof}
Parts (a) and (b) follow from \cref{PartitionProduct,GenusPartitionTreewidth4,GenusPartitionTreewidth9}. 
It remains to prove (c). 
We may assume that $G$ is edge-maximal with Euler genus $g\geq 1$, and is thus connected.
Let $(V_0,V_1,\dots)$ be a BFS layering of $G$. 
By \cref{MakePlanar},
there is a subgraph $Z\subseteq G$ with at most $2g$ vertices in each layer $V_i$, such that $G-V(Z)$ is planar, and there is a connected planar graph $G^+$ containing $G-V(Z)$ as a subgraph, such that there is a BFS layering $(W_0,W_1,\dots)$ of $G^+$ such that $W_i \cap V(G) \setminus V(Z) = V_i \setminus V(Z)$ for all $i\ge 0$.

By \cref{Width1LayeredPartition}, there is a planar graph $H$ with treewidth at most $8$, such that $G^+$
has an $H$-partition $(A_x:x\in V(H))$ of layered width $1$ with respect to $(W_0,\dots,W_n)$.
Let $A'_x := A_x \cap V(G) \setminus V(Z)$ for each $x \in V(H)$.
Thus $(A'_x  : x \in V(H) )$ is an $H$-partition of $G-V(Z)$ of layered width $1$ with respect to $(V_0\setminus V(Z),V_1\setminus V(Z),\dots)$ (since $W_i \cap V(G) \setminus V(Z) = V_i \setminus V(Z)$).

Let $z_1,\dots,z_{2g}$ be the vertices of a complete graph $K_{2g}$. 
Say $v_{i,1},\dots,v_{i,2g}$ are the vertices in $V(Z) \cap V_i$ for $i\ge 0$.
(Here some $v_{i,j}$ might be undefined.)\
Define $A'_{z_j} := \{ v_{i,j} : i\ge 0\} \}$.
Now, $(A'_x  : x\in V(H \join K_{2g} ))$ is an $(H \join K_{2g} )$-partition of $G$ of layered width $1$, which is equivalent to the claimed result by \cref{PartitionProduct}. 
\end{proof}

Note that in Theorems~\ref{PlanarProduct}(a) and \ref{GenusProduct}(a), the graph $H$ is a minor of the given graph (because the corresponding partition in connected). But we cannot make this conclusion in Theorems~\ref{PlanarProduct}(b), \ref{GenusProduct}(b) and \ref{GenusProduct}(c).

These results are generalised for $(g,p,k,a)$-almost-embeddable graphs as follows. 

\begin{theorem}
\label{AlmostEmbeddableStructureRevised}
Every $(g,p,k,a)$-almost-embeddable graph is a subgraph of:
\begin{compactenum}[(a)]
\item $ (H\boxtimes P \boxtimes K_{\max\{2g+4p,1\}} ) \join K_a$ for some graph $H$ with treewidth at most $11k+10$  and some path $P$, 
\item $( (H \join  K_{(2g+4p)(k+1)} )  \boxtimes P ) \join K_a$ for some graph $H$ with treewidth at most $9k+8$ and some path $P$. 
\end{compactenum}
\end{theorem}

\begin{proof}
	\cref{AlmostEmbeddableStructure,PartitionProduct} imply (a). It remains to prove (b). 
	Let $G$ be a $(g,p,k,a)$-almost-embeddable graph. 
	We use the notation from the definition of $(g,p,k,a)$-almost-embeddable. 
	In the proof of \cref{AlmostEmbeddableStructure}, since 
	$G_0^+$ has Euler genus at most $g+2p$, by \cref{GenusProduct}(c) 
	there is a graph $H_0$ with treewidth at most $8$, 
	such that $G_0^+ \subseteq ( H_0 \join K_{2g+4p}  ) \boxtimes P$. 
	That is, 
	$G_0^+$ has an $(H_0 \join K_{2g+4p} )$-partition of layered width $1$. 
	Apply the proof in \cref{AlmostEmbeddableStructure} to obtain 
	a graph $H$ with treewidth at most $9k+8$, such that $G$ has an $(H \join K_{(2g+4p)(k+1)})$-partition of layered width $1$. 
	That is,  $G \subseteq (H \join K_{(2g+4p)(k+1)} ) \boxtimes P$. 
	Adding apex vertices, every $(g,p,k,a)$-almost-embeddable graph is a subgraph of 
	$( (H \join K_{(2g+4p)(k+1)})  \boxtimes P ) \join K_a$ 
	for some graph $H$ with treewidth at most $9k+8$.
\end{proof}

%\begin{proof}
%\cref{AlmostEmbeddableStructure,PartitionProduct} imply (a). It remains to prove (b). 
%Let $G$ be a $(g,p,k,a)$-almost-embeddable graph. 
%We use the notation from the definition of $(g,p,k,a)$-almost-embeddable. 
%In the proof of \cref{AlmostEmbeddableStructure}, since 
%$G_0^+$ has Euler genus at most $g+2p$, by \cref{GenusPartitionLayeredWidth1}
%there is a graph $H_0$ with treewidth at most $8$, 
%such that $G_0^+$ has an $(H_0 \join (K_{2g+4p} )$-partition of layered width $1$. 
%That is,  $G_0^+ \subseteq ( H_0 \join K_{2g+4p}  ) \boxtimes P$. 
%Apply the proof in \cref{AlmostEmbeddableStructure} to obtain 
%a graph $H$ with treewidth at most $9k+8$, such that $G$ has an $(H \join K_{(2g+4p)(k+1)})$-partition of layered width $1$. That is,  $G \subseteq (H \join K_{(2g+4p)(k+1)} ) \boxtimes P$. 
%Adding apex vertices, every $(g,p,k,a)$-almost-embeddable graph is a subgraph of 
%$( (H \join K_{(2g+4p)(k+1)})  \boxtimes P ) \join K_a$ 
%for some graph $H$ with treewidth at most $9k+8$.
%\end{proof}

\begin{corollary}
\label{kAlmostEmbeddableStructure}
For $k\geq 1$ every $k$-almost-embeddable graph is a subgraph of:
\begin{compactenum}[(a)]
\item $ (H\boxtimes P \boxtimes K_{6k} ) \join K_k$ for some graph $H$ with treewidth at most $11k+10$  and some path $P$, 
\item $( (H \join  K_{(6k)(k+1)} )  \boxtimes P ) \join K_k$ for some graph $H$ with treewidth at most $9k+8$ and some path $P$. 
\end{compactenum}
\end{corollary}

\cref{Equivalent,PartitionProduct} imply:

\begin{corollary}
	\label{ApexMinorFreeProduct}
For every apex graph $X$ there exists $c\in\mathbb{N}$ such that every $X$-minor-free graph is a subgraph of $H\boxtimes P$ for some graph $H$ with treewidth at most $c$ and for some path $P$. 
\end{corollary}

\cref{GMST,AlmostEmbeddableStructureRevised} imply the following result for any proper minor-closed class.

\begin{theorem}
\label{MinorProduct}
For every proper minor-closed class $\mathcal{G}$ there are integers $k$ and $a$ such that every graph $G\in \mathcal{G}$ can be obtained by clique-sums of graphs $G_1,\dots,G_n$ such that  for $i\in\{1,\dots,n\}$, 
$$G_i \subseteq (H_i  \boxtimes P_i ) \join K_a,$$ 
for some graph $H_i$ with treewidth at most $k$ and some path $P_i$. 
\end{theorem}

\cref{MinorStructure,PartitionProduct} imply the following precise bound on $a$ for $X$-minor-free graphs. 

\begin{theorem}
\label{XMinorFree}
For every graph $X$ there is an integer $k$ such that every $X$-minor-free graph $G$ 
can be obtained by clique-sums of graphs $G_1,\dots,G_n$ such that  for $i\in\{1,\dots,n\}$, 
$$G_i \subseteq (H_i  \boxtimes P_i ) \join K_{\max\{a(X)-1,0\}},$$ 
for some graph $H_i$ with treewidth at most $k$ and some path $P_i$. 
\end{theorem}

Note that it is easily seen that in all of the above results, the graph $H$ and the path $P$ have at most $|V(G)|$ vertices. 

We can interpret these results as saying that strong products and complete joins form universal graphs for the above classes. For all $n$ and $k$ there is a graph $H_{n,k}$ with treewidth $k$ that contains every  graph with $n$ vertices and treewidth $k$ as a subgraph (for example, take the disjoint union of all such graphs). The proof of \cref{PlanarProduct} then shows that $H_{n,8}\boxtimes P_n$ contains every planar graph with $n$ vertices. There is a substantial literature on universal graphs for planar graphs and other classes~\citep{BKNRS-DM98,BCLR89,BPTW10,BCEGS82,AC07,AA02}. For example, \citet*{BCEGS82} constructed a graph on $O(n^{3/2})$ edges that contains every planar graph on $n$ vertices as a subgraph. While  $H_{n,8}\boxtimes P_n$ contains much more than $O(n^{3/2})$ edges, it has the advantage of being highly structured and with bounded average degree. Taking this argument one step further, there is an infinite graph $\mathcal{T}_k$ with treewidth $k$ that contains every (finite) graph with treewidth $k$ as a subgraph. Similarly, the infinite path $\mathcal{Q}$ contains every (finite) path as a subgraph. Thus our results imply that $\mathcal{T}_8\boxtimes \mathcal{Q}$ contains every planar graph. Analogous statements can be made for the other classes above.

%%%%%%%%%%%%%%%
\section{Non-Minor-Closed Classes}\label{non-minor-closed}

This section gives three examples of non-minor-closed classes of graphs that have bounded queue-number. The following lemma  will be helpful.

\begin{lemma}
\label{GeneralisedSubdivision}
Let $G_0$ be a graph with a $k$-queue layout. Fix integers $c\geq 1$ and $\Delta\geq 2$.
Let $G$ be the graph with $V(G):=V(G_0)$ where $vw\in E(G)$ whenever there is a $vw$-path $P$ in $G_0$ of length at most $c$, such that every internal vertex on $P$ has degree at most $\Delta$.
Then $$\qn(G) < 2(2k(\Delta+1))^{c+1}.$$
\end{lemma}

\begin{proof}
Consider a $k$-queue layout of $G_0$.
Let $\preceq$ be the corresponding vertex ordering and let $E_1,\dots,E_k$ be the partition of $E(G_0)$ into queues with respect to $\preceq$.

For each edge $xy\in E_i$, let $q(xy):=i$. For distinct vertices $a,b\in V(G_0)$, let $f(a,b):=1$ if $a\prec b$ and let $f(a,b):=-1$ if $b\prec a$.
For $\ell\in\{1,\dots,c\}$, let $X_\ell$ be the set of edges $vw\in E(G)$ such that the corresponding $vw$-path $P$ in $G_0$ has length exactly $\ell$.
We will use distinct sets of queues for the $X_\ell$ in our queue layout of $G$.

By Vizing's Theorem, there is an edge-colouring $h$ of $G$ with $\Delta+1$ colours, such that any two edges incident with a vertex of degree at most $\Delta$ receive distinct colours. (Edges incident with a vertex of degree greater than $\Delta$ can be assigned the same colour.) 

Consider an edge $vw$ in $X_\ell$ with $v \prec w$. Say $(v=x_0,x_1,\dots, x_\ell, x_{\ell+1}=w)$ is the corresponding path in $G_0$. Let
\begin{align*}
f(vw) & :=( f(x_0,x_1), \dots, f( x_\ell, x_{\ell+1}) )\\
q(vw) & := ( q(x_0,x_1),  \dots, q(x_{\ell},x_{\ell+1}) )\\
h(vw) & := ( h(x_0,x_1),  \dots, h(x_{\ell},x_{\ell+1}) ).
\end{align*}
Consider  edges $vw,pq\in X_\ell$ with $v,w,p,q$ distinct and $f(vw)=f(pq)$ and $g(vw)=g(pq)$  and $h(vw)=h(pq)$.
Assume $v\prec p$. Say $(v=x_0,x_1,\dots, x_\ell, x_{\ell+1}=w)$ and $(p=y_0,y_1,\dots, x_\ell, x_{\ell+1}=q)$ are the paths respectively corresponding to $vw$ and $pq$ in $G_0$.
Thus $f(x_i,x_{i+1})=f(y_i,y_{i+1})$ and $q(x_ix_{i+1})=q(y_iy_{i+1})$  and $h(x_ix_{i+1})=h(y_iy_{i+1})$ for $i\in\{0,1,\dots,\ell\}$.
Thus $x_ix_{i+1}$ and $y_iy_{i+1}$ are not nested.
Since $v=x_0\prec y_0 = p$, we have $x_1 \preceq y_1$.
Since $h(x_0x_1)=h(y_0y_1)$ and both $x_1$ and $y_1$ have degree at most $\Delta$ in $G_0$, we have $x_1 \prec y_1$.
It follows by induction that $x_i \prec y_i$ for $i\in\{0,1,\dots,\ell+1\}$, where in the last step we use the assumption that $w\neq q$.
In particular, $w=x_{\ell+1} \prec y_{\ell+1}=q$.  Thus $vw$ and $pq$ are not nested. There are $2^{\ell+1}$ values for $f$, and $k^{\ell+1}$ values for $q$, and $(\Delta+1)^{\ell+1}$ values for $h$.
Thus $(2k(\Delta+1))^{\ell+1}$ queues suffice for $X_\ell$. The total number of queues is
$\sum_{\ell=1}^c(2k(\Delta+1))^{\ell+1}< 2(2k(\Delta+1))^{c+1}$.
\end{proof}

\subsection{Allowing Crossings}

Our result for graphs of bounded Euler genus generalises to allow for a bounded number of crossings per edge. A graph is \emph{$(g,k)$-planar} if it has a drawing in a surface of Euler genus $g$ with at most $k$ crossings per edge and with no three edges crossing at the same point. A $(0,k)$-planar graph is called \emph{$k$-planar}; see~\citep{KLM17} for a survey about $1$-planar graphs. Even in the simplest case, there are $1$-planar graphs that contain arbitrarily large complete graph minors~\citep{DEW17}. Nevertheless, such graphs have bounded queue-number.

\begin{prop}
\label{gkPlanar}
Every $(g,k)$-planar graph $G$ has queue-number at most $2(40g+490)^{k+2}$.
\end{prop}

\begin{proof}
Let $G_0$ be the graph obtained from $G$ by replacing each crossing point by a vertex.
Thus $G_0$ has Euler genus at most $g$, and thus has queue-number at most $4g+49$ by \cref{GenusQueue}.
Note that for every edge $vw$ in $G$ there is a $vw$-path $P$ in $G_0$ of length at most $k+1$, such that every internal vertex has degree 4.
The result follows from \cref{GeneralisedSubdivision} with $c=k+1$ and $\Delta=4$.
\end{proof}

\cref{gkPlanar} can also be concluded from a result of \citet{DujWoo05} in conjunction with \cref{GenusQueue}.

%%%%%%%%%%%%%%%%%%
\subsection{Map Graphs}

Map graphs are defined as follows. Start with a graph $G_0$ embedded in a surface of Euler genus $g$, with each face labelled a `nation' or a `lake', where each vertex of $G_0$ is incident with at most $d$ nations. Let $G$ be the graph whose vertices are the nations of $G_0$, where two vertices are adjacent in $G$ if the corresponding faces in $G_0$ share a vertex. Then $G$ is called a \emph{$(g,d)$-map graph}. A $(0,d)$-map graph is called a (plane) \emph{$d$-map graph}; such graphs have been extensively studied~\citep{FLS-SODA12,DFHT05,CGP02,Chen07,Chen01}. The $(g,3)$-map graphs are precisely the graphs of Euler genus at most $g$ (see~\citep{DEW17}). So $(g,d)$-map graphs provide a natural generalisation of graphs embedded in a surface.

\begin{prop}
\label{MapGraph}
Every $(g,d)$-map graph $G$ has queue-number at most $2\big(8g+98)(d+1)\big)^{3}$.
\end{prop}

\begin{proof}
It is known that $G$ is the half-square of a bipartite graph $G_0$ with Euler genus $g$ (see~\citep{DEW17}).
This means that $G_0$ has a bipartition $\{A,B\}$, such that every vertex in $B$ has degree at most $k$,
$V(G)=A$, and for every edge $vw\in E(G)$, there is a common neighbour of $v$ and $w$ in $B$. 
By \cref{GenusQueue}, $G_0$ has a $(4g+49)$-queue layout. 
The result follows from \cref{GeneralisedSubdivision} with $c=2$ and $\Delta=d$.
\end{proof}

%%%%%%%
\subsection{String Graphs}
\label{StringGraphs}

A \emph{string graph} is the intersection graph of a set of curves in the plane with no three curves meeting at a single point \cite{PachToth-DCG02,SS-JCSS04,SSS-JCSS03,Krat-JCTB91,FP10,FP14}. For an integer $k\geq 2$, if each curve is in at most $k$ intersections with other curves, then the corresponding string graph is called a \emph{$k$-string graph}. A \emph{$(g,k)$-string} graph is defined analogously for curves on a surface of Euler genus at most $g$. 

\begin{prop} 
\label{StringGraph}
For all integers $g\geq 0$ and $k\geq 2$, every $(g,k)$-string graph has queue-number at most $2(40g+490)^{2k+1}$.
\end{prop}

\begin{proof}
We may assume that in the representation of $G$, no curve is self-intersecting,
no three curves intersect at a common point, and
no two curves intersect at an endpoint of one of the curves.
Let $G_0$ be the graph obtained by adding a vertex at the intersection point of any two distinct curves, and at the endpoints of each curve. Each section of a curve between two such vertices becomes an edge in $G_0$.  So $G_0$ is embedded without crossings and has Euler genus at most $g$. Associate each vertex $v$ of $G$ with a vertex $v_0$ of $G_0$ at the endpoint of the curve representing $v$. For each edge $vw$ of $G$, there is $v_0w_0$-path in $G_0$ of length at most $2k$, such that every internal vertex on $P$ has degree at most 4. 
By \cref{GenusQueue}, $G_0$ has a $(4g+49)$-queue layout. 
The result then follows from  \cref{GeneralisedSubdivision} with $\Delta=4$ and $c=2k$. 
\end{proof}

%%%%%%%%%%%%%%%%%%%
\section{Applications and Connections}\label{other}

In this section, we show that layered partitions lead to a simple proof of a known result about low treewidth colourings, and we discuss implications of our results such as resolving open problems about 3-dimensional graph drawings.

\subsection{Low Treewidth Colourings}

\citet{DDOSRSV04} proved that every graph in a proper minor-closed class can be edge 2-coloured so that each monochromatic subgraph has bounded treewidth, and more generally, that for fixed $c\geq 2$, every such graph can be edge $c$-coloured such that the union of any $c-1$ colour classes has bounded treewidth. They also showed analogous vertex-colouring results. (Of course, in both cases, by a colouring we mean a non-proper colouring). Here we show that these results can be easily  proved using layered partitions. The reader should not confuse this result with a different result by \citet{DDOSRSV04} that has subsequently been generalised for any bounded expansion class by \citet{Sparsity}. 

\begin{lemma}
\label{cEdgeColourAlmost}
For every $k$-almost-embeddable graph $G$ and integer $c\geq 2$, there are induced subgraphs $G_1,\dots,G_c$ of $G$, such that $G=\bigcup_{j=1}^c G_j$, 
and for $j\in\{1,\dots,c\}$ if 
$$X_j := G_1 \cup \dots \cup G_{j-1} \cup G_{j+1} \cup\dots \cup G_c,$$
then $X_j$ is an induced subgraph of $G$ and $X_j$ has a  tree-decomposition $(B^j_x:x\in V(T_j))$ of width at most $66k(k+1)(2c-1)+k-1$.
%, such that  for every clique $C$ of $G$, %$C\cap V(X_j)$ is a clique of $X_j$ and thus 
%$C\cap V(X_j)$ is a subset of some bag $B^j_x$.
\end{lemma}

\begin{proof}
Note that we allow $G_i$ and $G_j$ to have vertices and edges in common. 
Let $A$ be the set of apex vertices in $G$ (as described in the definition of $k$-almost-embeddable). Thus $|A|\leq k$. 
By \cref{AlmostEmbeddableStructure}, $G-A$ has an $H$-partition $(Z_h:h\in V(H))$ of layered width at most $6k$, for some graph $H$ with treewidth at most $11k+10$. 
Let $(V_0,V_1,\dots)$ be the corresponding layering of $G-A$. Let $V_i:=\emptyset$ if $i<0$.
For $j\in\{1,\dots,c\}$, let 
$$G_j := G\Big[ A\cup \bigcup_{i\geq 0}  V_{2ci+2j-2} \cup V_{2ci+2j-1} \cup V_{2ci+2j} \Big].$$ 
Note that $G=\bigcup_{j=1}^c G_j$, as claimed. 
For $i\in\mathbb{Z}$ and $j\in\{1,\dots,c\}$, let 
$$X_{i,j} := G[V_{2ci+2j} \cup V_{2ci+2j+1} \cup \dots \cup V_{2c(i+1)+2j-2} ].$$
Note that  $X_j$ is the induced subgraph $G[ \bigcup_{i\in\mathbb{Z}} V( X_{i,j} ) \cup A ]$. 

Let $(H_x:x \in V(T))$ be a tree-decomposition of $H$ in which every bag has size at most $11(k+1)$. 
For $i\in\mathbb{Z}$ and $j\in\{1,\dots,c\}$, let $T_{i,j}$ be a copy of $T$, and 
let $D_x:= \bigcup_{h\in H_x} Z_h \cap V(X_{i,j})$ for each node $x\in V(T_{i,j})$. 
Then $(D_x:x\in V(T_{i,j}))$ is a tree-decomposition of $X_{i,j}$ because: 
(1) each vertex $v$ of $X_{i,j}$ is in one part $Z_h$ of our $H$-partition, and thus $v$ is in precisely those bags corresponding to nodes $x$ of $T$ for which $h\in H_x$, which form a subtree of $T$; and 
(2) for each edge $vv'$ of $X_{i,j}$, $v$ is in one part $Z_h$ and $v'$ is in one part $Z_{h'}$ of our $H$-partition, and  thus $h=h'$ or $hh'\in E(H)$, implying that $h$ and $h'$ are in a common bag $H_x$, and thus $v$ and $v'$ are in a common bag $D_x$. 
Since $X_{i,j}$ consists of $2c-1$ layers, and our $H$-partition has layered width at most $6k$, 
we have $|Z_h \cap V(X_{i,j}) | \leq 6k( 2c-1 )$. 
Thus $(D_x:x\in V(T_{i,j}))$ has width at most $66k(k+1)(2c-1)-1$.  

For $j\in\{1,\dots,c\}$, let  $(B^j_x:x\in V(T_j))$ be the tree-decomposition of $X_j$ obtained as follows: 
First, let $T_j$ be the tree obtained from the disjoint union $\bigcup_{i\in\mathbb{Z}} T_{i,j}$ 
by adding an edge between $T_{i,j}$ and $T_{i+1,j}$ for all $i\in\mathbb{Z}$. 
Then for each node $x$ of $T_j$, let $B^j_x:= D_x \cup A$, where $D_x$ is the bag corresponding to $x$ in the tree-decomposition of $X_{i,j}$ where $i$ is such that $x\in V(T_{i,j})$. 
Since  $X_{i,j}$ and $X_{i',j}$ are disjoint for $i\neq i'$, and $A$ is a subset of every bag, 
$(B^j_x:x\in V(T_j))$ is a tree-decomposition of $X_j$ with width at most $66k(k+1)(2c-1)+k-1$.  
%
%A key property of the above construction is that any two consecutive layers intersect $X_j$ in at most one $X_{i,j}$. 
%More precisely, for $\ell\geq 0$ and $j\in\{1,\dots,c\}$, either 
%$(V_{\ell}\cup V_{\ell+1} ) \cap V( X_j ) =\emptyset$ or $(V_{\ell}\cup V_{\ell+1} ) \cap V( X_j ) \subseteq V( X_{i,j} )$ for some $i\in\mathbb{Z}$.
%
%Now consider a clique $C$ of $G$. 
%Our goal is to show that $C\cap V( X_j )$ is a subset of some bag $B_x^j$ for all $j\in\{1,\dots,c\}$. 
%If $( C\setminus A ) \cap V( X_j ) = \emptyset$ then 
%$C\cap V( X_j ) \subseteq A$, implying that $C\cap V( X_j )$ is a subset of every bag $B^j_x$. 
%Now assume that $( C\setminus A ) \cap V( X_j ) \neq \emptyset$. 
%By the definition of layering, 
% $C\setminus A \subseteq V_\ell \cup V_{\ell+1}$ for some $\ell\geq 0$. By the above key property, 
%  $( C\setminus A ) \cap V( X_j ) \subseteq V( X_{i,j} )$ for some    $i\in\mathbb{Z}$ (depending on $\ell$). 
%Let $W$ be the set of vertices $h$ of $H$ such that $Z_h \cap (C\setminus A) \neq\emptyset$. 
%Since $C\setminus A$ is a clique, $W$ is a clique of $H$. 
%Hence, some bag $H_x$ of the tree-decomposition of $H$ contains $W$. 
%By construction, $D_x$ contains $(C\setminus A) \cap V( X_j )$ in the tree-decomposition of $X_{i,j}$, and thus 
%$B^j_x$ contains $C \cap V(X_j )$ in the tree-decomposition of $X_j$. 
\end{proof}

\begin{lemma}
\label{cEdgeColourCliqueSum}
Fix an integer $c\geq 2$. 
For $i\in\{1,2\}$, let $G^i$ be a graph for which there are induced subgraphs $G^i_1,\dots,G^i_c$ satisfying \cref{cEdgeColourAlmost}. 
Let $G$ be a clique-sum of $G^1$ and $G^2$. 
Then $G$ has induced subgraphs $G_1,\dots,G_c$ satisfying \cref{cEdgeColourAlmost}.
\end{lemma}

\begin{proof}
Let $C^1$ and $C^2$ be the cliques  respectively in $G^1$ and $G^2$ involved in the clique-sum. 
For $i\in\{1,2\}$,  let  $X^i_j := G^i_1 \cup \dots \cup G^i_{j-1} \cup G^i_{j+1} \cup\dots \cup G^i_c$. 
By assumption,  $G^i = \bigcup_{j=1}^c G^i_j$, and  $X^i_j$ is an induced subgraph of $G^i$ that has a  tree-decomposition $(B^{i,j}_x:x\in V(T_j))$ of width at most $66k(k+1)(2c-1)+k-1$. %, such that  for every clique $C$ of $G^i$, $C\cap V(X^i_j)$ is a subset of some bag $B^{i,j}_x$.
For $j\in\{1,\dots,c\}$, let $G_j := G^1_j \cup G^2_j$. 
This means that for vertices $v_1\in C^1$ and $v_2\in C^2$, if $v_1$ and $v_2$ are identified into $v$ in the clique-sum, 
and $v_1\in V(G^1_j)$ or $v_2\in V(G^2_j)$, then $v$ is in $V(G_j)$. Similarly, for vertices $v_1,w_1\in C^1$ and $v_2,w_2\in C^2$, if $v_1$ and $v_2$ are identified into $v$ in the clique-sum, and  $w_1$ and $w_2$ are identified into $w$ in the clique-sum,  and $v_1w_1\in E(G^1_j)$ or $v_2w_2\in E(G^2_j)$, then $vw$ is in $E(G_j)$.  
Take the disjoint union of the tree-decompositions of $X^1_j$ and $X^2_j$ and add an edge between a bag containing $C^1\cap V(X^1_j)$ and a bag containing $C^2\cap V(X^2_j)$ to obtain a tree-decomposition of $X_j$ of width  $66k(k+1)(2c-1)+k-1$. Such bags exist since $C^i\cap V(X^i_j)$ is a clique of $X^i_j$ and is thus a subset of some bag. 
%Every clique $C$ of $G$ is a clique of $G^1$ or $G^2$, and thus  $C\cap V(X_j)$ is a subset of some bag of the tree-decomposition of $X_j$ by assumption. 
\end{proof}

We now prove the main result of this section.

\begin{theorem}[\citep{DDOSRSV04}]
\label{2EdgeColour}
For every proper minor-closed class $\mathcal{G}$ and integer $c\geq 2$, there is a constant $k$ such that every graph in $\mathcal{G}$ can be edge $c$-coloured or vertex $c$-coloured  so that the union of any $c-1$ colour classes has treewidth at most $k$. 
\end{theorem}

\begin{proof}
\cref{GMST,cEdgeColourAlmost,cEdgeColourCliqueSum} imply that there exists an integer $k$ such that every graph  $G\in\mathcal{G}$ has  subgraphs $G_1,\dots,G_c$, such that $G=\bigcup_{j=1}^c G_j$, 
and for $j\in\{1,\dots,c\}$ the subgraph $X_j := G_1 \cup \dots \cup G_{j-1} \cup G_{j+1} \cup\dots \cup G_c$ has treewidth at most $66k(k+1)(2c-1)+k-1$. 

First we prove the edge-colouring result. Colour each edge $e$ of $G$ by an integer $j$ for which $e\in E(G_j)$. The subgraph of $G$ induced by the edges not coloured $j$ is a subgraph of $X_j$, and thus has treewidth at most $66k(k+1)(2c-1)+k-1$. 

For the vertex-colouring result, colour each vertex $v$ of $G$ by an integer $j$ for which $v\in V(G_j)$. The subgraph of $G$ induced by the vertices not coloured $j$ is a subgraph of $X_j$, and thus has treewidth at most $66k(k+1)(2c-1)+k-1$. 
\end{proof}

\subsection{Track Layouts}
\label{Track}

Track layouts are a type of graph layout closely related to queue layouts. A \emph{vertex} $k$\emph{-colouring} of a graph $G$ is a partition $\{V_1,\dots,V_k\}$ of $V(G)$ into independent sets; that is, for every edge $vw\in E(G)$, if $v\in V_i$ and $w\in V_j$ then $i\ne j$. A \emph{track} in $G$ is an independent set equipped with a linear ordering. A partition $\{\overrightarrow{V_1},\dots,\overrightarrow{V_k}\}$ of $V(G)$ into $k$ tracks is a \emph{$k$-track layout} if for distinct $i,j\in\{1,\dots,k\}$ no two edges of $G$ cross between $\overrightarrow{V_i}$ and $\overrightarrow{V_j}$. That is, for all distinct edges $vw,xy\in E(G)$ with $v,x\in V_i$ and $w,y\in V_j$, if $v\prec x$ in $\overrightarrow{V_i}$ then $w\preceq y$ in $\overrightarrow{V_j}$. The minimum $k$ such that $G$ has a $k$-track layout is called the \emph{track-number} of $G$, denoted by $\tn(G)$. \citet{DMW05} proved the following connection to queue-number.

\begin{lemma}[\citep{DMW05}]
\label{TrackQueue}
For every graph $G$, $\qn(G)\leq\tn(G)-1.$
\end{lemma}

The proof of \cref{TrackQueue} simply puts the tracks one after the other to produce a queue layout. In this sense, track layouts can be thought of as a richer structure than queue layouts. This structure was the key to an inductive proof by \citet{DMW05} that graphs of bounded treewidth have bounded track-number (which implies bounded queue-number by \cref{TrackQueue}). 
%In general, it is easier to deal with the clique-sum operation in the setting of track layouts than in queue layouts, which is critical in our proof for proper minor-closed classes (see \cref{Minors}). 
Nevertheless, \citet{DPW04} proved the following converse to \cref{TrackQueue}:

\begin{lemma}[\citep{DPW04}]
\label{QueueTrack}
There is a function $f$ such that $\tn(G) \leq f( \qn(G) )$ for every graph $G$. In particular, every graph with queue-number at most $k$ has track-number at most
$$4k \cdot 4^{ k (2k-1)(4k-1) }.$$
\end{lemma}

\cref{TrackQueue,QueueTrack} together say that queue-number and track-number are tied.

The following lemma often gives better bounds on the track-number than \cref{QueueTrack}. A proper graph colouring is \emph{acyclic} if every cycle gets at least three colours. The \emph{acyclic chromatic number} of a graph $G$ is the minimum integer $c$ such that $G$ has an acyclic $c$-colouring.

\begin{lemma}[\citep{DMW05}]
\label{AcyclicQueueTrack}
Every graph $G$ with acyclic chromatic number at most $c$ and queue-number at most $k$ has track-number at most $c(2k)^{c-1}$.
\end{lemma}

\citet{Borodin79} proved that planar graphs have acyclic chromatic number at most 5, which with \cref{AcyclicQueueTrack,PlanarQueue} implies:

\begin{theorem}
\label{PlanarTrack}
Every planar graph has track-number at most $5 (2 \cdot 49)^4= 461,184,080$. 
\end{theorem}

Note that the best lower bound on the track-number of planar graphs is 8, due to \citet{Pupyrev19}.

\citet{Heawood1890} and \citet{AMS96} respectively proved that every graph with Euler genus $g$ has chromatic number $O(g^{1/2})$ and acyclic chromatic number $O(g^{4/7})$. \cref{AcyclicQueueTrack,GenusQueue} then imply:

\begin{theorem}
\label{GenusTrack}
Every graph with Euler genus $g$ has track-number at most $g^{O( g^{4/7})}$.
\end{theorem}

%\citet{NesOss03} proved that every proper minor-closed class has bounded acyclic chromatic number. The best upper bound for $K_t$-minor-free graphs is $O(t^2)$ which follows from results by \citet{HOQRS17}. \cref{AcyclicQueueTrack,MinorQueue} then imply:

For proper minor-closed classes, \cref{QueueTrack,MinorQueue} imply:

\begin{theorem}
\label{MinorTrack}
Every proper minor-closed class has bounded track-number.
\end{theorem}

We now briefly show that $(g,k)$-planar graphs have bounded track-number. 
First note that every graph with layered treewidth $k$ has acyclic chromatic number at most $5k$ [\emph{Proof.} Van den Heuvel and Wood~\citep{vdHW17} proved that every graph with layered treewidth $k$ has strong $r$-colouring number at most $k(2r + 1)$, and \citet{KY03} proved that every graph has acyclic chromatic number at most its strong $2$-colouring number.]\ \citet{DEW17} proved that every $(g,k)$-planar graph $G$ has layered treewidth at most $(4g + 6)(k + 1)$. Thus $G$ has  acyclic chromatic number at most $5(4g + 6)(k + 1)$, and has bounded track-number by \cref{AcyclicQueueTrack,gkPlanar}. \citet{DEW17} also proved that every $(g,d)$-map graph and every $(g,k)$-planar graph has bounded layered treewidth. By the same argument, such graphs have bounded track-number. 

%\comment{Add reference for the claim that every graph has acyclic chromatic number at most its $1$-strong colouring number. Kiersted and Yang (original paper) OR Uniform Orderings for Generalized Coloring Numbers, Jan van den Heuvel H.A. Kierstead} 

\subsection{Three-Dimensional Graph Drawing}

Further motivation for studying queue and track layouts is their connection with 3-dimensional graph drawing. A \emph{3-dimensional grid drawing} of a graph $G$ represents the vertices of $G$ by distinct grid points in $\mathbb{Z}^3$ and represents each edge of $G$ by the open segment between its endpoints so that no two edges intersect. The \emph{volume} of a 3-dimensional grid drawing is the number of grid points in the smallest axis-aligned grid-box that encloses the drawing. For example,  \citet{CELR-Algo96} proved that the complete graph $K_n$ has a 3-dimensional grid drawing with volume $O(n^3)$ and this bound is optimal. \citet{PTT99} proved that every graph with bounded chromatic number has a 3-dimensional grid drawing with volume $O(n^2)$, and this bound is optimal for $K_{n/2,n/2}$.

Track layouts and 3-dimensional graph drawings are connected by the following lemma.

\begin{lemma}[\citep{DMW05,DujWoo-SubQuad-AMS}]
\label{Track2Drawing}
If a $c$-colourable $n$-vertex graph $G$ has a $t$-track layout, then $G$ has 3-dimensional grid drawings with $O(t^2n)$ volume and  with $O(c^7tn)$ volume. Conversely, if a graph $G$ has a 3-dimensional grid drawing with
$A\times B\times C$ bounding box, then $G$ has track-number at most $2AB$.
\end{lemma}

\cref{Track2Drawing} is the foundation for all of the following results. \citet{DujWoo-SubQuad-AMS} proved that every graph with bounded maximum degree has a 3-dimensional grid drawing with volume $O(n^{3/2})$, and the same bound holds for graphs from a proper minor-closed class. In fact, for fixed $d$, every $d$-degenerate graph\footnote{A graph is \emph{$d$-degenerate} if every subgraph has minimum degree at most $d$.} has a 3-dimensional grid drawing with $O(n^{3/2})$ volume~\citep{DujWoo-Order06}. \citet{DMW05} proved that every graph with bounded treewidth has a 3-dimensional grid drawing with volume $O(n)$.

Prior to this work, whether planar graphs have 3-dimensional grid drawings with $O(n)$ volume was a major open problem, due to Felsner, Liotta, and Wismath~\citep{FLW-GD01-ref}. The previous best known bound on the volume of 3-dimensional grid drawings of planar graphs was $O(n\log n)$ by \citet{Duj15}. \cref{Track2Drawing,PlanarTrack} together resolve the open problem of Felsner et al.~\citep{FLW-GD01-ref}.

%,FLW-JGAA03

\begin{theorem}
\label{PlanarDrawing}
Every planar graph with $n$ vertices has a 3-dimensional grid drawing with $O(n)$ volume.
\end{theorem}

\cref{Track2Drawing,GenusTrack,MinorTrack} imply the following strengthenings of \cref{PlanarDrawing}.

\begin{theorem}
\label{GenusDrawing}
Every graph with Euler genus $g$ and $n$ vertices has a 3-dimensional grid drawing with $g^{O(g^{4/7})}n$ volume. 
\end{theorem}

\begin{theorem}
\label{MinorDrawing}
For every proper minor-closed class $\mathcal{G}$, every graph in $\mathcal{G}$ with $n$ vertices has a 3-dimensional grid drawing with $O(n)$ volume.
\end{theorem}

As shown in \cref{Track}, $(g,k)$-planar graphs, $(g,d)$-map graphs and $(g,k)$-string graphs have bounded track-number (for fixed $g,k,d$). By \cref{Track2Drawing}, such graphs have 3-dimensional grid drawings with $O(n)$ volume. 

%%%%%%%%%%%%%%%%%%%%
\section{Open Problems}\label{conclusion}
\label{Questions}

\begin{enumerate}

\item What is the maximum queue-number of planar graphs? We can tweak our proof of \cref{PlanarQueue} to show that every planar graph has queue-number at most 48, but it seems new ideas are required to obtain a significant improvement. The best lower bound on the maximum queue-number of planar graphs is $4$, due to \citet{ABGKP18}.

 More generally, does every graph with Euler genus $g$ have $o(g)$ queue-number? Complete graphs provide a $\Theta(\sqrt{g})$ lower bound. Note that every graph with Euler genus $g$ has $O(\sqrt{g})$ stack-number~\citep{Malitz94b}.

\item As discussed in \cref{Introduction} it is open whether there is a function $f$ such that $\sn(G) \leq f( \qn(G) )$ for every graph $G$. \citet{HLR92} proved that every 1-queue graph has stack-number at most $2$. \citet{DujWoo05} showed that there is such a function $f$ if and only if every 2-queue graph has bounded stack-number. 

Similarly, it is open whether there is a function $f$ such that $\qn(G) \leq f( \sn(G) )$ for every graph $G$. \citet{HLR92} proved that every 1-stack graph has queue-number at most 2. Since 2-stack graphs are planar, this paper solves the first open case of this question. \citet{DujWoo05} showed that there is such a function $f$ if and only if every 3-stack graph has bounded queue-number. 

\item \citet{OOW19} introduced the following definition:
 A graph $G$ is said to be \emph{$k$-close to Euler genus $g$} if every subgraph $H$ of $G$ has a drawing in a surface of Euler genus $g$ with at most $k\,|E(H)|$ crossings (that is, with $O(k)$ crossings per edge on average). Does every such graph have queue-number at most $f(g,k)$ for some function $f$?

\item Is there a proof of \cref{MinorQueue} that does not use the graph minor structure theorem and with more reasonable bounds?

\smallskip
\item Queue layouts naturally extend to posets. The \emph{cover graph} $G_P$ of a poset $P$ is the undirected graph with vertex set $P$, where $vw\in E(G)$ if $v<_Pw$ and $v<_Px<_Pw$ for no $x\in P$ (or $w<_Pv$ and $w<_Px<_Pv$ for no $x\in P$). Thus the cover graph encodes relations in $P$ that are not implied by transitivity. A \emph{$k$-queue layout} of a poset $P$ consists of a linear extension  $\preceq$ of $P$ and a partition $E_1,E_2,\dots,E_k$ of $E(G_P)$ into queues with respect to $\preceq$. The \textit{queue-number} of a poset $P$ is the minimum integer $k$ such that $P$ has a $k$-queue layout. \citet{HP97} conjectured that the queue-number of a planar poset is at most its height (the maximum number of pairwise comparable elements).  
This was disproved by \citet{KMU18} who presented a poset of height $2$ and queue-number $4$.
\cref{PlanarQueue} and results of Knauer, Micek and Ueckerdt imply  that planar posets of height $h$ have queue-number $O(h)$; see Theorem~6 in~\citep{KMU18}. \citet{HP97} also conjecture that every poset of width $w$ (the maximum number of pairwise incomparable elements) has queue-number at most $w$. The best known upper bounds are $O(w^2)$ for general posets and $3w-2$ for planar posets~\citep{KMU18}.

%\comment{DW: The above paragraph is new; please check it.\\
%Piotr: Heath and Pemmaraju conjecture concerning the height was disproved. I changed the text accordingly.}

\item It is natural to ask for the largest class of graphs with bounded queue-number. First note that \cref{MinorQueue} cannot be extended to the setting of an excluded topological minor, since graphs with bounded degree have arbitrarily high queue-number~\citep{Wood-QueueDegree,HLR92}. However, it is possible that every class of graphs with strongly sub-linear separators has bounded queue-number. Here a class $\mathcal{G}$ of graphs has \emph{strongly sub-linear separators} if $\mathcal{G}$ is closed under taking subgraphs, and there exists constants $c,\beta>0$, such that every $n$-vertex graph in $\mathcal{G}$ has a balanced separator of order $cn^{1-\beta}$. Already the $\beta=\frac12$ case looks challenging, since this would imply \cref{MinorQueue}.

\item Is there a polynomial function $f$ such that every graph with treewidth $k$ has queue-number at most $f(k)$? The best lower and upper bounds on $f(k)$ are $k+1$ and $2^k-1$, both due to \citet{Wiechert17}.

\item Do the results in the present paper have algorithmic applications? Consider the method of \citet{Baker94} for designing polynomial-time approximation schemes for problems on planar graphs. This method partitions the graph into BFS layers, such that the problem can be solved optimally on each layer (since the induced subgraph has bounded treewidth), and then combines the solutions from each layer. Our results (\cref{Width1LayeredPartition}) give a more precise description of the layered structure of planar graphs (and other more general classes). It is conceivable that this extra structural information is useful when designing algorithms.  

%\comment{Torsten: in the runtime discussion we don't mention Lemma 16. is that also quadratic to find the set $Z$? DW: Added a comment about \cref{MakePlanar}}

Note that all our proofs lead to polynomial-time algorithms for computing the desired decomposition and queue layout. 
\citet{PS18} claim $O(n^2)$ time complexity for their decomposition. The same is true for \cref{NearTriang}: Given the colours of the vertices on $F$, we can walk down the BFS tree $T$ in linear time and colour every vertex. Another linear-time enumeration of  the faces contained in $F$ finds the trichromatic triangle.  It is easily seen that \cref{MakePlanar} has polynomial time complexity (given the embedding). Polynomial-time algorithms for our other results follow based on the linear-time algorithm of  \citet{Mohar99} to test if a given graph has Euler genus at most any fixed number $g$, and the polynomial-time algorithm of \citet{DHK-FOCS05} for computing the decomposition in the graph minor structure theorem (\cref{GMST}). 

\end{enumerate}

%%%%%%%%%%%%%%%%
\subsection*{Acknowledgements}

This research was completed at the 7th Annual Workshop on Geometry and Graphs held at Bellairs Research Institute in March 2019. Thanks to the other workshop participants for creating a productive working atmosphere.

%%%%%%%%%%%%%%%%
\subsection*{Note Added in Proof}

This paper has motivated several follow-up works. Analogues of \cref{PlanarProduct,GenusProduct} have been proved for bounded degree graphs in any minor-closed class~\citep{DEMWW} and for $k$-planar graphs and several other non-minor-closed classes of interest~\citep{DMW19a}. See~\citep{DHGLW} for a survey of such `product structure theorems'. \citet{Morin20} presents $O(n\log n)$ time algorithms for finding the partitions in \cref{Width1LayeredPartition,Width3LayeredPartition}. \citet{Pupyrev19} improves the bound on the track-number of planar graphs in \cref{PlanarTrack} (by constructing a track layout directly from \cref{Width1LayeredPartition} instead of using an intermediate queue layout).

\label{lastpage}

%%%%%%%%%%%%%%%%%%%%%%%%%
%%  Squashing the bibliography
  \let\oldthebibliography=\thebibliography
  \let\endoldthebibliography=\endthebibliography  
  \renewenvironment{thebibliography}[1]{%
    \begin{oldthebibliography}{#1}%
      \setlength{\parskip}{0ex}%
      \setlength{\itemsep}{0ex}%
  }{\end{oldthebibliography}}

%\bibliographystyle{../../../BibTex/myNatbibStyle}
%\bibliography{../../../BibTex/myBibliography}

\begin{thebibliography}{113}
	\providecommand{\natexlab}[1]{#1}
	\providecommand{\url}[1]{\texttt{#1}}
	\providecommand{\urlprefix}{}
	\expandafter\ifx\csname urlstyle\endcsname\relax
	\providecommand{\doi}[1]{doi:\discretionary{}{}{}#1}\else
	\providecommand{\doi}{doi:\discretionary{}{}{}\begingroup
		\urlstyle{rm}\Url}\fi
	
	\bibitem[{Aigner and Ziegler(2010)}]{Proofs4}
	\textsc{Martin Aigner and G{\"u}nter~M. Ziegler}.
	\newblock \emph{Proofs from {T}he {B}ook}.
	\newblock Springer, 4th edn., 2010.
	
	\bibitem[{Alam et~al.(2018)Alam, Bekos, Gronemann, Kaufmann, and
		Pupyrev}]{ABGKP18}
	\textsc{Jawaherul~Md. Alam, Michael~A. Bekos, Martin Gronemann, Michael
		Kaufmann, and Sergey Pupyrev}.
	\newblock Queue layouts of planar 3-trees.
	\newblock In \textsc{Therese~C. Biedl and Andreas Kerren}, eds., \emph{Proc.
		26th International Symposium on Graph Drawing and Network Visualization} (GD '18), vol. 11282 of \emph{Lecture Notes in Comput. Sci.}, pp. 213--226.
	Springer, 2018.
	\newblock \doi{10.1007/978-3-030-04414-5\_15}.
	
	\bibitem[{Alam et~al.(2020)Alam, Bekos, Gronemann, Kaufmann, and
		Pupyrev}]{ABGKP20}
	\textsc{Jawaherul~Md. Alam, Michael~A. Bekos, Martin Gronemann, Michael
		Kaufmann, and Sergey Pupyrev}.
	\newblock Queue layouts of planar 3-trees.
	\newblock \emph{Algorithmica}, 2020.
	\newblock \doi{10.1007/s00453-020-00697-4}.
	
	\bibitem[{Alon and Asodi(2002)}]{AA02}
	\textsc{Noga Alon and Vera Asodi}.
	\newblock Sparse universal graphs.
	\newblock \emph{J. Comput. Appl. Math.}, 142(1):1--11, 2002.
	\newblock \doi{10.1016/S0377-0427(01)00455-1}.
	\newblock \msn{1910514}.
	
	\bibitem[{Alon and Capalbo(2007)}]{AC07}
	\textsc{Noga Alon and Michael Capalbo}.
	\newblock Sparse universal graphs for bounded-degree graphs.
	\newblock \emph{Random Structures Algorithms}, 31(2):123--133, 2007.
	\newblock \doi{10.1002/rsa.20143}.
	
	\bibitem[{Alon et~al.(2002)Alon, Grytczuk, Ha{\l}uszczak, and Riordan}]{AGHR02}
	\textsc{Noga Alon, {Jaros{\l}aw} Grytczuk, Mariusz Ha{\l}uszczak, and Oliver
		Riordan}.
	\newblock Nonrepetitive colorings of graphs.
	\newblock \emph{Random Structures Algorithms}, 21(3--4):336--346, 2002.
	\newblock \doi{10.1002/rsa.10057}.
	\newblock \msn{1945373}.
	
	\bibitem[{Alon et~al.(1996)Alon, Mohar, and Sanders}]{AMS96}
	\textsc{Noga Alon, Bojan Mohar, and Daniel~P. Sanders}.
	\newblock On acyclic colorings of graphs on surfaces.
	\newblock \emph{Israel J. Math.}, 94:273--283, 1996.
	\newblock \doi{10.1007/BF02762708}.
	
	\bibitem[{Andreae(1986)}]{Andreae86}
	\textsc{Thomas Andreae}.
	\newblock On a pursuit game played on graphs for which a minor is excluded.
	\newblock \emph{J. Comb. Theory, Ser. {B}}, 41(1):37--47, 1986.
	\newblock \doi{10.1016/0095-8956(86)90026-2}.
	\newblock \msn{0854602}.
	
	\bibitem[{Babai et~al.(1982)Babai, Chung, Erd\H{o}s, Graham, and
		Spencer}]{BCEGS82}
	\textsc{L{\'a}szl{\'o} Babai, Fan R.~K. Chung, Paul Erd\H{o}s, Ron~L. Graham,
		and Joel~H. Spencer}.
	\newblock On graphs which contain all sparse graphs.
	\newblock In \emph{Theory and practice of combinatorics}, vol.~60 of
	\emph{North-Holland Math. Stud.}, pp. 21--26. 1982.
	\newblock \msn{806964}.
	
	\bibitem[{Baker(1994)}]{Baker94}
	\textsc{Brenda~S. Baker}.
	\newblock Approximation algorithms for {NP}-complete problems on planar graphs.
	\newblock \emph{J. ACM}, 41(1):153--180, 1994.
	\newblock \doi{10.1145/174644.174650}.
	\newblock \msn{1369197}.
	
	\bibitem[{Bannister et~al.(2019)Bannister, Devanny, Dujmovi\'c, Eppstein, and
		Wood}]{BDDEW18}
	\textsc{Michael~J. Bannister, William~E. Devanny, Vida Dujmovi\'c, David
		Eppstein, and David~R. Wood}.
	\newblock Track layouts, layered path decompositions, and leveled planarity.
	\newblock \emph{Algorithmica}, 81(4):1561--1583, 2019.
	\newblock \doi{10.1007/s00453-018-0487-5}.
	\newblock \msn{3936168}.
	
	\bibitem[{Bekos et~al.(2019)Bekos, F{\"{o}}rster, Gronemann, Mchedlidze,
		Montecchiani, Raftopoulou, and Ueckerdt}]{BFGMMRU19}
	\textsc{Michael~A. Bekos, Henry F{\"{o}}rster, Martin Gronemann, Tamara
		Mchedlidze, Fabrizio Montecchiani, Chrysanthi~N. Raftopoulou, and Torsten
		Ueckerdt}.
	\newblock Planar graphs of bounded degree have bounded queue number. 
	\newblock \emph{SIAM J. Comput.}, 48(5):1487--1502, 2019. 
	\newblock \doi{10.1137/19M125340X}.
	
	\bibitem[{Bhatt et~al.(1989)Bhatt, Chung, Leighton, and Rosenberg}]{BCLR89}
	\textsc{Sandeep~N. Bhatt, Fan R.~K. Chung, Frank~Thomson Leighton, and
		Arnold~L. Rosenberg}.
	\newblock Universal graphs for bounded-degree trees and planar graphs.
	\newblock \emph{{SIAM} J. Discrete Math.}, 2(2):145--155, 1989.
	\newblock \doi{10.1137/0402014}.
	\newblock \msn{990447}.
	
	\bibitem[{Bodlaender(1998)}]{Bodlaender-TCS98}
	\textsc{Hans~L. Bodlaender}.
	\newblock A partial $k$-arboretum of graphs with bounded treewidth.
	\newblock \emph{Theoret. Comput. Sci.}, 209(1-2):1--45, 1998.
	\newblock \doi{10.1016/S0304-3975(97)00228-4}.
	\newblock \msn{1647486}.
	
	\bibitem[{Bodlaender and Engelfriet(1997)}]{BodEng-JAlg97}
	\textsc{Hans~L. Bodlaender and Joost Engelfriet}.
	\newblock Domino treewidth.
	\newblock \emph{J. Algorithms}, 24(1):94--123, 1997.
	\newblock \doi{10.1006/jagm.1996.0854}.
	\newblock \msn{1453952}.
	
	\bibitem[{Bonamy et~al.(2020)Bonamy, Gavoille, and Pilipczuk}]{BGP20}
	\textsc{Marthe Bonamy, Cyril Gavoille, and Micha{\l} Pilipczuk}.
	\newblock Shorter labeling schemes for planar graphs.
	\newblock In \textsc{Shuchi Chawla}, ed., \emph{Proc. ACM-SIAM Symposium on Discrete Algorithms} (SODA '20), pp.~446--462, 2020.
	\newblock \doi{10.1137/1.9781611975994.27}.
	\newblock \arXiv{1908.03341}.
	
	\bibitem[{Borodin(1979)}]{Borodin79}
	\textsc{Oleg~V. Borodin}.
	\newblock On acyclic colorings of planar graphs.
	\newblock \emph{Discrete Math.}, 25(3):211--236, 1979.
	\newblock \doi{10.1016/0012-365X(79)90077-3}.
	
	\bibitem[{Borodin et~al.(1998)Borodin, Kostochka, Ne{\v{s}}et{\v{r}}il,
		Raspaud, and Sopena}]{BKNRS-DM98}
	\textsc{Oleg~V. Borodin, Alexandr~V. Kostochka, Jaroslav Ne{\v{s}}et{\v{r}}il,
		Andr\'e Raspaud, and \'Eric Sopena}.
	\newblock On universal graphs for planar oriented graphs of a given girth.
	\newblock \emph{Discrete Math.}, 188(1--3):73--85, 1998.
	\newblock \doi{10.1016/S0012-365X(97)00276-8}.
	
	\bibitem[{B{\"{o}}ttcher et~al.(2010)B{\"{o}}ttcher, Pruessmann, Taraz, and
		W{\"{u}}rfl}]{BPTW10}
	\textsc{Julia B{\"{o}}ttcher, Klaas~Paul Pruessmann, Anusch Taraz, and Andreas
		W{\"{u}}rfl}.
	\newblock Bandwidth, expansion, treewidth, separators and universality for
	bounded-degree graphs.
	\newblock \emph{European J. Combin.}, 31(5):1217--1227, 2010.
	\newblock \doi{10.1016/j.ejc.2009.10.010}.
	
	\bibitem[{Buss and Shor(1984)}]{BS84}
	\textsc{Jonathan~F. Buss and Peter Shor}.
	\newblock On the pagenumber of planar graphs.
	\newblock In \emph{Proc. 16th ACM Symp. on Theory of Computing} (STOC '84), pp.
	98--100. ACM, 1984.
	\newblock \doi{10.1145/800057.808670}.
	
	\bibitem[{Cabello et~al.(2012)Cabello, de~Verdi{\`{e}}re, and Lazarus}]{CCL12}
	\textsc{Sergio Cabello, {\'{E}}ric~Colin de~Verdi{\`{e}}re, and Francis
		Lazarus}.
	\newblock Algorithms for the edge-width of an embedded graph.
	\newblock \emph{Comput. Geom.}, 45(5-6):215--224, 2012.
	\newblock \doi{10.1016/j.comgeo.2011.12.002}.
	
	\bibitem[{Chandran et~al.(2008)Chandran, Kostochka, and Raju}]{CKR08}
	\textsc{L.~Sunil Chandran, Alexandr Kostochka, and J.~Krishnam Raju}.
	\newblock Hadwiger number and the {C}artesian product of graphs.
	\newblock \emph{Graphs Combin.}, 24(4):291--301, 2008.
	\newblock \doi{10.1007/s00373-008-0795-7}.
	
	\bibitem[{Chen(2001)}]{Chen01}
	\textsc{Zhi-Zhong Chen}.
	\newblock Approximation algorithms for independent sets in map graphs.
	\newblock \emph{J. Algorithms}, 41(1):20--40, 2001.
	\newblock \doi{10.1006/jagm.2001.1178}.
	
	\bibitem[{Chen(2007)}]{Chen07}
	\textsc{Zhi-Zhong Chen}.
	\newblock New bounds on the edge number of a {$k$}-map graph.
	\newblock \emph{J. Graph Theory}, 55(4):267--290, 2007.
	\newblock \doi{10.1002/jgt.20237}.
	\newblock \msn{2336801}.
	
	\bibitem[{Chen et~al.(2002)Chen, Grigni, and Papadimitriou}]{CGP02}
	\textsc{Zhi-Zhong Chen, Michelangelo Grigni, and Christos~H. Papadimitriou}.
	\newblock Map graphs.
	\newblock \emph{J. ACM}, 49(2):127--138, 2002.
	\newblock \doi{10.1145/506147.506148}.
	\newblock \msn{2147819}.
	
	\bibitem[{Cohen et~al.(1996)Cohen, Eades, Lin, and Ruskey}]{CELR-Algo96}
	\textsc{Robert~F. Cohen, Peter Eades, Tao Lin, and Frank Ruskey}.
	\newblock Three-dimensional graph drawing.
	\newblock \emph{Algorithmica}, 17(2):199--208, 1996.
	\newblock \doi{10.1007/BF02522826}.
	
	\bibitem[{Demaine et~al.(2005{\natexlab{a}})Demaine, Fomin, Hajiaghayi, and
		Thilikos}]{DFHT05}
	\textsc{Erik~D. Demaine, Fedor~V. Fomin, Mohammadtaghi Hajiaghayi, and
		Dimitrios~M. Thilikos}.
	\newblock Fixed-parameter algorithms for {$(k,r)$}-center in planar graphs and
	map graphs.
	\newblock \emph{ACM Trans. Algorithms}, 1(1):33--47, 2005{\natexlab{a}}.
	\newblock \doi{10.1145/1077464.1077468}.
	
	\bibitem[{Demaine and Hajiaghayi(2004{\natexlab{a}})}]{DH-Algo04}
	\textsc{Erik~D. Demaine and MohammadTaghi Hajiaghayi}.
	\newblock Diameter and treewidth in minor-closed graph families, revisited.
	\newblock \emph{Algorithmica}, 40(3):211--215, 2004{\natexlab{a}}.
	\newblock \doi{10.1007/s00453-004-1106-1}.
	\newblock \msn{2080518}.
	
	\bibitem[{Demaine and Hajiaghayi(2004{\natexlab{b}})}]{DH-SODA04}
	\textsc{Erik~D. Demaine and MohammadTaghi Hajiaghayi}.
	\newblock Equivalence of local treewidth and linear local treewidth and its
	algorithmic applications.
	\newblock In \emph{Proc. 15th Annual ACM-SIAM Symposium on Discrete Algorithms}
		(SODA '04), pp. 840--849. SIAM, 2004{\natexlab{b}}.
	\newblock \url{https://dl.acm.org/doi/abs/10.5555/982792.982919}. 
	
	\bibitem[{Demaine et~al.(2005{\natexlab{b}})Demaine, Hajiaghayi, and
		Kawarabayashi}]{DHK-FOCS05}
	\textsc{Erik~D. Demaine, MohammadTaghi Hajiaghayi, and {{Ken-ichi}}
		Kawarabayashi}.
	\newblock Algorithmic graph minor theory: Decomposition, approximation, and
	coloring.
	\newblock In \emph{Proc. 46th Annual IEEE Symposium on Foundations of Computer
		Science} (FOCS '05), pp. 637--646. IEEE, 2005{\natexlab{b}}.
	\newblock \doi{10.1109/SFCS.2005.14}.
	
	\bibitem[{DeVos et~al.(2004)DeVos, Ding, Oporowski, Sanders, Reed, Seymour, and
		Vertigan}]{DDOSRSV04}
	\textsc{Matt DeVos, Guoli Ding, Bogdan Oporowski, Daniel~P. Sanders, Bruce
		Reed, Paul Seymour, and Dirk Vertigan}.
	\newblock Excluding any graph as a minor allows a low tree-width 2-coloring.
	\newblock \emph{J. Combin. Theory Ser. B}, 91(1):25--41, 2004.
	\newblock \doi{10.1016/j.jctb.2003.09.001}.
	\newblock \msn{2047529}.
	
	\bibitem[{{Di Battista} et~al.(2013){Di Battista}, Frati, and Pach}]{DFP13}
	\textsc{Giuseppe {Di Battista}, Fabrizio Frati, and J\'anos Pach}.
	\newblock On the queue number of planar graphs.
	\newblock \emph{SIAM J. Comput.}, 42(6):2243--2285, 2013.
	\newblock \doi{10.1137/130908051}.
	\newblock \msn{3141759}.
	
	\bibitem[{Di~Giacomo and Meijer(2004)}]{DM-GD03}
	\textsc{Emilio Di~Giacomo and Henk Meijer}.
	\newblock Track drawings of graphs with constant queue number.
	\newblock In \textsc{Giuseppe Liotta}, ed., \emph{Proc. 11th International
		Symp. on Graph Drawing} (GD '03), vol. 2912 of \emph{Lecture Notes in Comput.
		Sci.}, pp. 214--225. Springer, 2004.
	\newblock \doi{10.1007/978-3-540-24595-7\_20}.
	
	\bibitem[{Diestel(2016)}]{Diestel5}
	\textsc{Reinhard Diestel}.
	\newblock \emph{Graph theory}, vol. 173 of \emph{Graduate Texts in
		Mathematics}.
	\newblock Springer, 5th edn., 2016.
	\newblock \urlprefix\url{http://diestel-graph-theory.com/}.
	\newblock \msn{3644391}
	
	\bibitem[{Diestel and K{\"u}hn(2005)}]{DK05}
	\textsc{Reinhard Diestel and Daniela K{\"u}hn}.
	\newblock Graph minor hierarchies.
	\newblock \emph{Discrete Appl. Math.}, 145(2):167--182, 2005.
	\newblock \doi{10.1016/j.dam.2004.01.010}.
	
	\bibitem[{Ding and Oporowski(1995)}]{DO95}
	\textsc{Guoli Ding and Bogdan Oporowski}.
	\newblock Some results on tree decomposition of graphs.
	\newblock \emph{J. Graph Theory}, 20(4):481--499, 1995.
	\newblock \doi{10.1002/jgt.3190200412}.
	\newblock \msn{1358539}.
	
	\bibitem[{Ding and Oporowski(1996)}]{DO96}
	\textsc{Guoli Ding and Bogdan Oporowski}.
	\newblock On tree-partitions of graphs.
	\newblock \emph{Discrete Math.}, 149(1--3):45--58, 1996.
	\newblock \doi{10.1016/0012-365X(94)00337-I}.
	\newblock \msn{1375097}.
	
	\bibitem[{Ding et~al.(1998)Ding, Oporowski, Sanders, and Vertigan}]{DOSV98}
	\textsc{Guoli Ding, Bogdan Oporowski, Daniel~P. Sanders, and Dirk Vertigan}.
	\newblock Partitioning graphs of bounded tree-width.
	\newblock \emph{Combinatorica}, 18(1):1--12, 1998.
	\newblock \doi{10.1007/s004930050001}.
	\newblock \msn{1645638}.
	
	\bibitem[{Ding et~al.(2000)Ding, Oporowski, Sanders, and
		Vertigan}]{DOSV-JCTB00}
	\textsc{Guoli Ding, Bogdan Oporowski, Daniel~P. Sanders, and Dirk Vertigan}.
	\newblock Surfaces, tree-width, clique-minors, and partitions.
	\newblock \emph{J. Combin. Theory Ser. B}, 79(2):221--246, 2000.
	\newblock \doi{10.1006/jctb.2000.1962}.
	\newblock \msn{1769192}.
	
	\bibitem[{D\k{e}bski et~al.(2020)D\k{e}bski, Felsner, Micek, and
		Schr\"{o}der}]{DFMS20}
	\textsc{Micha{\l} D\k{e}bski, Stefan Felsner, Piotr Micek, and Felix
		Schr\"{o}der}.
	\newblock Improved bounds for centered colorings.
	\newblock In \textsc{Shuchi Chawla}, ed., 
	\emph{Proc. ACM-SIAM Symposium on Discrete Algorithms} (SODA '20), 
	pp.~2212--2226, 2020.
	\newblock \doi{10.1137/1.9781611975994.136}.
	\newblock \arXiv{1907.04586}.
	
	\bibitem[{Dujmovi{\'c}(2015)}]{Duj15}
	\textsc{Vida Dujmovi{\'c}}.
	\newblock Graph layouts via layered separators.
	\newblock \emph{J. Combin. Theory Series B.}, 110:79--89, 2015.
	\newblock \doi{10.1016/j.jctb.2014.07.005}.
	\newblock \msn{3279388}.
	
	\bibitem[{Dujmovi\'c et~al.(2018)Dujmovi\'c, Eppstein, Joret, Morin, and
		Wood}]{DEJMW}
	\textsc{Vida Dujmovi\'c, David Eppstein, Gwena\"el Joret, Pat Morin, and
		David~R. Wood}.
	\newblock Minor-closed graph classes with bounded layered pathwidth.
	\newblock 2018.
	\newblock \arXiv{1810.08314}.
	
	\bibitem[{Dujmovi\'c et~al.(2017)Dujmovi\'c, Eppstein, and Wood}]{DEW17}
	\textsc{Vida Dujmovi\'c, David Eppstein, and David~R. Wood}.
	\newblock Structure of graphs with locally restricted crossings.
	\newblock \emph{SIAM J. Discrete Math.}, 31(2):805--824, 2017.
	\newblock \doi{10.1137/16M1062879}.
	\newblock \msn{3639571}.
	
	\bibitem[{Dujmovi\'c et~al.(2020)Dujmovi\'c, Esperet, Joret, Gavoille, Micek,
		and Morin}]{DEJGMM}
	\textsc{Vida Dujmovi\'c, Louis Esperet, Gwena\"el Joret, Cyril Gavoille, Piotr
		Micek, and Pat Morin}.
	\newblock Adjacency labelling for planar graphs (and beyond).
	\newblock 2020.
	\newblock \arXiv{2003.04280}.
	
	\bibitem[{Dujmovi{\'c} et~al.(2020{\natexlab{a}})Dujmovi{\'c}, Esperet, Joret,
		Walczak, and Wood}]{DEJWW20}
	\textsc{Vida Dujmovi{\'c}, Louis Esperet, Gwena\"{e}l Joret, Bartosz Walczak,
		and David~R. Wood}.
	\newblock Planar graphs have bounded nonrepetitive chromatic number.
	\newblock \emph{Advances in Combinatorics}, 5, 2020{\natexlab{a}}.
	\newblock \doi{10.19086/aic.12100}.
	
	\bibitem[{Dujmovi{\'c} et~al.(2020{\natexlab{b}})Dujmovi{\'c}, Esperet, Morin,
		Walczak, and Wood}]{DEMWW}
	\textsc{Vida Dujmovi{\'c}, Louis Esperet, Pat Morin, Bartosz Walczak, and
		David~R. Wood}.
	\newblock Clustered 3-colouring graphs of bounded degree.
	\newblock 2020{\natexlab{b}}.
	\newblock \arXiv{2002.11721}.
	
	\bibitem[{{Dujmovi\'c} and {Frati}(2018)}]{DF18}
	\textsc{{Vida} {Dujmovi\'c} and {Fabrizio} {Frati}}.
	\newblock Stack and queue layouts via layered separators.
	\newblock \emph{J. Graph Algorithms Appl.}, 22(1):89--99, 2018.
	\newblock \doi{10.7155/jgaa.00454}.
	\newblock \msn{3757347}.
	
	\bibitem[{Dujmovi{\'c} et~al.(2005)Dujmovi{\'c}, Morin, and Wood}]{DMW05}
	\textsc{Vida Dujmovi{\'c}, Pat Morin, and David~R. Wood}.
	\newblock Layout of graphs with bounded tree-width.
	\newblock \emph{SIAM J. Comput.}, 34(3):553--579, 2005.
	\newblock \doi{10.1137/S0097539702416141}.
	\newblock \msn{2137079}.
	
	\bibitem[{Dujmovi{\'c} et~al.(2017)Dujmovi{\'c}, Morin, and Wood}]{DMW17}
	\textsc{Vida Dujmovi{\'c}, Pat Morin, and David~R. Wood}.
	\newblock Layered separators in minor-closed graph classes with applications.
	\newblock \emph{J. Combin. Theory Ser. B}, 127:111--147, 2017.
	\newblock \doi{10.1016/j.jctb.2017.05.006}.
	\newblock \msn{3704658}.
	
	\bibitem[{Dujmovi{\'c} et~al.(2019)Dujmovi{\'c}, Morin, and Wood}]{DMW19}
	\textsc{Vida Dujmovi{\'c}, Pat Morin, and David~R. Wood}.
	\newblock Queue layouts of graphs with bounded degree and bounded genus, 2019.
	\newblock \arXiv{1901.05594}.
	
	\bibitem[{Dujmovi{\'c} et~al.(2019)Dujmovi{\'c}, Morin, and Wood}]{DMW19a}
	\textsc{Vida Dujmovi{\'c}, Pat Morin, and David~R. Wood}.
	\newblock Graph product structure for non-minor-closed classes, 2019.
	\newblock \arXiv{1907.05168}.

	\bibitem[{Dujmovi{\'c} et~al.(2004)Dujmovi{\'c}, P\'or, and Wood}]{DPW04}
	\textsc{Vida Dujmovi{\'c}, Attila P\'or, and David~R. Wood}.
	\newblock Track layouts of graphs.
	\newblock \emph{Discrete Math. Theor. Comput. Sci.}, 6(2):497--522, 2004.
	\newblock \urlprefix\url{http://dmtcs.episciences.org/315}.
	\newblock \msn{2180055}.
	
	\bibitem[{Dujmovi{\'c} and Wood(2004{\natexlab{a}})}]{DujWoo04}
	\textsc{Vida Dujmovi{\'c} and David~R. Wood}.
	\newblock On linear layouts of graphs.
	\newblock \emph{Discrete Math. Theor. Comput. Sci.}, 6(2):339--358,
	2004{\natexlab{a}}.
	\newblock \urlprefix\url{http://dmtcs.episciences.org/317}.
	\newblock \msn{2081479}.
	
	\bibitem[{Dujmovi{\'c} and Wood(2004{\natexlab{b}})}]{DujWoo-SubQuad-AMS}
	\textsc{Vida Dujmovi{\'c} and David~R. Wood}.
	\newblock Three-dimensional grid drawings with sub-quadratic volume.
	\newblock In \textsc{J\'{a}nos Pach}, ed., \emph{Towards a Theory of Geometric
		Graphs}, vol. 342 of \emph{Contemporary Mathematics}, pp. 55--66. Amer. Math.
	Soc., 2004{\natexlab{b}}.
	\newblock \msn{2065252}.
	
	\bibitem[{Dujmovi{\'c} and Wood(2005)}]{DujWoo05}
	\textsc{Vida Dujmovi{\'c} and David~R. Wood}.
	\newblock Stacks, queues and tracks: Layouts of graph subdivisions.
	\newblock \emph{Discrete Math. Theor. Comput. Sci.}, 7:155--202, 2005.
	\newblock \urlprefix\url{http://dmtcs.episciences.org/346}.
	\newblock \msn{2164064}.
	
	\bibitem[{Dujmovi{\'c} and Wood(2006)}]{DujWoo-Order06}
	\textsc{Vida Dujmovi{\'c} and David~R. Wood}.
	\newblock Upward three-dimensional grid drawings of graphs.
	\newblock \emph{Order}, 23(1):1--20, 2006.
	\newblock \doi{10.1007/s11083-006-9028-y}.
	\newblock \msn{2258457}.
	
	\bibitem[{Dvo{\v{r}}{\'a}k et~al.(2020)Dvo{\v{r}}{\'a}k, Huynh, Joret, and
		Liu}]{DHGLW}
	\textsc{Zden{\v{e}}k Dvo{\v{r}}{\'a}k, Tony Huynh, Gwena\"el Joret, Chun-Hung Liu, David R. Wood}.
	\newblock Notes on graph product structure theory.
	\newblock 2020.
	\newblock \arXiv{2001.08860}.
	
	\bibitem[{Dvo{\v{r}}{\'a}k and Thomas(2014)}]{DvoTho}
	\textsc{Zden{\v{e}}k Dvo{\v{r}}{\'a}k and Robin Thomas}.
	\newblock List-coloring apex-minor-free graphs.
	\newblock 2014.
	\newblock \arXiv{1401.1399}.
	
	\bibitem[{Eppstein(2000)}]{Eppstein-Algo00}
	\textsc{David Eppstein}.
	\newblock Diameter and treewidth in minor-closed graph families.
	\newblock \emph{Algorithmica}, 27(3--4):275--291, 2000.
	\newblock \doi{10.1007/s004530010020}.
	\newblock \msn{1759751}.
	
	\bibitem[{Erickson and Whittlesey(2005)}]{EW05}
	\textsc{Jeff Erickson and Kim Whittlesey}.
	\newblock Greedy optimal homotopy and homology generators.
	\newblock In \emph{Proceedings 16th {A}nnual {ACM}-{SIAM} {S}ymposium on
		{D}iscrete {A}lgorithms}, pp. 1038--1046. ACM, 2005.
	
	\bibitem[{Felsner et~al.(2002)Felsner, Liotta, and Wismath}]{FLW-GD01-ref}
	\textsc{Stefan Felsner, Giussepe Liotta, and Stephen~K. Wismath}.
	\newblock Straight-line drawings on restricted integer grids in two and three
	dimensions.
	\newblock In \textsc{Petra Mutzel, Michael J{\"u}nger, and Sebastian Leipert},
	eds., \emph{Proc. 9th International Symp. on Graph Drawing} (GD '01), vol.
	2265 of \emph{Lecture Notes in Comput. Sci.}, pp. 328--342. Springer, 2002.
	\doi{10.1007/3-540-45848-4\_26}.
		
	\bibitem[{Fomin et~al.(2012)Fomin, Lokshtanov, and Saurabh}]{FLS-SODA12}
	\textsc{Fedor~V. Fomin, Daniel Lokshtanov, and Saket Saurabh}.
	\newblock Bidimensionality and geometric graphs.
	\newblock In \emph{Proc. 23rd {A}nnual {ACM}-{SIAM} {S}ymposium on {D}iscrete
		{A}lgorithms}, pp. 1563--1575. 2012.
	\newblock \doi{10.1137/1.9781611973099.124}.
	\newblock \msn{3205314}.
	
	\bibitem[{Fox and Pach(2010)}]{FP10}
	\textsc{Jacob Fox and J{\'a}nos Pach}.
	\newblock A separator theorem for string graphs and its applications.
	\newblock \emph{Combin. Probab. Comput.}, 19(3):371--390, 2010.
	\newblock \doi{10.1017/S0963548309990459}.
	
	\bibitem[{Fox and Pach(2014)}]{FP14}
	\textsc{Jacob Fox and J{\'a}nos Pach}.
	\newblock Applications of a new separator theorem for string graphs.
	\newblock \emph{Combin. Probab. Comput.}, 23(1):66--74, 2014.
	\newblock \doi{10.1017/S0963548313000412}.
	
	\bibitem[{Harvey and Wood(2017)}]{HW17}
	\textsc{Daniel~J. Harvey and David~R. Wood}.
	\newblock Parameters tied to treewidth.
	\newblock \emph{J. Graph Theory}, 84(4):364--385, 2017.
	\newblock \doi{10.1002/jgt.22030}.
	\newblock \msn{3623383}.
	
	\bibitem[{Hasunuma(2007)}]{Hasunuma-DAM}
	\textsc{Toru Hasunuma}.
	\newblock Queue layouts of iterated line directed graphs.
	\newblock \emph{Discrete Appl. Math.}, 155(9):1141--1154, 2007.
	\newblock \doi{10.1016/j.dam.2006.04.045}.
	\newblock \msn{2321020}.
	
	\bibitem[{Heath et~al.(1992)Heath, Leighton, and Rosenberg}]{HLR92}
	\textsc{Lenwood~S. Heath, F.~Thomson Leighton, and Arnold~L. Rosenberg}.
	\newblock Comparing queues and stacks as mechanisms for laying out graphs.
	\newblock \emph{SIAM J. Discrete Math.}, 5(3):398--412, 1992.
	\newblock \doi{10.1137/0405031}.
	\newblock \msn{1172748}.
	
	\bibitem[{Heath and Pemmaraju(1997)}]{HP97}
	\textsc{Lenwood~S. Heath and Sriram~V. Pemmaraju}.
	\newblock Stack and queue layouts of posets.
	\newblock \emph{SIAM J. Discrete Math.}, 10(4):599--625, 1997.
	\newblock \doi{10.1137/S0895480193252380}. 
	
	\bibitem[{Heath and Rosenberg(1992)}]{HR92}
	\textsc{Lenwood~S. Heath and Arnold~L. Rosenberg}.
	\newblock Laying out graphs using queues.
	\newblock \emph{SIAM J. Comput.}, 21(5):927--958, 1992.
	\newblock \doi{10.1137/0221055}.
	\newblock \msn{1181408}.
	
	\bibitem[{Heath and Rosenberg(2011)}]{HR11}
	\textsc{Lenwood~S. Heath and Arnold~L. Rosenberg}.
	\newblock Graph layout using queues, 2011.
	\newblock
	\urlprefix\url{https://www.researchgate.net/publication/220616637_Laying_Out_Graphs_Using_Queues}.
	
	\bibitem[{Heawood(1890)}]{Heawood1890}
	\textsc{Percy~J. Heawood}.
	\newblock Map colour theorem.
	\newblock \emph{Quart. J. Pure Appl. Math.}, 24:332--338, 1890.
	\newblock \doi{10.1112/plms/s2-51.3.161}.
	
	\bibitem[{van~den Heuvel et~al.(2017)van~den Heuvel, Ossona~de Mendez, Quiroz, Rabinovich, and Siebertz}]{HOQRS17}
	\textsc{Jan van~den Heuvel, Patrice Ossona~de Mendez, Daniel Quiroz, Roman Rabinovich, and Sebastian Siebertz}.
	\newblock On the generalised colouring numbers of graphs that exclude a fixed minor.
	\newblock \emph{European J. Combin.}, 66:129--144, 2017.
	\newblock \doi{10.1016/j.ejc.2017.06.019}.

\bibitem[{van~den Heuvel and Wood(2017)}]{vdHW17}
\textsc{Jan van~den Heuvel and David~R. Wood}.
\newblock Improper colourings inspired by {H}adwiger's conjecture, 2017.
\newblock \arXiv{1704.06536}.

\bibitem[{van~den Heuvel and Wood(2018)}]{vdHW18}
\textsc{Jan van~den Heuvel and David~R. Wood}.
\newblock Improper colourings inspired by {H}adwiger's conjecture.
\newblock \emph{J. London Math. Soc.}, 98:129--148, 2018.
\newblock \doi{10.1112/jlms.12127}.

	\bibitem[Kannan et~al.(1988)Kannan, Naor, and
Rudich]{KNR88}
\textsc{Sampath Kannan, Moni Naor, and Steven Rudich}.
\newblock Implicit representation of graphs.
\newblock In \emph{Proc. 20th Annual {ACM} Symposium on Theory of Computing} (STOC 1988), pages 334--343. {ACM}, 1988.
\newblock \doi{10.1145/62212.62244}.

\bibitem[Kannan et~al.(1992)Kannan, Naor, and Rudich]{KNR92}
\textsc{Sampath Kannan, Moni Naor, and Steven Rudich}.
\newblock Implicit representation of graphs.
\newblock \emph{{SIAM} J. Discrete Math.}, 5\penalty0 (4):\penalty0 596--603,
1992.
\newblock \doi{10.1137/0405049}.

	\bibitem[{Kierstead and Yang(2003)}]{KY03}
	\textsc{Hal~A. Kierstead and Daqing Yang}.
	\newblock Orderings on graphs and game coloring number.
	\newblock \emph{Order}, 20(3):255--264, 2003.
	\newblock \doi{10.1023/B:ORDE.0000026489.93166.cb}.
	
	\bibitem[{Knauer et~al.(2018)Knauer, Micek, and Ueckerdt}]{KMU18}
	\textsc{Kolja Knauer, Piotr Micek, and Torsten Ueckerdt}.
	\newblock The queue-number of posets of bounded width or height.
	\newblock In \textsc{Therese~C. Biedl and Andreas Kerren}, eds., \emph{Proc. 26th International Symposium on Graph Drawing and Network Visualization} (GD '18), vol. 11282 of \emph{Lecture Notes in Comput. Sci.}, pp. 200--212. Springer, 2018.
	\newblock \doi{10.1007/978-3-030-04414-5_14}.
	
	\bibitem[{Kobourov et~al.(2017)Kobourov, Liotta, and Montecchiani}]{KLM17}
	\textsc{Stephen~G. Kobourov, Giuseppe Liotta, and Fabrizio Montecchiani}.
	\newblock An annotated bibliography on 1-planarity.
	\newblock \emph{Comput. Sci. Rev.}, 25:49--67, 2017.
	\newblock \doi{10.1016/j.cosrev.2017.06.002}.
	\newblock \msn{3697129}.
	
	\bibitem[{Kotlov(2001)}]{Kotlov01}
	\textsc{Andre{\u\i} Kotlov}.
	\newblock Minors and strong products.
	\newblock \emph{European J. Combin.}, 22(4):511--512, 2001.
	\newblock \doi{10.1006/eujc.2000.0428}.
	\newblock \msn{1829745}.
	
	\bibitem[{Kozawa et~al.(2014)Kozawa, Otachi, and Yamazaki}]{KOY14}
	\textsc{Kyohei Kozawa, Yota Otachi, and Koichi Yamazaki}.
	\newblock Lower bounds for treewidth of product graphs.
	\newblock \emph{Discrete Appl. Math.}, 162:251--258, 2014.
	\newblock \doi{10.1016/j.dam.2013.08.005}.
	\newblock \msn{3128527}.
	
	\bibitem[{Kratochv{\'{\i}}l(1991)}]{Krat-JCTB91}
	\textsc{Jan Kratochv{\'{\i}}l}.
	\newblock String graphs. {II}. {R}ecognizing string graphs is {NP}-hard.
	\newblock \emph{J. Combin. Theory Ser. B}, 52(1):67--78, 1991.
	\newblock \doi{10.1016/0095-8956(91)90091-W}.
	
	\bibitem[{Kratochv{\'{\i}}l and Vaner(2012)}]{KV12}
	\textsc{Jan Kratochv{\'{\i}}l and Michal Vaner}.
	\newblock A note on planar partial 3-trees, 2012.
	\newblock \arXiv{1210.8113}.
	
	\bibitem[{Liu and Wood(2019)}]{LW1}
	\textsc{Chun-Hung Liu and David~R. Wood}.
	\newblock Clustered graph coloring and layered treewidth.
	\newblock 2019.
	\newblock \arXiv{1905.08969}.
	
	\bibitem[{Malitz(1994)}]{Malitz94b}
	\textsc{Seth~M. Malitz}.
	\newblock Genus $g$ graphs have pagenumber ${O}(\sqrt g)$.
	\newblock \emph{J. Algorithms}, 17(1):85--109, 1994.
	\newblock \doi{10.1006/jagm.1994.1028}.
	\newblock \msn{1279270}.
	
	\bibitem[{Mohar(1999)}]{Mohar99}
	\textsc{Bojan Mohar}.
	\newblock A linear time algorithm for embedding graphs in an arbitrary surface.
	\newblock \emph{{SIAM} J. Discrete Math.}, 12(1):6--26, 1999.
	\newblock \doi{10.1137/S089548019529248X}.
	
	\bibitem[{Mohar and Thomassen(2001)}]{MoharThom}
	\textsc{Bojan Mohar and Carsten Thomassen}.
	\newblock \emph{Graphs on surfaces}.
	\newblock Johns Hopkins University Press, 2001.
	\newblock \msn{1844449}.
	
	\bibitem[{Morin(2020)}]{Morin20}
	\textsc{Pat Morin}.
	\newblock A fast algorithm for the product structure of planar graphs, 2020.
	\newblock \arXiv{2004.02530}.
	
	\bibitem[{Ne{\v{s}}et{\v{r}}il and Ossona~de Mendez(2012)}]{Sparsity}
	\textsc{Jaroslav Ne{\v{s}}et{\v{r}}il and Patrice Ossona~de Mendez}.
	\newblock \emph{Sparsity}, vol.~28 of \emph{Algorithms and Combinatorics}.
	\newblock Springer, 2012.
	\newblock \doi{10.1007/978-3-642-27875-4}.
	\newblock \msn{2920058}.
	
	\bibitem[{Ollmann(1973)}]{Ollmann73}
	\textsc{L.~Taylor Ollmann}.
	\newblock On the book thicknesses of various graphs.
	\newblock In \textsc{Frederick Hoffman, Roy~B. Levow, and Robert S.~D. Thomas},
	eds., \emph{Proc. 4th {S}outheastern {C}onference on {C}ombinatorics, {G}raph
		{T}heory and {C}omputing}, vol. VIII of \emph{Congr. Numer.}, p. 459.
	Utilitas Math., 1973.
	
	\bibitem[{Ossona~de Mendez et~al.(2019)Ossona~de Mendez, Oum, and Wood}]{OOW19}
	\textsc{Patrice Ossona~de Mendez, {Sang-il} Oum, and David~R. Wood}.
	\newblock Defective colouring of graphs excluding a subgraph or minor.
	\newblock \emph{Combinatorica}, 39(2):377--410, 2019.
	\newblock \doi{10.1007/s00493-018-3733-1}.
	
	\bibitem[{Pach et~al.(1999)Pach, Thiele, and T\'{o}th}]{PTT99}
	\textsc{J\'{a}nos Pach, Torsten Thiele, and G\'{e}za T\'{o}th}.
	\newblock Three-dimensional grid drawings of graphs.
	\newblock In \textsc{Bernard Chazelle, Jacob~E. Goodman, and Richard Pollack},
	eds., \emph{Advances in discrete and computational geometry}, vol. 223 of
	\emph{Contemporary Mathematics}, pp. 251--255. Amer. Math. Soc., 1999.
	\newblock \msn{1661387}.
	
	\bibitem[{Pach and T{\'o}th(2002)}]{PachToth-DCG02}
	\textsc{J{\'a}nos Pach and G{\'e}za T{\'o}th}.
	\newblock Recognizing string graphs is decidable.
	\newblock \emph{Discrete Comput. Geom.}, 28(4):593--606, 2002.
	\newblock \doi{10.1007/s00454-002-2891-4}.
	
	\bibitem[{Pemmaraju(1992)}]{Pemmaraju-PhD}
	\textsc{Sriram~V. Pemmaraju}.
	\newblock \emph{Exploring the Powers of Stacks and Queues via Graph Layouts}.
	\newblock Ph.D. thesis, Virginia Polytechnic Institute and State University,
	U.S.A., 1992.
	
	\bibitem[{Pilipczuk and Siebertz(2019)}]{PS18}
	\textsc{Micha{\l} Pilipczuk and Sebastian Siebertz}.
	\newblock Polynomial bounds for centered colorings on proper minor-closed graph
	classes.
	\newblock In \textsc{Timothy M. Chan}, ed., 
	\emph{Proc. 30th {A}nnual {ACM}-{SIAM} {S}ymposium on {D}iscrete
		{A}lgorithms}, pp. 1501--1520. 2019.
	\newblock \doi{10.1137/1.9781611975482.91}.
	\newblock \arXiv{1807.03683}.
	
	\bibitem[{Pupyrev(2019)}]{Pupyrev19}
	\textsc{Sergey Pupyrev}.
	\newblock Improved bounds for track numbers of planar graphs, 2019.
	\newblock \arXiv{1910.14153}.
	
	\bibitem[{Reed(1997)}]{Reed97}
	\textsc{Bruce~A. Reed}.
	\newblock Tree width and tangles: a new connectivity measure and some
	applications.
	\newblock In \emph{Surveys in combinatorics}, vol. 241 of \emph{London Math.
		Soc. Lecture Note Ser.}, pp. 87--162. Cambridge Univ. Press, 1997.
	\newblock \doi{10.1017/CBO9780511662119.006}.
	\newblock \msn{1477746}.
	
	\bibitem[{Reed and Seymour(1998)}]{ReedSeymour-JCTB98}
	\textsc{Bruce~A. Reed and Paul Seymour}.
	\newblock Fractional colouring and {H}adwiger's conjecture.
	\newblock \emph{J. Combin. Theory Ser. B}, 74(2):147--152, 1998.
	\newblock \doi{10.1006/jctb.1998.1835}.
	\newblock \msn{1654153}.
	
	\bibitem[{Rengarajan and Veni~Madhavan(1995)}]{RM-COCOON95}
	\textsc{S.~Rengarajan and C.~E. Veni~Madhavan}.
	\newblock Stack and queue number of $2$-trees.
	\newblock In \textsc{Ding-Zhu Du and Ming Li}, eds., \emph{Proc. 1st Annual
		International Conf. on Computing and Combinatorics} (COCOON '95), vol. 959 of
	\emph{Lecture Notes in Comput. Sci.}, pp. 203--212. Springer, 1995.
	\newblock \doi{10.1007/BFb0030834}.
	
	\bibitem[{Robertson and Seymour(1986)}]{RS-II}
	\textsc{Neil Robertson and Paul Seymour}.
	\newblock Graph minors. {II}. {A}lgorithmic aspects of tree-width.
	\newblock \emph{J. Algorithms}, 7(3):309--322, 1986.
	\newblock \doi{10.1016/0196-6774(86)90023-4}.
	\newblock \msn{0855559}.
	
	\bibitem[{Robertson and Seymour(2003)}]{RS-XVI}
	\textsc{Neil Robertson and Paul Seymour}.
	\newblock Graph minors. {XVI}. {E}xcluding a non-planar graph.
	\newblock \emph{J. Combin. Theory Ser. B}, 89(1):43--76, 2003.
	\newblock \doi{10.1016/S0095-8956(03)00042-X}.
	\newblock \msn{1999736}.
	
	\bibitem[{Schaefer et~al.(2003)Schaefer, Sedgwick, and
		{\v{S}}tefankovi{\v{c}}}]{SSS-JCSS03}
	\textsc{Marcus Schaefer, Eric Sedgwick, and Daniel {\v{S}}tefankovi{\v{c}}}.
	\newblock Recognizing string graphs in {NP}.
	\newblock \emph{J. Comput. System Sci.}, 67(2):365--380, 2003.
	\newblock \doi{10.1016/S0022-0000(03)00045-X}.
	
	\bibitem[{Schaefer and {\v{S}}tefankovi{\v{c}}(2004)}]{SS-JCSS04}
	\textsc{Marcus Schaefer and Daniel {\v{S}}tefankovi{\v{c}}}.
	\newblock Decidability of string graphs.
	\newblock \emph{J. Comput. System Sci.}, 68(2):319--334, 2004.
	\newblock \doi{10.1016/j.jcss.2003.07.002}.
	
	\bibitem[{Scott et~al.(2019)Scott, Seymour, and Wood}]{SSW19}
	\textsc{Alex Scott, Paul Seymour, and David~R. Wood}.
	\newblock Bad news for chordal partitions.
	\newblock \emph{J. Graph Theory}, 90:5--12, 2019.
	\newblock \doi{10.1002/jgt.22363}.
	
	\bibitem[{Seese(1985)}]{Seese85}
	\textsc{Detlef Seese}.
	\newblock Tree-partite graphs and the complexity of algorithms.
	\newblock In \textsc{Lothar Budach}, ed., \emph{Proc. International Conf. on
		Fundamentals of Computation Theory}, vol. 199 of \emph{Lecture Notes in
		Comput. Sci.}, pp. 412--421. Springer, 1985.
	\newblock \doi{10.1007/BFb0028825}.
	\newblock \msn{0821258}.
	
	\bibitem[{Shahrokhi(2013)}]{Shahrokhi13}
	\textsc{Farhad Shahrokhi}.
	\newblock New representation results for planar graphs.
	\newblock In \emph{29th European Workshop on Computational Geometry} (EuroCG
		2013), pp. 177--180. 2013.
	\newblock \arXiv{1502.06175}.
	
	\bibitem[{Wang(2017)}]{Wang17}
	\textsc{Jiun{-}Jie Wang}.
	\newblock Layouts for plane graphs on constant number of tracks, 2017.
	\newblock \arXiv{1708.02114}.
	
	\bibitem[{Wiechert(2017)}]{Wiechert17}
	\textsc{Veit Wiechert}.
	\newblock On the queue-number of graphs with bounded tree-width.
	\newblock \emph{Electron. J. Combin.}, 24(1):1.65, 2017.
	\newblock \urlprefix\url{https://doi.org/10.37236/6429}.
	\newblock \msn{3651947}.
	
	\bibitem[{Wood(2005)}]{Wood-Queue-DMTCS05}
	\textsc{David~R. Wood}.
	\newblock Queue layouts of graph products and powers.
	\newblock \emph{Discrete Math. Theor. Comput. Sci.}, 7(1):255--268, 2005.
	\newblock \urlprefix\url{http://dmtcs.episciences.org/352}.
	\newblock \msn{2183176}.
	
	\bibitem[{Wood(2008{\natexlab{a}})}]{Wood-QueueDegree}
	\textsc{David~R. Wood}.
	\newblock Bounded-degree graphs have arbitrarily large queue-number.
	\newblock \emph{Discrete Math. Theor. Comput. Sci.}, 10(1):27--34,
	2008{\natexlab{a}}.
	\newblock \urlprefix\url{http://dmtcs.episciences.org/434}.
	\newblock \msn{2369152}.
	
	\bibitem[{Wood(2008{\natexlab{b}})}]{WoodBanff08}
	\textsc{David~R. Wood}.
	\newblock The structure of Cartesian products, 2008{\natexlab{b}}.
	\newblock
	\urlprefix\url{https://www.birs.ca/workshops/2008/08w5079/report08w5079.pdf}.
	
	\bibitem[{Wood(2009)}]{Wood09}
	\textsc{David~R. Wood}.
	\newblock On tree-partition-width.
	\newblock \emph{European J. Combin.}, 30(5):1245--1253, 2009.
	\newblock \doi{10.1016/j.ejc.2008.11.010}.
	\newblock \msn{2514645}.
	
	\bibitem[{Wood(2011)}]{Wood-ProductMinor}
	\textsc{David~R. Wood}.
	\newblock Clique minors in cartesian products of graphs.
	\newblock \emph{New York J. Math.}, 17:627--682, 2011.
	\newblock \urlprefix\url{http://nyjm.albany.edu/j/2011/17-28.html}.
	
	\bibitem[{Wood(2013)}]{Wood12}
	\textsc{David~R. Wood}.
	\newblock Treewidth of {C}artesian products of highly connected graphs.
	\newblock \emph{J. Graph Theory}, 73(3):318--321, 2013.
	\newblock \doi{10.1002/jgt.21677}.
	
	\bibitem[{Wu et~al.(2010)Wu, Yang, and Yu}]{QYY10}
	\textsc{Zefang Wu, Xu~Yang, and Qinglin Yu}.
	\newblock A note on graph minors and strong products.
	\newblock \emph{Appl. Math. Lett.}, 23(10):1179--1182, 2010.
	\newblock \doi{10.1016/j.aml.2010.05.007}.
	\newblock \msn{2665591}.
	
	\bibitem[{Yannakakis(1989)}]{Yannakakis89}
	\textsc{Mihalis Yannakakis}.
	\newblock Embedding planar graphs in four pages.
	\newblock \emph{J. Comput. System Sci.}, 38(1):36--67, 1989.
	\newblock \doi{10.1016/0022-0000(89)90032-9}.
	\newblock \msn{0990049}.
	
\end{thebibliography}
% PLEASE DO NOT TOUCH ANYTHING FROM HERE ON
% ANY UPDATES WILL BE DONE BY BIBTEX

\def\soft#1{\leavevmode\setbox0=\hbox{h}\dimen7=\ht0\advance \dimen7
	by-1ex\relax\if t#1\relax\rlap{\raise.6\dimen7
		\hbox{\kern.3ex\char'47}}#1\relax\else\if T#1\relax
	\rlap{\raise.5\dimen7\hbox{\kern1.3ex\char'47}}#1\relax \else\if
	d#1\relax\rlap{\raise.5\dimen7\hbox{\kern.9ex \char'47}}#1\relax\else\if
	D#1\relax\rlap{\raise.5\dimen7 \hbox{\kern1.4ex\char'47}}#1\relax\else\if
	l#1\relax \rlap{\raise.5\dimen7\hbox{\kern.4ex\char'47}}#1\relax \else\if
	L#1\relax\rlap{\raise.5\dimen7\hbox{\kern.7ex
			\char'47}}#1\relax\else\message{accent \string\soft \space #1 not
		defined!}#1\relax\fi\fi\fi\fi\fi\fi}

\end{document}